%% file: main.tex
\documentclass[acmsmall,nonacm]{acmart}
\usepackage{amsthm, amsmath, mathtools}
\usepackage{upgreek}
\usepackage{amsfonts}
\usepackage{enumitem}
\usepackage[normalem]{ulem}
\usepackage[linesnumbered,vlined,ruled,commentsnumbered]{algorithm2e}
\usepackage[para,online,flushleft]{threeparttable}
\usepackage{multirow}
\usepackage{multicol}
\usepackage{hyperref}
\usepackage{url}
\usepackage{bm}
\usepackage[labelfont=bf,skip=0pt]{caption,subfig}
\usepackage{flushend}

\usepackage{tabularx}

\usepackage[english]{babel}
\usepackage{graphicx}
\usepackage{thmtools}

\usepackage{listings}
\usepackage{tablefootnote}
\usepackage{enumitem}
\usepackage{bm}
\usepackage{comment}
\usepackage{tikz}
\usepackage{wrapfig}

\DeclareMathOperator*{\argmax}{arg\,max}
\DeclareMathOperator*{\argmin}{arg\,min}

\newcommand{\opt}{\textsf{opt}}

\newcommand{\Q}{\boldsymbol{q}}
\newcommand{\E}{\mathcal{E}}

\newcommand{\R}{\mathcal{R}}
\newcommand{\T}{\mathcal{T}}

\newcommand{\eps}{\varepsilon}
\renewcommand{\epsilon}{\varepsilon}
\renewcommand{\Re}{\mathbb{R}}

\newcommand{\allattr}{{\mathbf A}}

\newcommand{\allrel}{{\mathbf R}}

\newcommand{\dom}{{\texttt{dom}}}

\newcommand{\fhw}{\mathsf{fhw}}

\newcommand{\D}{\mathbf{D}}
\newcommand{\I}{\D}
\newcommand{\net}{\mathcal{C}}

\newcommand{\dist}{\phi}

\newcommand{\kcenter}{\textbf{u}}
\newcommand{\kmedian}{\textbf{v}}
\newcommand{\kmeans}{\mu}
\renewcommand{\opt}{\mathsf{OPT}}

\newcommand{\coreset}{\mathcal{C}}
\newcommand{\ret}{\mathcal{S}}

\newcommand{\tree}{\mathcal{T}}
\newcommand{\node}{u}

\newcommand{\timeCenter}{T^{\mathsf{cen}}}
\newcommand{\timeMedian}{T^{\mathsf{med}}}
\newcommand{\timeMeans}{T^{\mathsf{mean}}}

\renewcommand{\O}{\tilde{O}}

\newcommand{\DkcenterAlg}{\mathsf{kCenterAlg}}

\newcommand{\DkmedianAlg}{\mathsf{kMedianAlg}}

\newcommand{\DkmeansAlg}{\mathsf{kMeansAlg}}

\newcommand{\numqueries}{\sigma}

\renewcommand{\root}{\mathtt{root}}
\newcommand{\ball}{\mathcal{B}}
\newcommand{\canonical}{\mathcal{U}}
\newcommand{\counts}{c}
\newcommand{\cons}{\zeta}
\newcommand{\assigfunc}{\iota}

\newcommand{\revone}[1]{#1}
\newenvironment{envrevone}{\color{black}}{}
\newcommand{\revtwo}[1]{#1}
\newcommand{\revthree}[1]{#1}
\newenvironment{envrevthree}{\color{black}}{}

\newcommand{\badpoints}{Q_i^{\mathsf{bad}}}

\allowdisplaybreaks
\begin{document}

\title{Faster Relational Algorithms Using Geometric Data Structures} 
\titlenote{This work was partially supported by NSF grant IIS-2348919.}

\author{Aryan Esmailpour}
\orcid{0009-0000-3798-9578}
\affiliation{%
  \institution{Department of Computer Science, University of Illinois Chicago}
  \city{Chicago}
  \country{USA}}
\email{aesmai2@uic.edu}

\author{Stavros Sintos}
\orcid{0000-0002-2114-8886}
\affiliation{%
  \institution{Department of Computer Science, University of Illinois Chicago}
  \city{Chicago}
  \country{USA}}
\email{stavros@uic.edu}

\begin{comment}
\setcopyright{cc}
\setcctype{by --or -- by-nc-nd} 
% See YOUR completed rightsreview form confirmation email for the exact code to be used for setcopyright and setcctype code strings, these are unique for every paper.
\acmJournal{PACMMOD}
\acmYear{2026} \acmVolume{4} 
\acmNumber{2} \acmArticle{102} 
% XXX is your article id# -- be sure to use exact code from ACM rightsreview, do not use the v4 prefix) 
\acmMonth{5} \acmPrice{}
\acmDOI{10.1145/XXXXXXX}
% See YOUR completed rightsreview form confirmation email for the assigned DOI to your article -- this is where the article will live when published in the ACM DL.
\end{comment}

\input{abstract}

%\settopmatter{printacmref=false}

%\pagestyle{plain}

%\input{abstract}




%\received{May 2024}
% \received[revised]{JuneX 2024}
%\received[accepted]{August 2024}
%\input{letter}

\maketitle
\input{intro}
\input{treeConstruction}
\input{kcenter}

\input{kMedianMeans}
\input{extensions}
\input{conclusion}

%%
%% Bibliography
%%

%% Please use bibtex, 

\bibliography{ref}
%\newpage
\input{appendix}

\end{document}

%% file: abstract.tex
\begin{abstract}
Optimization tasks over relational data, such as clustering, often suffer from the prohibitive cost of join operations, which are necessary to access the full dataset. While geometric data structures like BBD trees yield fast approximation algorithms in the standard computational setting, their application to relational data remains unclear due to the size of the join output. In this paper, we introduce a framework that leverages geometric insights to design faster
algorithms when the data is stored as the results of a join query in a relational database.
Our core contribution is the development of the RBBD tree, a randomized variant of the BBD tree tailored for relational settings. Instead of completely constructing the RBBD tree, by leveraging efficient sampling and counting techniques over relational joins, we enable on-the-fly efficient expansion of the RBBD tree, maintaining only the necessary parts.
This allows us to simulate geometric query procedures without materializing the join result. As an application, we present algorithms that improve the state-of-the-art for relational $k$-center/means/median clustering by a factor of $k$ in running time while maintaining the same approximation guarantees. Our method is general and can be applied to various optimization problems
in the relational setting.
%This work opens a new direction for applying geometric methods to accelerate relational algorithms.
\end{abstract}

%% file: intro.tex
%\vspace{-1.2em}
\section{Introduction}
\label{sec:intro}
\vspace{-0.4em}

Optimization problems, such as clustering, are traditionally studied in computer science under the assumption of full access to the input data. However, in real-world applications, data is typically stored in \emph{relational databases}, where information is distributed across multiple interrelated tables. Each table contains a set of tuples, and tuples from different tables may join if they have the same value on the shared attributes.
This relational structure introduces a fundamental challenge: full access to the data is only possible after performing a join operation across the relevant tables. Such joins are computationally expensive, as the output size of a join can be polynomially larger than the size of the input tables~\cite{ngo2018worst, ngo2014skew}. Nonetheless, relational data is ubiquitous in practice. According to~\cite{link1, kaggle, link:relational}, the majority of modern database systems and data science tasks involve relational data.
%are relational . Furthermore, surveys by Kaggle~\cite{kaggle} show that 65\% of data science tasks involve relational data, and 70\% of production databases use a relational model. The market for relational data processing is projected to exceed \$122 billion by 2027~\cite{link:relational}. 

\begin{envrevone}
Machine learning applications over relational data typically involve two stages: (1) data preparation, where multiple tables are joined to construct the feature dataset, and (2) modeling, where machine learning algorithms (for example, clustering algorithms in unsupervised learning) are applied to the resulting data. However, this two-step pipeline can be prohibitively expensive in practice, largely due to the high cost of evaluating large joins.

$\bullet$ Consider a machine-learning workflow in which feature engineering over relational data is followed by unsupervised learning via clustering, based on the Yelp review dataset~\cite{yelp2017}. This dataset has been used for $k$-means clustering in prior work~\cite{curtin2020rk, chen2022coresets}. 
The database consists of five relations: 
(1) \textsf{Review(user\_id, business\_id, stars, date)} describing individual reviews, 
(2) \textsf{User(user\_id, review\_count, fans, yelping\_since)} providing reviewer attributes, 
(3) \textsf{Business(business\_id, city, state, avg\_rating)} describing business metadata, 
(4) \textsf{Category(business\_id, category)} listing the categories assigned to each business (e.g., ``Restaurant'', ``Coffee Shop''), and 
(5) \textsf{Attributes(business\_id, attribute)} recording additional business characteristics 
(e.g., ``OpenLate'', ``HasWiFi''). 
A typical feature-engineering query joins these relations as follows:
$Q = \textsf{Review} \Join \textsf{User} \Join \textsf{Business} \Join \textsf{Category} \Join \textsf{Attributes}$. 
Each tuple in $Q$ corresponds to a single review enriched with both numerical attributes (e.g., \textsf{review\_count}, \textsf{avg\_rating}) and categorical attributes (e.g., \textsf{category}, \textsf{attribute}), which are one-hot encoded before analysis. 
Clustering is then performed directly over these joined tuples, meaning that individual review events are grouped according to their joint user attributes, business characteristics, and semantic labels. 
In this representation, two reviews are considered similar if they are written by users with similar activity profiles for businesses with similar ratings, categories, and attributes. 
As a result, each cluster captures a coherent pattern of reviewer behavior and business semantics, such as highly active users reviewing late-night restaurants with WiFi, or infrequent users reviewing family-oriented dining venues. 
Crucially, because businesses often have multiple category and attribute labels, the join result $Q$ is much larger than any base relation. 
As reported in~\cite{curtin2020rk}, the base Yelp tables contain roughly $8$ million tuples in total, whereas the fully materialized join contains approximately $22$ million tuples. %Moreover, another join query over the Yelp dataset used in~\cite{olteanu2022givens} contains roughly $2$ million base tuples, yet the corresponding join result contains roughly $150$ million tuples. 
This significant blow-up highlights the computational cost of materializing relational joins before learning 
and motivates the need for algorithms that can operate directly on the underlying tables without explicitly constructing $Q$.

$\bullet$ Consider a table $\mathsf{Purchase}$ with attributes $\mathsf{price}$, $\mathsf{T}$, $\mathsf{item}$, and $\mathsf{category}$, where $\mathsf{price}$ is a real-valued attribute, $\mathsf{T}$ identifies a transaction, $\mathsf{item}$ is a product identifier, and $\mathsf{category}$ is a categorical attribute. The goal is to analyze pricing patterns of items that are commonly purchased together. For instance, a $k$-means clustering over the result of the self-join $\mathsf{Purchase}(\mathsf{T},\mathsf{price}_0,\mathsf{item}_0,\mathsf{category}_0) \Join \mathsf{Purchase}(\mathsf{T},\mathsf{price}_1,\mathsf{item}_1,\mathsf{category}_1)$ groups co-purchased item pairs according to their joint price characteristics. Each joined tuple represents a pair of items bought in the same transaction, embedded in a two-dimensional price space $(\mathsf{price}_0,\mathsf{price}_1)$, possibly enriched with categorical information (for example the categories of the corresponding items). Because a single purchase may contain several different items, each transaction generates many distinct item pairs in the self-join, which can lead to a substantial blow-up in the size of the join result (quadratic with respect to the size of the $\mathsf{Purchase}$ table). Even if one subsequently projects the join result onto the price attributes alone, it is crucial to retain duplicate tuples: repeated occurrences of the same price pair correspond to item pairs that are purchased together frequently and therefore carry greater statistical weight in the clustering objective.  Clustering in this space can reveal typical price-based co-purchasing patterns, such as low--low, low--high, or high--high price item pairs. Such analyses of commonly bought products are standard in market-basket analysis and retail data mining~\cite{agrawal1993mining, han2000mining, leskovec2014mining}.

\end{envrevone}

Recent work has aimed to circumvent materializing the full join output by designing \emph{relational algorithms} for optimization tasks such as clustering~\cite{chen2022coresets, curtin2020rk, moseley2021relational}. However, these approaches often suffer from large constant approximation ratios or expensive runtimes. Moreover, each known method typically relies on a different specialized technique tailored to a specific problem.

In this paper, we introduce a general framework that enables faster algorithms for a broad class of optimization problems on relational data. Given a join query and a database instance, our approach can be used to efficiently solve problems such as clustering, diversity maximization,  and fairness-aware selection.
Many of these problems have well-understood and near-linear time approximation algorithms in the standard computational setting, where the input is a single flat table. However, naively applying these techniques in the relational setting is inefficient due to the size of the join output. Our key idea is to simulate the behavior of geometric data structures---particularly those based on hierarchical tree structures such as the \emph{BBD tree}~\cite{arya1998optimal}, \emph{without} materializing the entire join result. We present our approach by implementing a relational version of the BBD tree, since it best serves the clustering problems
%(along with other optimization problems)
we aim to solve. However, the same techniques can be adopted to derive a relational version of other commonly used geometric data structures such as quad trees~\cite{finkel1974quad}, kd-trees~\cite{deBerg2008}, or partition trees~\cite{chan2010optimal}.
We introduce a new randomized data structure, the \emph{RBBD tree}, which maintains the desirable properties of the BBD tree but is tailored for relational data. Crucially, we show that given only one node of an RBBD tree \revtwo{(defined over the join results of an acyclic join query)}, its children can be computed in roughly $O(N)$ time, where $N$ is the size of the input database. \revthree{This enables efficient simulation of geometric queries over the join results—including range counting, range reporting, and range sampling queries—as these typically require accessing only a logarithmic number of tree nodes}.

Our method is general: any geometric algorithm that interacts with the input data through $\numqueries$ queries to a BBD tree can be transformed into a randomized relational algorithm \revtwo{(over an acyclic join query)} with total runtime $O(\numqueries \cdot N)$ (ignoring logarithmic factors). As many geometric algorithms fall into this category, our framework yields efficient relational counterparts for a variety of tasks. 
As an application, we show that the running time of the state-of-the-art algorithms for all relational $k$-center, $k$-median, and $k$-means clustering can be improved by a factor of $k$.

\revtwo{While in many applications $k$ is considered small, $k$-clustering for large $k$ has become increasingly relevant in modern applications, such as product quantization for nearest neighbor search in vector databases~\cite{jegou2010product}. This has motivated extensive algorithmic studies of large-$k$-clustering across various computational models~\cite{ene2011fast, bateni2021extreme, coy2023parallel, czumaj2024fully, filtser2025faster}.
Furthermore, in the relational setting $k$ can be even larger than $N$, because the number of join results might be orders of magnitude larger than $N$.
%Given this large output size, a data analyst might wish to extract a representative subset of size $k$. However, choosing a small $k$ (e.g., a constant) can lead to misleading conclusions. A more practical approach is to construct a representative set of size $N$ and store it as a new table in the database that can be processed separately for visualization or further analysis.
In the literature on relational clustering~\cite{curtin2020rk, moseley2021relational, esmailpour2024improved, agarwal2024computing}, existing algorithms typically assume that $k$ may be large, and successive work has focused on improving the dependency on $k$. We present the first relational clustering algorithms with linear dependency on both $k$ and $N$ for acyclic queries.
}
%For example, we provide improved relational algorithms for $k$-center, $k$-median, and $k$-means clustering, achieving state-of-the-art performance in both runtime and approximation quality.

\vspace{-0.8em}
\subsection{Notation and problem definition}
\vspace{-0.2em}
\noindent{\bf Join Queries.}
We are given a database schema $\allrel$ over a set of $d$ attributes $\allattr=\{A_1,\ldots, A_d\}$. The database schema $\allrel$ contains $m$ relations $R_1, \ldots, R_m$. Let $\allattr_j\subseteq \allattr$ be the set of attributes associated with relation $R_j\in \allrel$. For an attribute $A\in \allattr$, let $\dom(A)$ be the domain of attribute $A$.
Let $\dom(X)=\prod_{A\in X}\dom(A)$ be be the domain of a set of attributes $X\subseteq \allattr$.
%We assume that $\dom(A)=\Re$ for every $A\in \allattr$.
Let $\I$ be a database instance over the database schema $\allrel$. For simplicity, we assume that each relation $R_j \in \allrel$ contains $N$ tuples in $\I$.
Throughout the paper, we consider data complexity i.e., $m$ and $d$ are constants, while $N$ is a large integer.
We use $R_j$ to denote both the relation and the set of tuples from $\I$ stored in the relation.
For a subset of attributes $B\subseteq \allattr$ and a set $Y\subset\Re^d$, let $\pi_{B}(Y)$ be the set containing the projection of the tuples in $Y$ onto the attributes $B$. Notice that two different tuples in $Y$ might have the same projection on $B$, however, $\pi_{B}(Y)$ is defined as a set, so the projected tuple is stored once. %We also define the multi-set $\widebar{\pi}_{B}(Y)$ so that if for two tuples $t_1, t_2\in Y$ it holds that $\pi_{B}(t_1)=\pi_B(t_2)$, then the tuple $\pi_B(t_1)$ exists more than once in $\widebar{\pi}_B(Y)$.

Following the related work on relational clustering~\cite{chen2022coresets, curtin2020rk, moseley2021relational}, we are given a join query $\Q:=R_1\Join \ldots \Join R_m$. The set of results of a join query $\Q$ over the database instance $\I$ is defined as $\Q(\I)=\{t\in \dom(\allattr)\mid \forall j\in[1,m]:\pi_{\allattr_j}(t)\in R_j\}$. A join query $\Q$ is acyclic if there exists a tree, called a join tree, such that the nodes of the tree are the relations in $\allrel$ and for every attribute $A\in \allattr$, the set of nodes/relations that contain $A$ form a connected component.
All our algorithms are presented assuming that $\Q$ is an acyclic join query, however, in the end we extend the results to any general join query.
Let $\rho^\star(\Q)$ be the fractional edge cover of query $\Q$, which is a parameter that bounds the number of join results $\Q(\I)$ over any database instance: for every database instance $\I'$ with $O(N)$ tuples, it holds that 
$|\Q(\I')|=O(N^{\rho^\star(\Q)})=O(N^m)$ as shown in~\cite{atserias2013size}.

We use the notation $\fhw(\Q)$ to denote the \emph{fractional hypertree width}~\cite{gottlob2014treewidth} of the query $\Q$.
The fractional hypertree width roughly measures how close $\Q$ is to being acyclic.
For every acyclic join query $\Q$, we have $\fhw(\Q)=1$.
Given a cyclic join query $\Q$, we convert it to an equivalent acyclic query such that each relation is the result of a (possibly cyclic) join query with fractional edge cover at most $\fhw(\Q)$. Hence, a cyclic join query over a database instance with $O(N)$ tuples per relation can be converted, in $O(N^{\fhw(\Q)})$ time, to an equivalent acyclic join query over a database instance with $O(N^{\fhw(\Q)})$ tuples per relation~\cite{atserias2013size}.
A more formal definition of $\fhw$ is given in Appendix~\ref{appndx:generalQueries}. If $\Q$ is clear from the context, we write $\fhw$ instead of $\fhw(\Q)$.

\revthree{
For tuples $p,q \in \dom(\allattr)$, let $\dist(p,q) = ||p - q||_2 = \left( \sum_{j=1}^{d} (\pi_{A_j}(p) - \pi_{A_j}(q))^2 \right)^{1/2}$ denote their Euclidean distance.
If $A_j \in \allattr$ is a categorical attribute, we apply the \emph{one-hot encoding} (see Appendix~\ref{sec:cat}) and define
$(\pi_{A_j}(p) - \pi_{A_j}(q))^2 = 1$ if $\pi_{A_j}(p) \neq \pi_{A_j}(q)$, and $(\pi_{A_j}(p) - \pi_{A_j}(q))^2 = 0$ otherwise.
We focus on Euclidean distance for clarity and consistency with prior work, but our algorithms extend to any $\ell_p$ metric.
}
%Throughout the paper, we use the Euclidean distance to measure the cost of the clustering.
%Following existing work, we define all our problems in the Euclidean space, however our results can be extended to other $\ell_p$ metrics or doubling spaces.
%, following the other related papers on relational clustering~\cite{curtin2020rk, chen2022coresets, moseley2021relational}.

Throughout the paper, we use the term \emph{standard computational setting} when we consider the input data to be stored in a single flat table. In contrast, we use the term \emph{relational setting} when the data is the result of a join query as described above.

\paragraph{Clustering}
%In this paper, we focus on $k$-center, $k$-median, and $k$-means clustering. 
\revthree{In this paper, we focus on the three classical clustering objectives: 
$k$-center, $k$-median, and $k$-means clustering. 
These formulations capture different ways of measuring the quality of clustering and are among the most widely used models in both theory and practice~\cite{lloyd1982least, gonzalez1985clustering, arya2004local, har2004coresets, ostrovsky2012effectiveness}. 
%Intuitively, the choice between them depends on how we wish to measure the “fit” between tuples and their assigned centers: 
$k$-center focuses on minimizing the worst-case distance, 
$k$-median minimizes the total (average) distance, 
and $k$-means minimizes the total squared distance. 
In applications where we want to ensure that no point is too far from its nearest center (e.g., facility location or coverage problems), $k$-center is preferred. 
When we aim to minimize average travel cost or total assignment distance (e.g., clustering customers by proximity), $k$-median is more suitable. 
Finally, $k$-means is typically used when distances represent errors or variances, and we care about minimizing overall dispersion (e.g., in statistical data analysis or machine learning applications).
%Other clustering models exist (e.g., density-based or hierarchical clustering), but these three serve as canonical formulations for studying algorithmic properties and approximation guarantees.
}

We start with some useful general definitions.
Let $P$ be a set of tuples in $\dom(\allattr)$ and let $C\subset \dom(\allattr)$
be a set of $k$ tuples.
%\footnote{We slightly abuse the notation and we use $d$ as both the number of attributes in the relational setting and the number of dimensions of a point set in the standard computational setting.}. 
%Let $w:\Re^d\rightarrow \Re_{>0}$ be a weight function such that $w(p)$ is the weight of point $p\in P$.
For a tuple $p\in\dom(\allattr)$, let $\dist(p,C)=\min_{c\in C}\dist(p,c)$. We define
\vspace{-0.0em}
$$\kcenter_C(P)=\max_{p\in P}\dist(p,C),\quad \kmedian_{C}(P)=\sum_{p\in P}\dist(p,C),\quad \text{ and } \quad \kmeans_{C}(P)=\sum_{p\in P}\dist^2(p,C).$$
\vspace{-0.6em}
%Without loss of generality, in this paper it is sufficient to assume that $\sum_{p\in P}w(p)=O(\poly(n))$, where $\poly(n)$ is a polynomial function of $n$.
%If $P$ is an unweighted set, then $w(p)=1$ for every $p\in P$.
%By slightly abusing the notation, in the unweighted case where $w(p)=1$ for every $p\in P$, we define $\kmedian_C(P)=\kmedian_{C,w}(P)$ and $\kmeans_C(P)=\kmeans_{C,w}(P)$.

\hspace{-1em}\textbf{$k$-center clustering}: Given a set of tuples $P$ in $\dom(\allattr)$ and a parameter $k$, the goal is to find a set $C\subseteq P$ with $|C|=k$ such that $\kcenter_{C}(P)$ is minimized. 
Let $\opt(P)=\argmin_{S\subseteq P, |S|= k}\kcenter_S(P)$ be a set of $k$ centers with the minimum $\kcenter_{\opt(P)}(P)$.
Let 
$\DkcenterAlg_\gamma$ be a (known) $\gamma$-approximation algorithm for the $k$-center clustering problem in the standard computational setting that runs in $\timeCenter_\gamma(|P|)$ time, where $\gamma$ is a constant.

\hspace{-1em}\textbf{$k$-median clustering}: Given a set of tuples $P$ in $\dom(\allattr)$ and a parameter $k$, the goal is to find a set $C\subseteq P$ with $|C|=k$ such that $\kmedian_{C}(P)$ is minimized.
%Let $C^*$ be the set of optimum centers for the $k$-median problem and let $C^*_d$ be the set of optimum centers for the discrete $k$-median problem.
Let $\opt(P)=\argmin_{S\subseteq P, |S|= k}\kmedian_S(P)$ be a set of $k$ centers with the minimum $\kmedian_{\opt(P)}(P)$.
Let $\DkmedianAlg_\gamma$  be a (known) $\gamma$-approximation algorithm for the  $k$-median problem in the standard computational setting that runs in $\timeMedian_\gamma(|P|)$ time, where $\gamma$ is a constant.

\hspace{-1em}\textbf{$k$-means clustering}: Given a set of tuples $P$ in $\dom(\allattr)$ and a parameter $k$, the goal is to find a set $C\subseteq P$ with $|C|=k$ such that $\kmeans_{C}(P)$ is minimized.
Let $\opt(P)=\argmin_{S\subseteq P, |S|= k}\kmeans_S(P)$ be a set of $k$ centers with the minimum $\kmeans_{\opt(P)}(P)$. 
%There are many known algorithms for the (weighted) $k$-means clustering problem that return a constant approximation. For example, the algorithm in~\cite{har2004coresets} returns an $O(1)$-approximation in $O(n+k^5\log^9 n)$ time.
Let $\DkmeansAlg_\gamma$ be a (known) $\gamma$-approximation algorithm for the  $k$-means problem in the standard computational setting that runs in $\timeMeans_\gamma(|P|)$ time, where $\gamma$ is a constant.
%We note that we use the same notation $\timeMeans_\gamma(\cdot)$ for the running time of a known algorithm in both $k$-median and $k$-means clustering. It is always clear from the context whether we consider the $k$-median or the $k$-means problem.

%We note that we do not specify whether $\kmedianAlg_\gamma$ (resp. $\kmeansAlg_\gamma$) is an algorithm for the $k$-median (resp. $k$-means) clustering or the discrete $k$-median (resp. relational $k$-means) clustering. In the next section we always specify which of the two versions we consider depending on which version of the problem we study.

We use the same notation $\opt(\cdot)$ for the optimum solution for $k$-center/median/means clustering. It is always clear from the context whether we are referring to $k$-center, $k$-median, or $k$-means clustering.
Sometimes we call  the value of $\kcenter_C(P)$ (or $\kmedian_C(P), \kmeans_C(P)$) as the clustering cost of $C$ in $P$.

%By relational $k$-median or $k$-means clustering, we are referring to the geometric versions, unless explicitly stated otherwise for the discrete variants. We propose results for both versions. We note that the clustering problem on relational data is defined over the unweighted set $\Q(\I)$. In order to propose efficient algorithms we will need to construct weighted subset of tuples (\emph{coresets}) that are used to derive small approximation factors for the relational clustering problems.

\paragraph{Relative approximation}
We call a set $C\subset \dom(\allattr)$ an $\alpha$-approximation for the $k$-center (resp. $k$-median, $k$-means) clustering over a tuple set $P\subset \dom(\allattr)$, if $\kcenter_C(P)\leq \alpha \cdot\kcenter_{\opt(P)}(P)$ (resp.  $\kmedian_C(P)\leq \alpha\cdot \kmedian_{\opt(P)}(P)$, $\kmeans_C(P)\leq \alpha\cdot  \kmeans_{\opt(P)}(P)$). An $\alpha$-approximation algorithm on $P$ always returns a set of centers which is $\alpha$-approximation over $P$.

\paragraph{Coreset}
A set $\coreset\subseteq \dom(\allattr)$, with $|\coreset|\ll |P|$, is a coreset for the $k$-center/median/means clustering over $P$, if every $\alpha$-approximation algorithm executed on $C$ returns a set of centers which is a $(1+\eps)\alpha$-approximation for $P$. Throughout the paper, we assume that $\eps\in(0,1)$.

%Throughout the paper we assume that $\eps\in(0,1)$ specified by the user.

In this paper, we study $k$-center, $k$-median and $k$-means clustering on the results of a join query.

\vspace{-0.5em}
\begin{definition}
\label{def:kmedian}
    [Relational $k$-center/median/means clustering] Given a database instance $\I$, a join query $\Q$, and a positive integer parameter $k$, the goal is to compute a set $\ret\subseteq \Q(\I)$ of size $|\ret|\leq k$ such that $\kcenter_\ret(\Q(\I))$ (resp. $\kmedian_\ret(\Q(\I))$, $\kmeans_\ret(\Q(\I))$ is minimized.
    %In the discrete relational $k$-median clustering the set $\ret$ should be a subset of $\Q(\I)$.
\end{definition}
%In the relational setting, we note that $1\leq k\leq |\Q(\I)|=O(N^m)$.
\vspace{-0.5em}
\paragraph{\revone{Remark}}
\revone{
While we focus on $k$-clustering over full join results, all our algorithms extend straightforwardly to $k$-clustering over the results of project--join queries under \emph{bag semantics}, where duplicates are preserved (see also Appendix~\ref{sec:cat}). This setting is particularly useful in applications where the goal is to cluster all tuples in the join result only with respect to a subset of attributes 
(for example, using only the price attribute in the earlier $\mathsf{Purchase}$ table).
As discussed in Appendix~\ref{sec:cat}, the ability to perform clustering on project--join queries under bag semantics is also a key ingredient for handling clustering over joins that involve both numerical and categorical attributes.
}

\begin{comment}
\begin{definition}
\label{def:kmedian}
    [Relational $k$-median clustering] Given a database instance $\I$, a join query $\Q$, and a positive integer parameter $k$, the goal is to compute a set $\ret\subseteq \Q(\I)$ of size $|\ret|\leq k$ such that $\kmedian_\ret(\Q(\I))$ is minimized. %In the discrete relational $k$-median clustering the set $\ret$ should be a subset of $\Q(\I)$.
\end{definition}

\begin{definition}
\label{def:kmeans}
    [Relational $k$-means clustering] Given a database instance $\I$, a join query $\Q$, and a positive integer parameter $k$, the goal is to compute a set $\ret\subseteq \Q(\I)$ of size $|\ret|\leq k$ such that the $\kmeans_\ret(\Q(\I))$ is minimized. %In the discrete relational $k$-means clustering the set $\ret$ should be a subset of $\Q(\I)$.
\end{definition}
\end{comment}

\vspace{-0.5em}
\subsection{Related work}

For the relational $k$-center problem \revone{assuming only numerical attributes}, Agarwal et al.~\cite{agarwal2024computing} propose a $(2+\eps)$-approximation in $\O(\min\{k^2N^{\fhw}, kN^\fhw+k^{d/2}\})$ time.

%$O(1)$-approximation algorithm that runs in $\O(k^2 N^\fhw)$ time\footnote{The notation $\O$ is used to hide a $\log^{O(1)} N$ factor.} and an $O(1)$-approximation algorithm that returns $O(k)$ centers that runs in $\O(kN^\fhw+k^{d/2})$ time. Given an $O(1)$-approximate solution, they design a coreset and they run a known $\gamma$-approximation algorithm from the standard computational setting on the coreset to get a $(1+\eps)\gamma$-approximation algorithm in $\O(\min\{k^2N^\fhw, kN^\fhw+k^{d/2}\}+\timeCenter_\gamma(k^2))$ time. Applying the algorithm in~\cite{feder1988optimal} on the coreset, they get a $(2+\eps)$-approximation in $\O(\min\{k^2N^{\fhw}, kN^\fhw+k^{d/2}\})$ time.

For the relational $k$-means problem,
Khamis et al.~\cite{khamis2020functional}, gave an efficient implementation of the Lloyd’s $k$-means heuristic. %However it is known that the algorithm terminates in a local minimum without any guarantee on the approximation factor.
Relative approximation algorithms are also known.
Curtin et al.~\cite{curtin2020rk} construct a weighted set of tuples (\emph{grid-coreset}) such that an approximation algorithm for the weighted $k$-means clustering in the grid-coreset returns an approximation solution to the relational $k$-means clustering. \revone{If all attributes are numerical then} their algorithm runs in $\O(k^m N^\fhw +\timeMeans_\gamma(k^m))$ time and has a $(\gamma^2+4\gamma\sqrt{\gamma}+4\gamma)$-approximation factor. \revone{In the general case, having both numerical and categorical attributes, the running time is $\O(k^d N^\fhw +\timeMeans_\gamma(k^d))$.}
%For some join queries, the approximation factor can be improved to $(4\gamma+2\sqrt{\gamma}+1)$.
\revone{All subsequent works focus exclusively on the case where all attributes are numerical.}
Moseley et al.~\cite{moseley2021relational} designed a relational implementation of the $k$-means++ algorithm~\cite{aggarwal2009adaptive, arthur2007k}. For a constant $\eps\in(0,1)$, their algorithm runs in
$\O(k^4N^\fhw +k^2N^\fhw +\timeMeans_\gamma(k))$ 
%$O(k^4N^\fhw\log N +k^2N^\fhw\log^9 N+\timeMeans_\gamma(k\log N))$ 
expected
%$O(k^4N^\fhw\log^4 N +\frac{1}{\eps^2}k^2N^\fhw\log^8 N+\timeMeans_\gamma(k\log N))$
time and has a $(320+644(1+\eps)\gamma)$-approximation factor.
Finally, Esmailpour and Sintos~\cite{esmailpour2024improved} used the hierarchical aggregation method~\cite{chen2022coresets} to derive an $O(1)$-approximation randomized algorithm in $\O(k^2N^\fhw)$ time. Then, they construct a geometric coreset in $\O(k^2N^\fhw+k^4)$ time, leading to a $(4+\eps)\gamma$-approximation algorithm with high probability in $\O(k^2N^\fhw+k^4+\timeMeans_\gamma(k^2))$ time. %Running a known $\gamma$-approximation algorithm from the standard computational setting on the coreset, they get $(4+\eps)\gamma$-approximation algorithm with high probability in $\O(k^2N+k^4+\timeMeans_\gamma(k^2))$ time.
%Their algorithm can be made deterministic but the running time increases.
The algorithm in~\cite{esmailpour2024improved} can be extended for the relational $k$-median clustering getting a $(2+\eps)\gamma$-approximation algorithm in $\O(k^2N^\fhw+k^4+\timeMedian_\gamma(k^2))$ time.
%This is the only known relative approximation algorithm for the relational $k$-median problem.
Additive approximation algorithms are also known for relational clustering problems~\cite{chen2022coresets}. Recently, Surianarayanan et al.~\cite{10.1145/3767712} studied the relational $k$-center clustering problem with outliers.
Finally, there is a lot of recent work on relational algorithms for learning problems such as, linear regression and factorization~\cite{rendle2013scaling, khamis2018ac, kumar2015learning, schleich2016learning}, SVMs~\cite{abo2021relational, abo2018database, yang2020towards}, Independent Gaussian Mixture models~\cite{cheng2019nonlinear, cheng2021efficient}. %In~\cite{schleich2019learning} the authors give a nice survey about learning over relational data.
Database research has explored various combinatorial problems defined over relational data, such as ranked enumeration~\cite{deep2021ranked, deep2022ranked, tziavelis2020optimal},  quantiles~\cite{tziavelis2023efficient}, direct access~\cite{carmeli2023tractable}, diversity~\cite{merkl2023diversity, arenas2024towards, agarwal2024computing}, and top-$k$~\cite{tziavelis2020optimal}. A survey is given by Schleich et al.~\cite{schleich2019learning}.
\vspace{-0.5em}
\subsection{Our results}
\begin{table*}[t]
%\resizebox{\linewidth}{!}{
\begin{tabular}{|c|c|c|c|}
\hline
\textbf{Problem}&\textbf{Method}&\textbf{Approximation}&\textbf{Running time}\\\hline
\multirow{3}{3.9em}{$k$-center}
&\cite{agarwal2024computing}&$(1+\eps)\gamma$&$\O(k^2N^\fhw+\timeCenter_\gamma(k^2))$\\\cline{2-4}
&\cite{agarwal2024computing}&$(1+\eps)\gamma$&$\O(kN^\fhw+k^{\left\lceil d/2\right\rceil+1}+\timeCenter_\gamma(k))$\\\cline{2-4}
&Theorem~\ref{thm:kCenter}&$(1+\eps)\gamma$&$\O(kN^\fhw+\timeCenter_\gamma(k))$\\\hline\hline
\multirow{2}{3.9em}{$k$-median}
&\cite{esmailpour2024improved}&$(2+\eps)\gamma$&$\O(k^2N^\fhw+k^4+\timeMedian_\gamma(k^2))$\\\cline{2-4}
&Theorem~\ref{thm:kmeansopt}&$(2+\eps)\gamma$&$\O(kN^\fhw+\timeMedian_\gamma(k))$\\\hline\hline
\multirow{4}{3.9em}{$k$-means}
&\cite{curtin2020rk}&$\gamma^2+4\gamma\sqrt{\gamma}+4\gamma$&$\O(k^m N^\fhw +\timeMeans_\gamma(k^m))$\\\cline{2-4}
&\cite{moseley2021relational}&$320+644(1+\eps)\gamma$&$\O(k^4N^\fhw+\timeMeans_\gamma(k))$\\\cline{2-4}
&\cite{esmailpour2024improved}&$(4+\eps)\gamma$&$\O(k^2N^
\fhw+k^4+\timeMeans_\gamma(k^2))$\\\cline{2-4}
&Theorem~\ref{thm:kmeansopt}&$(4+\eps)\gamma$&$\O(kN^
\fhw+\timeMeans_\gamma(k))$\\\hline
\end{tabular}
%}
\caption{Comparison of our clustering algorithms with state-of-the-art relative approximation methods for datasets \revone{where all attributes are numerical (real or integer numbers)}.
%For our new algorithms we show the approximation for the discrete relational $k$-center, $k$-median, and $k$-means clustering problems. For the geometric version of the studied problems the approximation factor of our algorithms is always $(1+\eps)\gamma$.
The running time is shown in data complexity. We assume that $\eps\in(0,1)$ is an arbitrary constant. The notation $\O$ is used to hide $\log^{O(1)}N$ factors from the running time. $\timeCenter_\gamma(y)$ (resp. $\timeMedian_\gamma(y)$, $\timeMeans_\gamma(y)$) is the running time of a known $\gamma$-approximation algorithm over $y$ points for the $k$-center (resp. $k$-median, $k$-means) clustering in the standard computational setting. $\fhw$ is the fractional hypertree width of $\Q$. Recall that $d$ is the number of attributes in query $\Q$.
%The number of attributes in the query $\Q$ is denoted by $d$.
}
\label{table:results}
\vspace{-1em}
\end{table*}
%\revone{We focus on the case where all attributes are numerical; that is, the domain of each attribute $A \in \allattr$ consists of either real or integer values, as in~\cite{moseley2021relational, chen2022coresets, agarwal2024computing, esmailpour2024improved}. In the main Sections~\ref{sec:RBBD}, \ref{sec:kCenter}, and~\ref{sec:kmeans}, we present faster algorithms than the state of the art for relational clustering when all attributes are numerical. In Section~\ref{sec:cat}, we combine our results with the general framework of~\cite{curtin2020rk} to derive efficient algorithms for mixed-attribute data, where attributes may be either categorical or numerical. Furthermore, throughout this paper, we describe our results under the assumption that $\Q$ is an acyclic join query. We later extend our analysis to general cyclic queries using standard tools from database theory.}
\revtwo{Unless otherwise stated, we assume that $\dom(A) = \Re$ (or $\dom(A)=\mathbb{Z}$) for each $A \in \allattr$ and the query $\Q$ is acyclic. In Appendix~\ref{sec:cat}, we show how our algorithms extend when the attribute set $\allattr$ contains both numerical and categorical domains, while in Section~\ref{sec:ext} we extend our results to cyclic queries.}

First, in Section~\ref{sec:overview} we give the high-level ideas of our new method. We highlight the differences of our approach with the previously known algorithms, and we propose a new geometric insight on how to improve the running time.

%Our results are summarized in the next bullets.

$\bullet$ In Section~\ref{sec:RBBD} we present a variation of the BBD tree called Randomized BBD tree (or simply RBBD tree) that has similar properties to the BBD tree in the standard computational setting. However, we show that given a node of the RBBD tree, its children can be constructed in almost linear time with respect to the size of the input database in the relational setting. This is the key idea to design faster algorithms for various optimization problems in the relational setting.

$\bullet$ In Section~\ref{sec:kCenter} we show efficient algorithms for the relational $k$-center clustering. First, we propose an $O(1)$-approximation algorithm that runs in $\O(kN)$ time, with high probability. The result of our new algorithm is given as input to the coreset construction from~\cite{agarwal2024computing} to get a $(1+\eps)\gamma$-approximation algorithm in $\O(kN+\timeCenter_\gamma(k))$ time. %where $\timeCenter_\gamma(k)$ is the running time of a known $\gamma$-approximation for the $k$-center clustering problem in the standard computational setting. %For example, running the Feder and Greene~\cite{} algorithm on the coreset, we get a $(2+\eps)$-approximation algorithm in $\O(kN)$ time.

$\bullet$ In Section~\ref{sec:kmeans} we design efficient algorithms for the relational $k$-median/means clustering. We get our intuition from~\cite{har2004coresets} and we propose a geometric approach using our RBBD tree in the relational setting to design an $O(1)$-approximation algorithm for the relational $k$-median/means clustering that returns $O(k\log^3 N)$ centers.
The returned set is given as input to the coreset construction method (in the relational setting) in~\cite{esmailpour2024improved} to derive the final result. Interestingly, our new observations in this work can be used to accelerate the coreset construction of~\cite{esmailpour2024improved}. In the end, we get a $(2+\eps)\gamma$-approximation (resp. $(4+\eps)\gamma$-approximation) for the relational $k$-median (resp. $k$-means) clustering in $\O(kN+\timeMedian_\gamma(k))$ (resp. $\O(kN+\timeMeans_\gamma(k))$) time, with high probability. %where $\timeMedian_\gamma(k)$ (resp. $\timeMeans_\gamma(k)$) is the running time of a $\gamma$-approximation algorithm for the $k$-median (resp. $k$-means) problem in the standard computational setting.

%$\bullet$ In Section~\ref{}, we use our new algorithm for the relational $k$-center clustering to implement an approximation version of the Gonzalez algorithm in the relational setting in $\O(kN)$ time. The efficient Gonzalez implementation allows us to solve various optimization problems in the relational setting in $\O(kN)$ time, such us diversity, and fair clustering/diversity problems.

%In fact, in any iteration, having already selected a set $S\subseteq \Q(\I)$, we compute in $\O(kN)$ time a point $t\in \Q(\I)$ such that $\dist(t,S)\geq (1-\eps)\max_{p\in \Q(\I)}\dist(p,S)$.

\revone{
$\bullet$ In Appendix~\ref{sec:cat}, we combine our results with the general framework of~\cite{curtin2020rk} to extend our algorithms for mixed-attribute data, where attributes may be either categorical or numerical.
}

$\bullet$ 
Due to the generality of our approach, in Section~\ref{sec:ext} and Appendix~\ref{appndx:gonzalez},
we show how our new method can be used to get faster algorithms for other optimization problems, such as diversity and fairness problems, in the relational setting.

%In Section~\ref{} we briefly explain how all our algorithms can be extended to general cyclic join queries increasing the running time from $\O(kN)$ or $\O(k^2N)$ to $\O(kN^{\fhw})$ or $\O(k^2N^{\fhw})$, respectively, where $\fhw$ is the fractional hypertree width of $\Q$.

The running times and approximation factors of our new relational clustering algorithms, along with a comparison to the state-of-the-art, are shown in Table~\ref{table:results}.

%\paragraph{Remark 1}
%Similarly to~\cite{esmailpour2024improved}, if the goal is to find a set of $k$ centers $\ret\subset \Re^d$ (instead of $\ret\subseteq \Q(\I)$), all our new algorithms for the relational clustering problems have a $(1+\eps)\gamma$ approximation factor.

%For Simplicity, we show all our algorithms for relational clustering using the definitions from Section~\ref{}, where $\ret$ should be a subset of $\Q(\I)$.

%In Section~\ref{} we combine our method with~\cite{} to design an algorithm considering both categorical and numerical attributes, while Section~\ref{} we use standard methods to extend our algorithm to general join queries.}

\vspace{-0.5em}
\section{Preliminaries}
\label{sec:prelim}
We revisit various known tools, algorithms, and data structures that we use in the next sections. For a point $p\in \Re^d$ and $r\in \Re$, let $\ball(p,r)=\{x\in\Re^d\mid \dist(p,x)\leq r\}$ be the ball of center $p$ and radius $r$. 

\paragraph{Box}
A box $\rho$ is an axis-parallel $d$-dimensional hyper-rectangle that is defined as the product of $d$ intervals over the attributes, i.e., $\rho=I_1\times \ldots\times I_d$, where $I_i=[a_i,b_i]$ for $a_i,b_i\in \Re$.
\begin{comment}
\subsection{Coreset}
We give the definition of coresets for all the relational $k$-center, $k$-median and $k$-means clustering problems. Our algorithms construct small enough coresets to approximate the cost of the relational clustering.

A set $\coreset\subseteq \Q(\I)$ is an $\eps$-coreset for the relational $k$-median clustering problem on $\Q(\I)$ 
if for any set of $k$ centers $Y\subseteq \Q(\I)$,
$$(1-\eps)\kcenter_Y(\Q(\I))\leq \kcenter_Y(C)\leq (1+\eps)\kcenter_Y(\Q(\I)).$$

A weighted set $\coreset\subseteq \Re^d$ is an $\eps$-coreset for the relational $k$-median clustering problem on $\Q(\I)$ if for any set of $k$ centers $Y\subset \Re^d$,
\begin{equation}
\label{def:coreset-kmedian}
(1-\eps)\kmedian_Y(\Q(\I))\leq \kmedian_{Y}(\coreset)\leq (1+\eps)\kmedian_Y(\Q(\I)).    
\end{equation}

Similarly, a weighted set $\coreset$ is an $\eps$-coreset for the relational $k$-means clustering problem, if for any set of $k$ centers $Y\subset \Re^d$,
\begin{equation}
\label{def:coreset-kmeans}
(1-\eps)\kmeans_Y(\Q(\I))\leq \kmeans_{Y}(\coreset)\leq (1+\eps)\kmeans_Y(\Q(\I)).    
\end{equation}
\end{comment}

\renewcommand{\net}{\mathcal{N}}
\newcommand{\prob}{\psi}
\paragraph{$\eps$-sample}
Let $(P,\mathcal{R})$ be a set system where $P$ is a set of $n$ points in $\Re^d$ and $\mathcal{R}$ is the set of all (infinitely many) possible boxes in $\Re^d$. For a parameter $\eps\in(0,1)$, a set $\net\subseteq P$ is an $\eps$-sample if for every $\rho\in \mathcal{R}$,
$|\frac{|P\cap \rho|}{|P|}-\frac{|\net\cap \rho|}{|\net|}|\leq \eps.$
It is known~\cite{har2011geometric, bernard2001discrepancy} that a set of $O(\eps^{-2}\log\prob^{-1})$ uniform random samples from $P$ is an $\eps$-sample of $(P,\mathcal{R})$ with probability at least $1-\prob$.

\paragraph{Oracles for aggregation queries}
We show known oracles that can be used to access some tuples from the join results or compute some statistics on the join results without explicitly evaluating the join query.
In the next sections, we aim to count or sample from the join results in $\Q(\I)$ that lie inside a box. More formally, let $\rho$ be a box in $\Re^d$.
The goal is i) count the number of tuples $|\Q(\I)\cap \rho|$, ii) sample uniformly at random from $\Q(\I)\cap \rho$, iii) report all tuples in $\Q(\I)\cap \rho$, and iv) pick any arbitrary tuple from $\Q(\I)\cap \rho$.
%We note that if we can execute ii) efficiently, then iv) can also be executed efficiently.
The next lemma follows from~\cite{esmailpour2024improved}. 
%\vspace{-0.5em}
\begin{comment}
The axis-parallel hyper-rectangle $\mathcal{R}$ is defined as the product of $d$ intervals over the attributes, i.e., $\mathcal{R}=\{I_1\times \ldots\times I_d)$, where $I_i=[a_i,b_i]$ for $a_i,b_i\in \Re$.
Hence, $\mathcal{R}$ defines a set of $d$ linear inequalities over the attributes, i.e., a tuple $t$ lies in $\mathcal{R}$ if and only if $a_j\leq \pi_{A_j}(t)\leq b_j$ for every $A_j\in\allattr$.
Let $p$ be a tuple in a relation $R_i$.
If $a_j\leq \pi_{A_j}(p)\leq b_j$ for every $A_j\in \allattr_i$, then we keep $p$ in $R_i$. Otherwise, we remove it. The set of surviving tuples is exactly the set of tuples that might lead to join results in $\mathcal{R}$. Let $\I'\subseteq \I$ be the new database instance such that $\Q(\I')=\Q(\I)\cap\mathcal{R}$. The set $\I'$ is found in $O(N)$ time.
Using Yannakakis algorithm~\cite{yannakakis1981algorithms} we can count $|\Q(\I')|$ and we can compute a representative tuple from $\Q(\I')$ in $O(N\log N)$ time and using~\cite{zhao2018random} we can sample $z$ tuples from $\Q(\I')$ in $O((N+z)\log N)$ time. Using hashing, we can improve the running time of the algorithms to $O(N)$ and $O(N+z\log N)$, respectively.
\end{comment}

\begin{lemma}[\cite{esmailpour2024improved}]
\label{lem:Rects}
    \revtwo{Let $\I$ be a database instance of $N$ tuples (with numerical attributes) and let $\Q$ be an acyclic query over $\I$. Let $\rho$ be a box in $\Re^d$. There exists an algorithm}
    \begin{enumerate}
        \item  \revtwo{$\mathsf{CountRect}(\Q, \I,\rho)$ to count $|\Q(\I)\cap \rho|$ in $O(N)$ time,}
        \item \revtwo{$\mathsf{SampleRect}(\Q, \I,\rho, z)$ to sample $z$ points from $\Q(\I)\cap \rho$ in $O(N+z\log N)$ time,}
        \item \revtwo{$\mathsf{ReportRect}(\Q,\I,\rho)$ to report all points in $\Q(\I)\cap \rho$ in $O(N+|\rho\cap \Q(\I)|)$ time,}
        \item \revtwo{$\mathsf{ReprRect}(\Q,\I,\rho)$ to find a representative tuple from $\Q(\I)\cap \rho$ in $O(N)$ time.}
    \end{enumerate}
       \revtwo{The running times of the oracles hold assuming perfect hashing.}\footnote{\revtwo{We assume perfect hashing for simplicity. A standard randomized hash table yields the same asymptotic bounds with high probability; perfect hashing assumption is used only to avoid carrying additional (routine) randomness through the analysis.}}
    %with probability at least $1-\frac{1}{N^{m+4}}$.
\end{lemma}
%\vspace{-0.7em}
%Let $B\subseteq \allattr$ be a subset of the attributes. The result in Lemma~\ref{lem:Rects} can also be used (straightforwardly) to compute $|\widebar{\pi}_{B}(\Q(\I))\cap \mathcal{R}|$ or sample $z$ samples from $\widebar{\pi}_{B}(\Q(\I))\cap \mathcal{R}$ in the same running time, since $|\widebar{\pi}_{B}(\Q(\I))\cap \mathcal{R}|=|\Q(\I)\cap \mathcal{R}|$.

\begin{comment}
As we will see in the next sections, we propose efficient approximation algorithms for the geometric 
clustering problems in relational data. Any such approximated solution can be converted to a solution (losing a small constant factor $2$ or $4$) for the discrete version of the problem by finding the nearest neighbor in $\Q(\I)$ from each center in $\Re^d$. In order to handle nearest neighbor queries among the tuples in $\Q(\I)$ in the Euclidean metric we use the data structure described in~\cite{deep2021ranked, tziavelis2020optimal}. Given a query point $p$, let $\textsf{NN}(p,\Q(\I))=\argmin_{t\in\Q(\I)}\dist(t,p)$. We give the proof of the next Lemma in Appendix~\ref{appndxProofs}.
%\vspace{-0.7em}
\begin{lemma}
\label{lem:NN}
    %Let $\Q$ be an acyclic join query.
    Given a point $p\in \Re^d$, there exists an algorithm $\mathsf{NearNeighbor}(\Q,\I,p)$ to compute $\textsf{NN}(p,\Q(\I))$ in $O(N)$ time.
\end{lemma}
\vspace{-1em}
\end{comment}

%\vspace{-0.5em}
%\subsection{Geometric data structures}
\paragraph{BBD tree}
\revthree{
The main geometric data structure we employ is the \emph{BBD tree}~\cite{arya2000approximate, arya1998optimal}, a variant of the quadtree~\cite{finkel1974quad}.
A BBD tree $\mathcal{T}$ on a set $P$ of $n$ points in $\Re^d$ is a binary tree of height $O(\log n)$ with $n$ leaves, where each point in $P$ is stored in exactly one leaf node.
Let $\square$ denote the axis-aligned box with the smallest volume containing $P$.
Each node $u$ of $\mathcal{T}$ is associated with a region $\square_u \subseteq \Re^d$, which is either a box or the difference of two boxes (a box with a hole).
This region contains precisely the points from $P$ that are stored in the leaves of the subtree rooted at $u$.
If $\root$ is the root node of $\mathcal{T}$, then $\square_{\root} = \square$.
Let $P_u = P \cap \square_u$ denote the points contained in the subtree rooted at $u$.
If $|P_u| = 1$, then $u$ is a leaf.
Otherwise, $u$ has two children, say $w$ and $z$, such that $\square_w$ and $\square_z$ form a partition of $\square_u$.
The regions associated with the nodes of $\mathcal{T}$ thus induce a hierarchical partition of $\Re^d$.
For illustration, consider the set of $23$ points shown in Figure~\ref{fig:part}.
A partially constructed (and, for simplicity, slightly simplified) BBD tree of this set is shown in Figure~\ref{fig:BBDexample}.
The blue segments in Figure~\ref{fig:part} depict the hierarchical partition induced by the constructed BBD tree.
For the node $v$ in Figure~\ref{fig:BBDexample}, the region $\square_v$ is the box defined by opposite corners $C$ and $J$, while for node $w$, the region $\square_w$ is the box defined by opposite corners $C$ and $K$, excluding the inner box defined by opposite corners $L$ and $M$. For every node $u$ of the BBD tree in Figure~\ref{fig:BBDexample}, we also show the number of points in $P_u$. For example, $|P_w|=4$ because the box defined by $[C,K]$ contains $8$ points, but the hole box $[L,M]$ contains $4$ points.
A BBD tree has $O(n)$ space and can be constructed in $O(n\log n)$ time. 
}

\revthree{
A BBD tree requires $O(n)$ space and can be constructed in $O(n \log n)$ time.
Given a parameter $\eps \in (0,1)$ and a ball $\mathcal{B}(x,r)$ in $\Re^d$, the BBD tree supports a query procedure, denoted $\mathcal{T}(x,r)$, that returns a set of \emph{canonical nodes} $\mathcal{U}(x,r) = {u_1, \ldots, u_\kappa}$, with $\kappa = O(\log n + \eps^{-d+1})$, in $O(\log n + \eps^{-d+1})$ time.
These nodes satisfy $\square_{u_i} \cap \square_{u_j} = \emptyset$ for all distinct $i,j$, and
$\mathcal{B}(x,r)\subseteq \bigcup_{1\leq i\leq \kappa}\square_{u_i}\subseteq \mathcal{B}(x,(1+\eps)r)$.
In other words, the BBD tree returns a set of boxes (with or without holes) such that the union of them completely covers the ball $\mathcal{B}(x,r)$ and might cover some part of $\mathcal{B}(x,(1+\eps)r)\setminus \mathcal{B}(x,r)$.
The query procedure is straightforward once the tree is built.
Starting from the root of $\mathcal{T}$ and an initially empty set $\mathcal{U}(x,r)$, the algorithm recursively traverses the tree.
When visiting a node $u$, if $\square_u \subseteq \mathcal{B}(x,(1+\eps)r)$, the node $u$ is added to $\mathcal{U}(x,r)$.
If $\square_u \cap \mathcal{B}(x,r) = \emptyset$, the traversal of this branch terminates.
Otherwise, the procedure continues recursively with both children of $u$.
For example, in Figure~\ref{fig:part}, let $\mathcal{B}(x,r)$ denote the inner red circle and $\mathcal{B}(x,(1+\eps)r)$ the outer dashed circle.
The corresponding canonical nodes $\mathcal{U}(x,r)$ are highlighted as red nodes in Figure~\ref{fig:BBDexample}. We note that the boxes in $\mathcal{U}(x,r)$ cover all points in the inner red circle and, in addition, cover a point from $P$ in the annulus $\mathcal{B}(x,(1+\eps)r)\setminus \mathcal{B}(x,r)$, because the box with opposite corners $N$ and $S$ lies entirely within the outer dashed circle.
BBD trees can be used for different types of queries, such as reporting, counting, or sampling queries.}
%For example, for counting queries, for every node $u$ of $\mathcal{T}$, we store $u.c=|P_u|$. Given a query ball $\mathcal{B}(x,r)$, we compute the set of of canonical nodes $\mathcal{U}(x,r)$ and we return $\mathsf{count}=\sum_{u\in\mathcal{U}(x,r)}u.c$. By definition, it holds that $|\mathcal{B}(x,r)\cap P|\leq \mathsf{count}\leq |\mathcal{B}(x,(1+\eps)r)\cap P|$. Similarly, in sampling queries we aim to return uniformly at random a point from $P\cap\mathcal{B}(x,r)$. In the preprocessing phase, for every node $u$ of $\mathcal{T}$, we compute the count $u.c$. Given a query ball $\mathcal{B}(x,r)$ we compute $\mathcal{U}(x,r)$ and we select a node $v\in \mathcal{U}(x,r)$ with probability $\frac{v.c}{\sum_{u\in\mathcal{U}(x,r)}u.c}$. Then we traverse the tree from $v$ to a leaf node of $\mathcal{T}$ following the left or right child of a nodes at random based on the precomputed counts. It is straightforward to guarantee that this method returns a point $p$ uniformly at random from an area which is a superset of $\mathcal{B}(x,r)$ and a subset of $\mathcal{B}(x,(1+\eps)r)$. The query time of both counting and sampling queries is $O(\log n + \eps^{-d+1})$.

%By reporting all points $P_{u_i}$ for $i\leq \kappa$, the BBD tree can be used for reporting all points in $P\cap \mathcal{B}(x,r)$ along with some points from $P\cap (\mathcal{B}(x,(1+\eps)r)\setminus \mathcal{B}(x,r))$.

%\paragraph{kd tree}

%\subsection{High level ideas}

\begin{figure}[h!]
    \centering
    \begin{minipage}{0.34\textwidth}
        \centering
        \includegraphics[width=\textwidth]{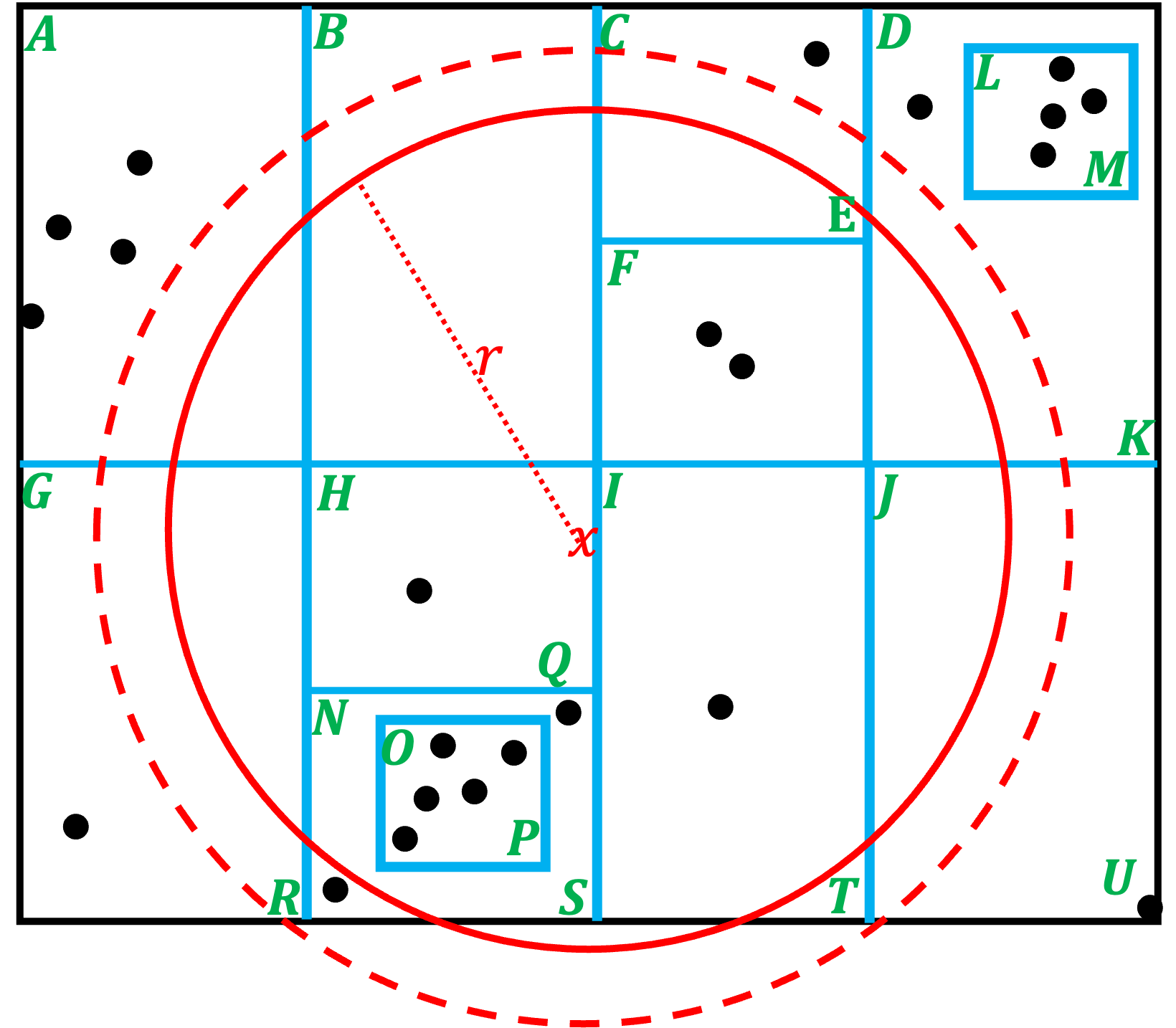}
        \caption{\revthree{Partition of a partially constructed BBD tree over the set of black points. Red circle represents the query ball $\mathcal{B}(x,r)$ while the larger dashed circle represents $\mathcal{B}(x,(1+\eps)r)$.}}
        \label{fig:part}
    \end{minipage}\hfill
    \begin{minipage}{0.64\textwidth}
        \centering
        \includegraphics[width=\textwidth]{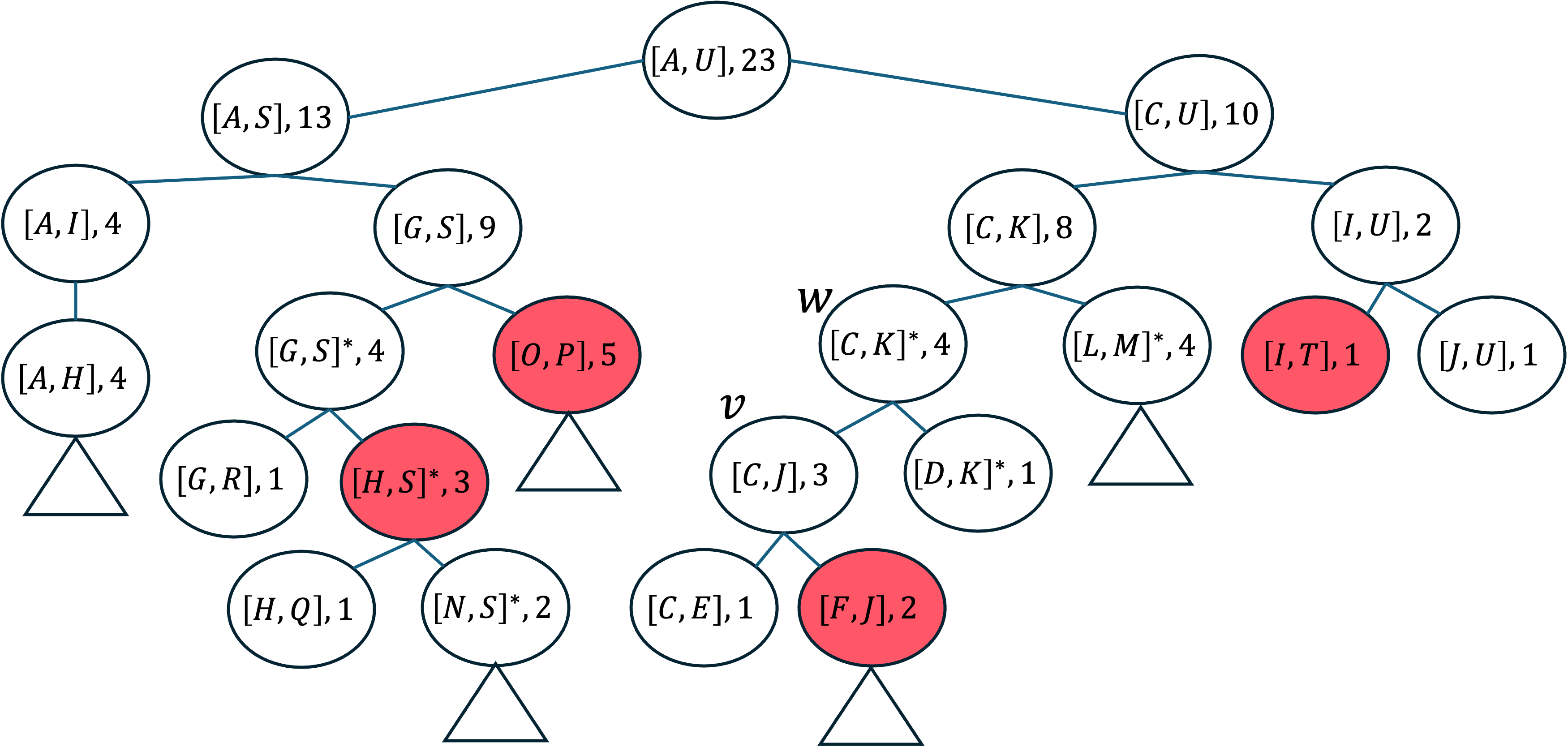}
        \caption{\revthree{Partially constructed BBD tree over the black points of Figure~\ref{fig:part}.
        A node $u$ of the form $[X,Y], c$ has a region $\square_u$ defined by the opposite corners $X, Y$, and $|P_u|=|\square_u\cap P|=c$.
Nodes with $[X,Y]^*,c$ correspond to nodes whose region is a box with a hole.
The triangles below certain nodes denote subtrees omitted for simplicity and clarity.
Red nodes indicate the set of canonical nodes $\mathcal{U}(x,r)$ corresponding to the ball $\mathcal{B}(x,r)$ shown in Figure~\ref{fig:part}.}}
        \label{fig:BBDexample}
    \end{minipage}
    \vspace{-1.5em}
\end{figure}

%% file: treeConstruction.tex
%\section{Construction of trees}
\vspace{-0.2em}
\section{Overview}
\label{sec:overview}
In this section, we describe the intuition and the high-level ideas of our approach. In the next sections, we formally prove the claims made in this section.
\vspace{-0.3em}

\paragraph{Previous methods}
Most of the previously known methods for relational clustering~\cite{esmailpour2024improved, agarwal2024computing, chen2022coresets, moseley2021relational, curtin2020rk} follow the same high-level approach. They first design a constant approximation algorithm with more than $k$ centers. Let $Z$ be that set. Then $Z$ is used to construct a small coreset $\coreset$ for the relational clustering problem. Finally, a known clustering algorithm in the standard computational setting is executed on $\coreset$ to derive the final solution. Interestingly, the coreset construction is 
relatively fast. In most cases, the running time is dominated by the first part, i.e., constructing any constant approximation solution with potentially more than $k$ points. For example, consider the state--of--the--art algorithms for the relational $k$-center~\cite{agarwal2024computing} and relational $k$-median/means clustering~\cite{esmailpour2024improved}. In these papers, they construct $O(1)$-approximation algorithms with $\O(k^2)$ centers in $\O(k^2N)$ time using a \emph{hierarchical aggregation method}.
At a high level, they place all relations from $\allrel$ as leaf nodes in a balanced binary tree. %Then they recursively do the following.
For a leaf node $u_1$ run a standard clustering algorithm
in the standard computational setting, considering only the tuples of the relation stored in $u_1$. Let $\ret_1$ be the returned set of $k$ centers. Equivalently, they get $\ret_2$ from the sibling node $u_2$. Then, they show that $\ret_{1\times 2}=\ret_1\times\ret_2$ is an $O(1)$-approximation solution for the clustering problem on the result of joining the tables in leaf nodes $u_1$ and $u_2$. Then they construct a coreset based on $\ret_{1\times 2}$ in time $\O(|\ret_{1\times 2}|\cdot N)$ and continue recursively on higher levels of the tree.
While this method is useful to design $O(1)$-approximation algorithms, the drawback is that it constructs intermediate sets of size $\Theta(k^2)$ (for example $\ret_{1\times 2}$), leading to $\Omega(k^2\cdot N)$ time algorithms.

\begin{wrapfigure}{r}{0.35\textwidth}
    \vspace{-10pt} % optional: adjust vertical alignment
    \centering
    \includegraphics[width=0.33\textwidth]{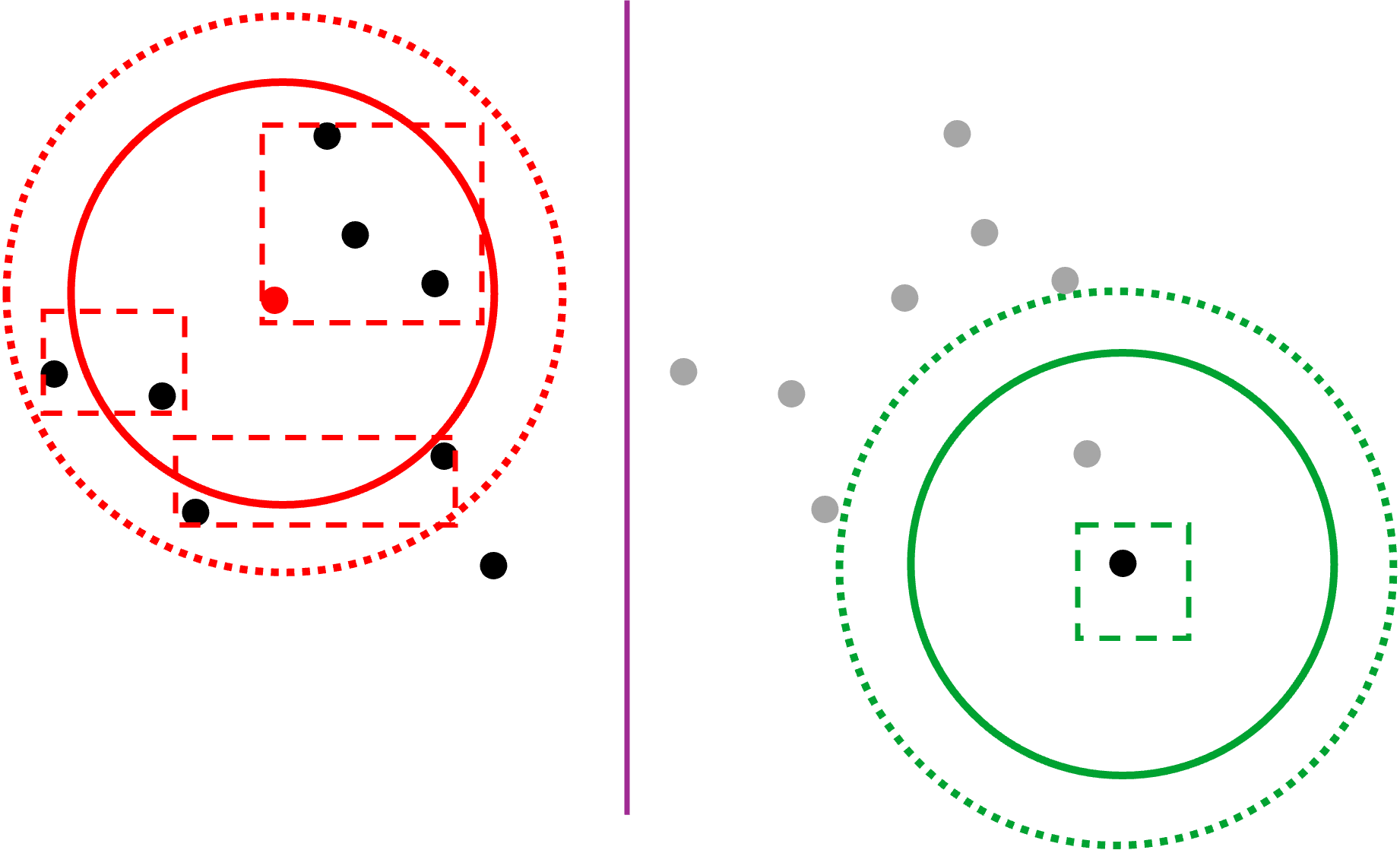}
    \vspace{-3pt} % optional: adjust spacing under image
    \caption{\revthree{%This is the same dataset as in Figure~\ref{fig:basicAlg}. 
        The algorithm first arbitrarily selects the red center (left). Using the BBD tree, it identifies the canonical nodes (three red dashed rectangles), which contain all points within distance $2r$ from the red center and three points within distance $(1+\eps)2r$. These canonical nodes are then marked as inactive. Next, the algorithm selects the green point (right), since it lies in an active region, and computes the canonical node (green dashed rectangle) within distance $2r$, corresponding to the green dashed rectangle, which is subsequently marked as inactive. In the end, all points lie in inactive regions, so the algorithm terminates, returning a $(2+\eps)$-approximation.
        }}
          \label{fig:ImprovedAlg}
          \vspace{-1em}
\end{wrapfigure}
\paragraph{Our approach}
We give the high-level idea of our approach using the $k$-center clustering as an example. Let $P$ be a set of $n$ points in $\Re^d$ and $k$ be a positive integer parameter in the standard computational setting.
In computational geometry, there are various efficient constant approximation algorithms for the $k$-center clustering. 
%For example, Feder and Greene~\cite{}, described an $O(n\log k)$ time algorithm that returns a set $\ret$ such that $\kcenter_\ret(P)\leq 2\kcenter_{\opt(P)}(P)$. This is the optimum deterministic algorithm for $k$-center clustering in Euclidean space. In~\cite{}, the authors give a randomized algorithm with $O(n)$ expected time. 
While these algorithms are fast in the standard computational setting, it is not clear how to efficiently extend them to the relational setting. 
%Next, we focus on another simpler approximation algorithm. %for a constant approximation for the $k$-center clustering problem that we will use to extend it in the relational setting.
Consider the 
following simple algorithm for the $k$-center problem on $P$.
Let $r=\kcenter_{\opt(P)}(P)$ be the optimum cost, and assume that $r$ is known upfront.
In the first iteration, we add in $\ret$ any arbitrary point $p_1\in P$. We remove all points in $\mathcal{B}(p_1,2r)\cap P$ from $P$, i.e., all points in $P$ within distance $2r$ from $p_1$, and we continue with the next iteration. In the $i$-th iteration we choose an arbitrary point $p_i\in P\setminus \bigcup_{j<i}\mathcal{B}(p_j,2r)$, we remove all points $\mathcal{B}(p_i,2r)\cap P$ and we continue for a total number of $k$ iterations.
It is easy to verify that this algorithm returns a $2$-approximation for the $k$-center problem on $P$.
%\revthree{An example implementation of this algorithm on a simple point set is illustrated in Figure~\ref{fig:basicAlg}.}
%Indeed, in every iteration $i$, the ball $\mathcal{B}(p_i,2r)\cap P$ completely includes one optimum cluster.
Assuming that the optimum $r$ is known, a naive implementation of this algorithm takes $O(nk)$ time. If $r$ is not known, we need to execute a binary search on the distances of $P$ (using a WSPD~\cite{callahan1995decomposition, har2005fast}), leading to $O(nk\log n)$ running time.
However, there is a faster implementation of this algorithm in the standard computational setting. We build a BBD tree $\mathcal{T}$ over $P$. For every node $u$ of $\mathcal{T}$, we store a representative point $u.rep$ from the points stored in the leaf nodes of the subtree rooted at $u$, i.e., $u.rep\in P_u$ (see Section~\ref{sec:RBBD}). In the first iteration $p_1=\root.rep$, is the representative point in the root. Then we run the query $\mathcal{T}(p_1,2r)$  and we get the set of canonical nodes $\canonical(p_1,2r)$. We mark every node in $\canonical(p_1,2r)$ as 
`inactive'. Then we start from the deepest node in $\canonical(p_1,2r)$ and we update the representative points of the active nodes to make sure that
no representative point lies in a leaf node of a subtree rooted at an inactive node.
%This can be done efficiently as follows. Let $u\in\canonical(p_1,2r)$ be the deepest canonical node.  Recall that we mark $u$ as inactive. Then starting from the parent node $w$ of $u$ we continue recursively: Let $v_1, v_2$ be the children nodes of $w$. If both $v_1,v_2$ are inactive, we mark $w$ as inactive. If $v_1$ (resp. $u_2$) is active, then we set $w.rep=v_1.rep$ (resp. $w.rep=v_2.rep$). After we are done with all updates in the deepest level of the tree, we continue the updates with nodes in the parents' level.
We continue
similarly in the next $k-1$ iterations, visiting only active nodes.
%The only difference is that during the query procedure to get $\canonical(p_i,2r)$, if we visit a node $u$ in $\mathcal{T}$, where $u$ is inactive, we stop the search towards this branch of the tree.
This implementation returns a set $\ret\subseteq P$ such that $\kcenter_{\ret}(P)\leq 2(1+\eps)\kcenter_{\opt(P)}(P)$ 
%($\tree$ implicitly marks all points in $\ball(p_i,2r)\cap P$ as inactive, and might mark some of the points in $(\ball(p_i,2r(1+\eps))\setminus \ball(p_i,2r))\cap P$ as inactive) 
with running time $O(n\log n + k(\eps^{-d+1}+\log n))=O(n\log n)$, in standard computational setting. \revthree{An example of the faster algorithm in the standard computational setting is shown in Figure~\ref{fig:ImprovedAlg}.}

%In the $i$-th iteration, we pick $p_i=\root.rep$. Notice that $\root.rep$ is a point in $P$ that lies in an active area, i.e., an area that has not been removed in a previous iteration. We query the BBD tree $\mathcal{T}$ to find the set of canonical nodes for the query ball $\ball(p_i,2r)$. During the query procedure, if we visit a node $v$ in $\mathcal{T}$ before we check whether it intersects partially or completely $\ball(p_i,2r)$, we check whether $v$ is inactive. If $v$ is inactive we stop the search towards this branch of the tree. Given the new set of canonical nodes $\canonical(p_i,2r)$, we make them inactive and we update the representative points from all nodes in $\canonical(p_i,2r)$ towards the root of $\mathcal{T}$, as described above.

A straightforward implementation of the BBD tree based algorithm for $k$-center clustering will be expensive in the relational setting. Just to construct the BBD tree over all join results, we need $\Omega(|\Q(\I)|)$ time.
However, we make the following simple but key observation. 
\textbf{We do not really need to construct the entire tree upfront. What if we construct the tree on the fly?}
Consider again the standard computational setting. Initially, assume that we have only constructed the root node $\root$ of the tree $\mathcal{T}$.
%, i.e.,  have computed the range $\square_\root$ and a representative point $\root.rep\in P$.
%Consider a query ball $\ball(p_i,2r)$.
The goal is to compute $\canonical(p_i,2r)$. 
% We start the search procedure from the root node as usual. 
Let $v$ be a node of the constructed part of $\mathcal{T}$ that the query procedure traverses. If the children of $v$ have not been constructed, we construct the left child $v_1$ and the right child $v_2$, and then we continue with the regular query procedure. 
%The update of the representative points is also trivial, since all needed nodes have been constructed.
%Notice that if we have already constructed all the children of the nodes that the query procedure needs in order to compute $\canonical(p_i,2r)$, then we can also update the representative points  without constructing new nodes, i.e, the update procedure of the representative points only use nodes that have been traversed by the query procedure.
%Given a node $v$ of the BBD tree, in the standard computational setting, the children $v_1, v_2$ can be computed in $O(|P_v|)$ time~\cite{}. Assuming the `on the fly' construction of the BBD tree, the $2(1+\eps)$ approximation algorithm for the $k$-center clustering problem in the standard computational setting also runs in $O(n\log n)$ time.
%More generally, 
Given a node $v$ of the BBD tree, if we can construct its two children $v_1, v_2$ in $O(t)$ time, then the execution time is bounded by $O(k\cdot t\cdot (\log n + \eps^{-d+1}))$. While this is a loose upper bound in the standard computational setting, we use it to bound the query procedure in the relational setting.

\revtwo{Our core idea is the on-the-fly construction of a new variant of the BBD tree in the relational setting. This enables the design of several new relational oracles, which in turn yield the fastest known approximation algorithms for a variety of optimization problems. While the underlying idea is conceptually simple, constructing a BBD tree in the relational setting requires multiple steps, combining techniques from computational geometry (BBD trees), database theory (relational counting and sampling oracles), and probability analysis ($\eps$-samples). 
A key advantage of our approach is generality. Many classical algorithms from standard settings can be directly translated to relational settings with minimal effort using the proposed data structure. This sharply contrasts with previous work on relational clustering (for example~\cite{agarwal2024computing, esmailpour2024improved, moseley2021relational}), where algorithms are specifically tailored for particular objectives and often differ significantly from their classical counterparts. 
%In our framework, the underlying data structure allows classical algorithms–such as the k-center algorithm–to be implemented in the relational setting almost verbatim. By comparison, methods like~\cite{agarwal2024computing, esmailpour2024improved, moseley2021relational} involve much more complex, problem-specific constructions for relational data.
}

In the next section, we propose a new variation of the BBD tree, called randomized BBD tree (or RBBD tree in short) such that if a node $v$ of an RBBD tree over $\Q(\I)$ is given, then the children $v_1, v_2$ can be computed in $O(N+\eps^{-2}\log N)$ time in the relational setting.

%In the standard computational setting, the construction of the BBD tree comes (almost) for free. Just to pass over all points we need $O(n)$ time. However,
\vspace{-0.5em}
\newcommand{\rec}{\rho}
\section{RBBD tree}
\label{sec:RBBD}
Before we describe the RBBD tree, we give an overview of how standard BBD trees are constructed as described in~\cite{arya1998optimal}.
%and explain why we cannot construct the children of a node $v$ efficiently in the relational setting.
Then, we introduce a randomized variation of the BBD tree called RBBD tree. 
Although the RBBD tree offers no advantage over the BBD tree in the standard computational setting, we show that it can be beneficial in the relational setting, as the children of any node in the RBBD tree over $\Q(\I)$ can be constructed efficiently.
\vspace{-0.5em}
\subsection{More details on the construction of a BBD tree}
%Next, we briefly show the BBD tree construction as described in~\cite{}.
Let $P$ be a set of $n$ points in $\Re^d$. Recall that every node $u$ of the BBD tree is associated with a region $\square_u$, which is either a box or a box with a hole, and recall that $P_u=\square_u\cap P$. In the latter case, let $\square^+_u$ be the outer box and $\square^-_u$ be the inner box (hole).
A \emph{midpoint box} is defined as any box that can be obtained by a recursive application of the following rule, starting from the initial bounding box $\square_\root$, then: Let $\rec$ be a midpoint box and let $\xi$ be the longest side of $\rec$. Split $\rec$ into two identical boxes by a hyperplane which is orthogonal to the $\xi$-th coordinate axis, passing through the center of $\rec$. See Figure~\ref{fig:midpoint} for an example of midpoint boxes.
The tree is constructed in a way that if $\square_u$ is a box, then $\square_u$ is a midpoint box, while if $\square_u$ is a box with a hole, then both $\square^+_u$ and $\square^-_u$ are midpoint boxes, for any node $u$ of the tree. 
The children of $u$ are constructed using one of the following two methods, \emph{fair split} and \emph{centroid shrink}. At the root of the tree, they apply a fair split to construct its children.
In any other case, if $u$ was generated by fair split (resp. centroid shrink) its children are generated by centroid shrink (resp. fair split).
\input{bbd-figs}

\paragraph{Fair split}
They compute a hyperplane that passes through the center of $\square_u$ and is orthogonal to the longest side of $\square_u$ (or $\square_u^+$ if $\square_u$ is a box with a hole). Notice that even if $\square_u$ contains a hole, then the hyperplane splitting $\square_u$ into two boxes will not intersect $\square_u^-$ because both $\square_u^-$ and $\square_u^+$ are midpoint boxes. In this case, if $w, v$ are the children nodes of $u$, then $\square_w$ will be a box, and $\square_v$ will be a box with a hole, having $\square_v^-=\square_u^-$.
The fair split is executed in $O(1)$ time.

\paragraph{Centroid shrink}
On a high level, the goal is to find a midpoint box $\rec^*$ inside $\square_u$ such that $|\rec^* \cap P|\leq \frac{2}{3}|P_u|$ and $|(\square_u\setminus\rec^*) \cap P|\leq \frac{2}{3}|P_u|$. In order to compute $\rec^*$ their algorithm works in iterations.
We describe the construction assuming that $\square_u$ is a box without a hole. In Appendix~\ref{appndx:RBBDhole}, we discuss the construction when $\square_u$ is a box with a hole.
Initially, let $\rho=\square_u$. They compute the smallest midpoint box that contains all points in $P_u \cap \rho$. Let $\rec'\subseteq \rec$ be this box. They apply $O(d)$ fair splits on $\rec'$ until they find a non-trivial fair split, i.e., a fair split that splits the points in $P_u \cap \rho'$ into two non-empty subsets. Let $\hat{\rec}\subset \rec'$ be one of the two boxes created by the non-trivial fair split such that $|\hat{\rec}\cap P_u|\geq \frac{|\rec'\cap P_u|}{2}$. We say that $\hat{\rho}$ is the majority midpoint box.
The key observation is that $|\hat{\rec}\cap P_u|\leq |\rec\cap P_u|-1$.
If $|\hat{\rec}\cap P_u|> \frac{2}{3}|P_u|$, they set $\rec=\hat{\rec}$ and they continue recursively.
If $|\hat{\rec}\cap P_u|\leq \frac{2}{3}|P_u|$, then $\rec^*=\hat{\rec}$ and then they construct two children nodes $w, v$ of $u$ such that $\square_w=\rec^*$ is a box (without a hole), and  $\square_v$ is a box with a hole with $\square_v^+=\square_u$ and $\square_v^-=\rec^* = \square_w$, as shown in Figure~\ref{fig:shrink}.
As shown in~\cite{arya1998optimal}, the smallest midpoint box $\rec'\subseteq\rec$ that contains all points in $P_u\cap \rec$ can be found in $O(1)$ time.
Since in every recursive iteration $|\hat{\rec}\cap P_u|\leq |\rec'\cap P_u|-1$, the centroid shrink procedure is executed in $O(|P_u|)$ time.

%\paragraph{Remark} The fair split procedure is straightforward to be executed in the relational setting (more details in Subsection~\ref{}). However, we are not aware of a straightforward way to implement the centroid shrink in the relational setting. A straightforward implementation takes $\Omega(N+|P_u|)$ time, however in the relational setting $|P_u|=O(|\Q(\I)|)$. In the next subsection, we design a new variation of the BBD tree which will allow us to execute the centroid shrink efficiently, as we will show in Subsection~\ref{}.

\vspace{-0.7em}
\subsection{Construction of the RBBD tree}
\label{subsec:constructionRBBD}
\vspace{-0.2em}
We define a new variation of the BBD tree called RBBD tree in the standard computational setting over a set of $n$ points $P\subset \Re^d$. %We note that the RBBD tree does not have any advantage compared to the BBD tree in the standard computational setting. Instead, its only advantage is that given a node of the RBBD tree $u$, we will be able to construct the children of $u$ in the relational setting, efficiently, as shown in Subsection~\ref{}.
The construction of the RBBD tree shares most of the steps and properties with the BBD tree. For example, all constructed boxes (both inner and outer) are midpoint boxes. The fair split is identical to BBD tree's fair split. The only difference is the definition and execution of the centroid shrink procedure, given a node $u$.

We define the \emph{randomized centroid shrink}.
%Let $t>1$ be any constant.
Let $u$ be a node.
We get a set $\net_u$ of $\Theta(\eps^{-2}\log n)$ uniform random samples from $\square_u\cap P_u$. 
If $\square_u$ is a box without a hole, then we simply sample in the box $\square_u$.
If $\square_u$ is a box with a hole, we sample in the region between $\square_u^+$ and $\square_u^-$. Then the procedure is similar to the centroid shrink from the BBD tree using $\net_u$ instead of $P_u$. We recursively compute the smallest midpoint box $\rho'$ that contains $\net_u$, and we apply $O(d)$ fair splits until we find a non-trivial fair split. We recursively continue the construction on the box $\hat{\rho}$ created by the non-trivial fair split procedure that contains most of the points in $\net_u$ (majority midpoint box). We repeat this recursive approach until we find a majority midpoint box $\hat{\rho}$ (after a non-trivial fair split) that contains at most $\frac{2}{3}|\net_u|$ samples from $\net_u$.
If $\square_u$ is a box with a hole, the procedure is shown in Appendix~\ref{appndx:RBBDhole}.
%During this procedure, we handle the inner box $\square_u^-$ as in the centroid shrink of the standard BBD tree: The first time we find a box $\hat{\rec}$ (after the non-trivial fair split) that does not contain $\square_u^-$ we construct the subtree with the six new nodes as described in the previous subsection.
In Appendix~\ref{appndx:constructionRBBD}, we prove the next theorem using the properties of the BBD and RBBD trees.
\vspace{-0.2em}

\begin{theorem}
\label{thm:RBBD}
    Given a set $P$ of $n$ points in $\Re^d$ a parameter $\eps\in (0,1)$, and a constant parameter $\cons>1$, the RBBD tree has height $O(\log n)$, it can be constructed in $O(n\log n + n\cdot(f(n,\eps)+\eps^{-2}\log n))$ time with space $O(n)$, where $f(n,\eps)$ denotes the running time to sample $\Theta(\eps^{-2}\log n^\cons)$ points from a box (or a box with a hole), such that given a query ball $\ball(p,r)$ for a point $p\in \Re^d$ and a radius $r$, a set $\canonical$ of $O(\log n +\eps^{-d+1})$ canonical nodes can be found in $O(\log n +\eps^{-d+1})$ time such that $\ball(p,r)\subseteq P\cap \bigcup_{u\in \canonical}\square_u \subseteq \ball(p,(1+\eps)r)$. The height, the space, the construction time of the tree, the size of $\canonical$, and the query time hold with probability at least $1-\frac{1}{n^{\cons-1}}$.
\end{theorem}

\vspace{-0.6em}

\subsection{RBBD tree in the relational setting}
\label{subsec:RBBDrel}
\vspace{-0.2em}
Let $\Q$ be an acyclic join query over a database instance $\I$.
%Recall that in this setting, the set of points $P$ where we construct nodes from the RBBD tree over $P$ is $\Q(\I)$, i.e., $P=\Q(\I)$. 
Due to the size of $\Q(\I)$, we cannot construct the entire RBBD tree as described in the previous subsection. Instead, we show how we can construct the children of a node $u$, %(along with their associated boxes or boxes with holes), 
given that $u$ has been constructed and it is a node of an RBBD tree over $\Q(\I)$.
More specifically, we show how we can implement the fair split and the randomized centroid shrink of the RBBD tree efficiently in the relational setting.
%Instead of only constructing the new nodes in the RBBD tree, we will also store all needed information that we use in the subsequent sections. In particular, 
Furthermore, for every new node $v$ of the RBBD tree we show how to compute the count $v.c=|\square_v\cap \Q(\I)|$ and a representative $v.rep\in \square_v\cap \Q(\I)$.
This information is useful for the design of our algorithms.

\paragraph{Root node} We first show how to construct the root node $\root$ of the RBBD tree, which is always a box without a hole. For every attribute $A_i\in \allattr$ we compute $\pi_{A_i}(\I)$ and we keep the maximum $M_i$ and minimum value $m_i$. We define the $\square_\root$ as a box without a hole with opposite corners $(m_1,\ldots, m_d)$ and $(M_1,\ldots, M_d)$. This procedure is executed straightforwardly in $O(N)$ time.

\paragraph{Fair split} 
The implementation of the fair split in the relational setting is identical to the fair split in the standard computational setting. We only need to identify the longest side of a box along with its center in $O(1)$ time.
%Let $u$ be a node of the RBBD tree over $\Q(\I)$. The implementation of the fair split in the relational setting is identical with the fair split in the standard computational setting. Without loss of generality assume that $\square_u$ is a box with a hole. Using the coordinates of the opposite corners we identify the longest side of $\square_u^+$.
%Let $h$ be the hyperplane orthogonal to the longest side of $\square_u^+$ that passes through the center of $\square_u^+$. We construct two children of $u$, namely $u_1, u_2$. The box in node $u_1$ is defined as $\square_{u_1}=\square_u\cap \{x\in \Re^d\mid h(x)< 0\}$ and the box in node $u_2$ is defined as $\square_{u_2}=\square_u\cap\{x\in \Re^d\mid h(x)\geq 0\}$. This is exactly the implementation of the fair split as described in Subsection~\ref{}, so the correctness follows. We need $O(1)$ to identify the longest side of $\square_u^+$, $O(1)$ time to construct the hyperplane $h$, and $O(1)$ time to define he nodes $u_1$ and $u_2$. Overall, the fair split is executed in $O(1)$ time.

\paragraph{Randomized centroid shrink} We only need to show how to implement the sampling procedure since all the other steps can be executed similarly to the standard computational setting.
We need to sample uniformly at random $\Lambda=\Theta(\eps^{-2}\log(N^{m+4}))=\Theta(\eps^{-2}\log N)$ points from $\square_u\cap \Q(\I)$.
If $\square_u$ is a box (without a hole) then we simply call $\mathsf{SampleRect}(\Q,\I,\square_u,\Lambda)$ from Lemma~\ref{lem:Rects} and we get the set $\net_u$ in $O(N+\eps^{-2}\log^2 N)$ time. 
%By definition $\net_u$ is an $\eps$-sample
In Appendix~\ref{appndx:randcentshr}, we show how we sample in the same running time from $\square_u\cap \Q(\I)$, when $\square_u$ is a box with a hole.

\paragraph{Count and representative tuples in newly constructed nodes} 
If $\square_u$ is a box without a hole, we call $u.c=\mathsf{CountRect}(\Q, \I,\square_u)$ and $u.rep=\mathsf{ReprRect}(\Q,\I,\square_u)$ from Lemma~\ref{lem:Rects}. Both of these procedures run in $O(N)$ time. %with high probability.
%Using the procedures $\mathsf{CountRect}$ from Lemma~\ref{}, to compute $u.\mathsf{count}=\mathsf{CountRect}(\Q, \I,\square_u)$. Furthermore, we use $\mathsf{ReprRect}$ from Lemma~\ref{} to compute $u.rep=\mathsf{ReprRect}(\Q,\I,\square_u)$. Both $u.count$ and $u.rep$ for a node $u$ are computed in $O(N)$ time with high probability, using the results in Lemma~\ref{}
In Appendix~\ref{appndx:RBBDcount}, we compute the count and the representative point in the same running time, when $\square_u$ is a box with a hole.

%We call the procedure of executing a fair split and computing the count and the representative tuple for every new constructed node (generated by the fair split procedure) as \emph{enhanced fair split}. Equivalently, we call the procedure of executing a randomized centroid shrink and computing the count and the representative tuple for every new constructed node (generated by the randomized centroid shrink procedure) as \emph{enhanced randomized centroid shrink}.
%\vspace{-0.em}

\begin{lemma}\label{thm:RBBDrel}
Let $\Q$ be an acyclic join query over a database instance $\I$ of size $N$. Given a node $u$ of an RBBD tree over $\Q(\I)$, the fair split procedure is executed in $O(N)$ time, while the randomized centroid shrink procedure is executed in $O(N+\eps^{-2}\log^2 N)$ time.
%, with probability at least $1-\frac{1}{N^{m+5}}$.
\end{lemma}

\vspace{-1em}

\newcommand{\inactive}{\mathsf{inactive}}
\newcommand{\sample}{\mathsf{sample}}
\newcommand{\coun}{\mathsf{count}}
\newcommand{\representative}{\mathsf{rep}}
\newcommand{\report}{\mathsf{report}}
\subsection{Relational 
oracles using a partially constructed RBBD tree}
\label{sec:reloracles}
\vspace{-0.3em}
Using Lemma~\ref{thm:RBBDrel} and Theorem~\ref{thm:RBBD}, we derive new oracles that will allow us to solve optimization problems in the relational setting, efficiently.
Let $\tree$ be a partially constructed RBBD tree over $\Q(\I)$, i.e., $\tree$ consists of a connected subset of nodes (including the root, denoted by $\root$) of an RBBD tree over $\Q(\I)$. Initially, $\tree$ consists only of the root, as constructed in Section~\ref{subsec:RBBDrel}.
For every constructed node $u\in \tree$, we store i) $\square_u$ the region corresponding to node $u$ which is either a box or a box with a hole, and ii) $u.a$ a boolean variable which is $1$ if $u$ is an \emph{active} node, and $0$ if it is an \emph{inactive} node. A point $p\in \Q(\I)$ is called an \emph{inactive point} if there exists a node $u\in \tree$ such that $p\in \square_u$ and $u.a=0$. A point which is not inactive is \emph{active}.
While we do not compute it explicitly, for convenience we use $Q'\subseteq \Q(\I)$ to denote the set of current active points in $\Q(\I)$.
Furthermore, in each node $u\in \tree$, we store $u.rep$ which is an arbitrary representative point from $\square_u\cap Q'$. Finally, for each node $u\in \tree$, we store the number of active points in the region $\square_u$, i.e., $u.c=|\square_u\cap Q'|$.

Given such a partially constructed tree $\tree$, we show how to implement five oracles in the relational setting while maintaining $\tree$ (new nodes might need to be constructed).
%Every oracle in this section works correctly with probability at least $1-\frac{1}{N^{m+2}}$. Furthermore, 
The running time of every oracle in this section holds with probability at least $1-\frac{1}{N^{m+3}}$. Due to space requirements, we show the proofs of all lemmas in Appendix~\ref{appndx:reloracles}.

\paragraph{Make points in a ball inactive}
Given a ball $\ball(x,r)$, where $x\in\Re^d$ and $r$ is a non-negative real number, the goal is to (implicitly) make all points in $\ball(x,r)\cap Q'$ inactive.
Let $\tree.\inactive(x,r)$ be this oracle, which works as follows.
We first execute a query procedure similar to the query procedure on a BBD tree. Let $\canonical(x,r)$ be the set of canonical nodes we will compute. 
Let $u$ be a node of $\tree$ that the query procedure processes. Initially, $u=\root$ and $\canonical(x,r)=\emptyset$.
If $u.a=0$, we stop the execution of the query through this branch of the tree. In the rest, we assume that $u.a=1$.
If the children of $u$ have not been constructed in $\tree$ we use Lemma~\ref{thm:RBBDrel} to construct its children.
Next, we continue checking whether $\square_u$ intersects $\ball(x,r)$. If $\square_u\subseteq \ball(x,(1+\epsilon)r)$, then we add $u$ in $\canonical(x,r)$ and we stop the execution through this branch of the tree. If $\square_u\cap \ball(x,r)=\emptyset$, we stop the execution of the query through this branch of the tree. In any other case, we continue the query procedure in both children of $u$ recursively.
%In the end of this process we have computed a set of canonical nodes $\canonical(x,r)$. 
For every added node $u\in \canonical(x,r)$, we set $u.a=0$. Before we finish the procedure, we need to update the representative points and the counts of the ancestor nodes of $\canonical(x,r)$.
%, similarly to what we did for the BBD trees in Section~\ref{}. 
Starting from the parent node $v$ of the deepest node in $\canonical(x,r)$, we run recursively: Let $u, w$ be $v$'s two children. If $u.a=w.a=0$, then we set $v.a=0$. If $u.a=1$ and $w.a=0$ (resp. $u.a=0$ and $w.a=1$), we set $v.c=u.c$ and $v.rep=u.rep$ (resp. $v.c=w.c$ and $v.rep=w.rep$). If $u.a=w.a=1$, we set $v.c=u.c+w.c$ and $v.rep=u.rep$. After we process all canonical nodes' ancestors in one level of the tree, we continue the update of the nodes in the level above.

%We prove the next lemma.
\vspace{-0.5em}
\begin{lemma}
\label{lem:oracle1}
    $\tree.\inactive(x,r)$ makes all points in $\ball(x,r)\cap Q'$ inactive and might make some points in $\ball(x,(1+\eps)r)\cap Q'$ inactive. Furthermore, in the end of the execution of the oracle, for every node $u\in \tree$ the values $u.rep, u.c$ are correct. The running time is $O((\eps^{-d+1}+\log N)(N+\eps^{-2}\log^2 N))$.
\end{lemma}
\vspace{-0.5em}
\paragraph{Approximately count active points in a ball}
Given a ball $\ball(x,r)$ the goal is to compute $|\ball(x,r)\cap Q'|$ approximately. We define the oracle
$\tree.\coun(x,r)$, and run the same query procedure as in $\tree.\inactive(x,r)$ to compute $\canonical(x,r)$. We return $\sum_{u\in \canonical(x,r)}u.c$.
\begin{lemma}
\label{lem:oracle2}
    $\tree.\coun(x,r)$ returns a number $c$ such that $|\ball(x,r)\cap Q'|\leq c\leq |\ball(x,(1+\eps)r)\cap Q'|$ in  $O((\eps^{-d+1}+\log N)(N+\eps^{-2}\log^2 N))$ time.
\end{lemma}

\vspace{-0.5em}
\paragraph{Get a representative from $Q'$}
Next, the goal is to return any point from $Q'$. We define the oracle $\tree.\representative()$. If $\root.a=0$ then we return $\emptyset$, i.e., $Q'=\emptyset$. If $\root.a=1$, then we return $\root.rep$.
\vspace{-0.2em}
\begin{lemma}
\label{lem:oracle3}
$\tree.\representative()$ returns a point in $Q'$ if $Q' \neq \emptyset$, and $\emptyset$ otherwise, in $O(1)$ time.
\end{lemma}

\vspace{-0.5em}
\paragraph{Report all points in $Q'$}
Assume that we aim to report all points in $Q'$. We define the oracle $\tree.\report()$. Starting from the root node $\root$ we execute a depth-first search over $\tree$. If we reach a node $u$ such that $u.a=0$ we stop the traversal through this branch of the tree. If we reach a node $u$ that we have not constructed its children or $u$ is a leaf node of the RBBD tree, we run $\mathsf{ReportRect}(\Q,\I,\square_u)$. We continue the depth-first search until we visit all constructed nodes in $\tree$.
\vspace{-0.2em}
\begin{lemma}
\label{lem:oracle4}
   $\tree.\report()$ returns $Q'$ in time $O(|\tree| N + |Q'|)$, where $|\tree|$ is the number of nodes in $\tree$.
\end{lemma}
\vspace{-0.5em}

\paragraph{Sample $z$ points from $Q'$}
The goal is to sample (with replacement) $z$ points from $Q'$, uniformly at random.
We describe the oracle $\tree.\sample(z)$. We use $\tree$ to get each sample one by one.
If $\root.a=0$ we return no sample because $Q'=\emptyset$.
%We start from the node $\root$. While visiting a node $u\in \tree$,
%Let $u$ be a node of $\tree$, initially $u=\root$.
%if $u.a=0$ we return no sample because $Q'=\emptyset$.
Otherwise we set $u=\root$ and we act as follows.
If $u$ has two active children $v, w$ then we move to the node $v$ with probability $\frac{v.c}{v.c+w.c}$ and to node $w$ with probability $\frac{w.c}{v.c+w.c}$.
If $v.a=1$ and $w.a=0$ (resp. $v.a=0$ and $w.a=1$) we move to the node $v$ (resp. $w$). We recursively repeat this process until we reach a leaf node of the RBBD tree or an active node whose children have not been constructed in $\tree$. If we have reached a leaf node, then we trivially add the point stored in the leaf node to the sample set. In case we reach a node $u$ whose children have not been constructed in $\tree$, we sample one point from $\square_u\cap \Q(\I)$ as described in Section~\ref{subsec:RBBDrel}.
\vspace{-0.3em}
\begin{lemma}
\label{lem:oracle5}
    $\tree.\sample(z)$ returns $z$ samples uniformly at random from $Q'$ in $O(z\cdot N)$ time.
\end{lemma}
\vspace{-0.3em}
%Given $x\in \Re^d$ and $r\in \Re$, the oracle above can be extended to return $z$ samples from $\ball(x,r)\cap Q'$ in $O((\eps^{-d+1}+\log N)(\eps^{-2}\log^2 N)+zN)$ time. Even though we do not use this oracle in the next sections, it might be useful in other applications.

%% file: bbd-figs.tex
\begin{figure}[t]
    \centering
    % First image with its own caption
    \begin{minipage}{0.20\textwidth}
        \centering
        \vspace{0mm} % Move the image down for alignment
        \includegraphics[width=\linewidth]{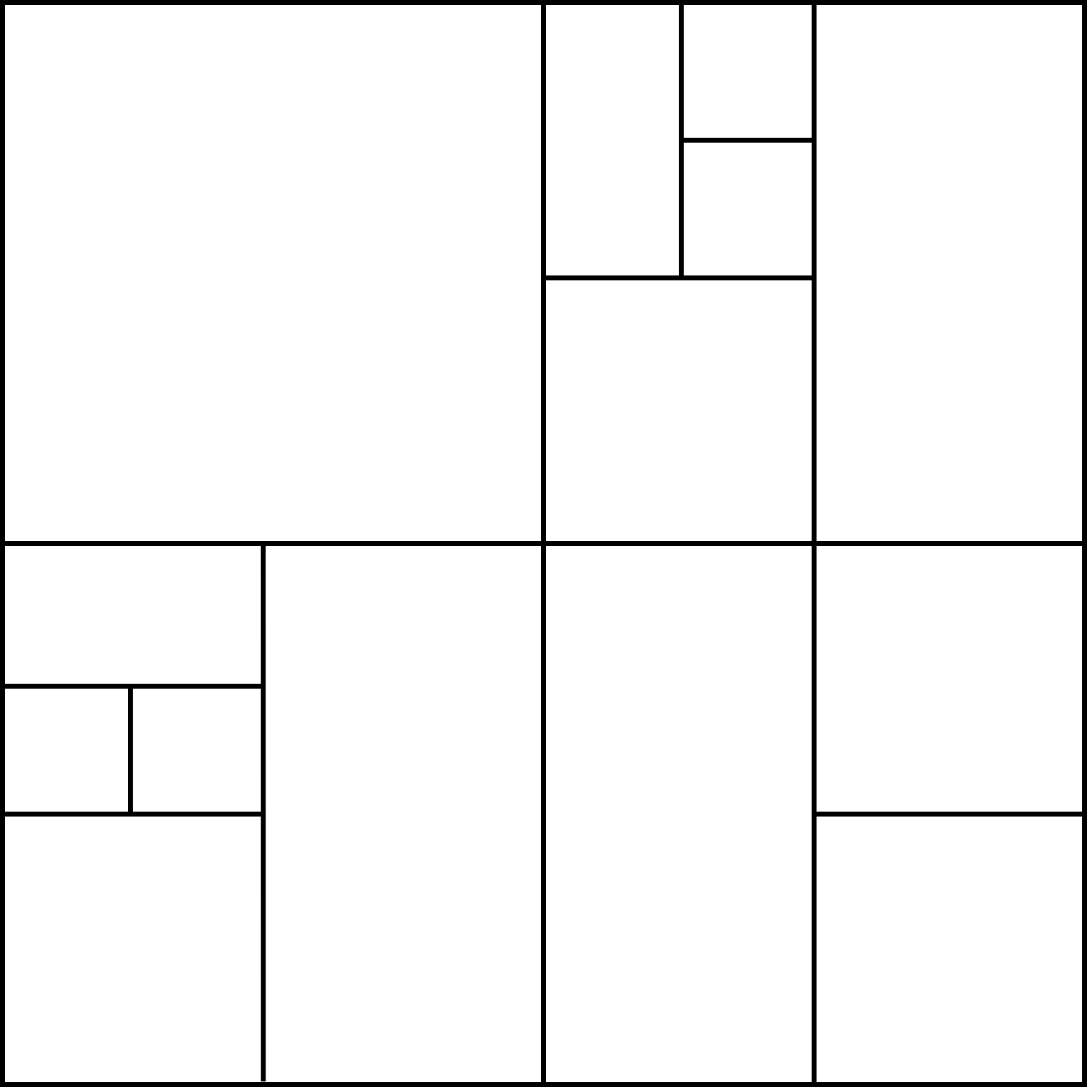}
        \vspace{-3mm} % Space between image and caption
        \caption{All the cells are midpoint boxes.}
        \label{fig:midpoint}
    \end{minipage}
    \hspace{0.04\textwidth}
    % Group of three images with spacing
    \begin{minipage}{0.70\textwidth}
        \centering
        \vspace{-3.8mm}
        \includegraphics[width=0.29\linewidth]{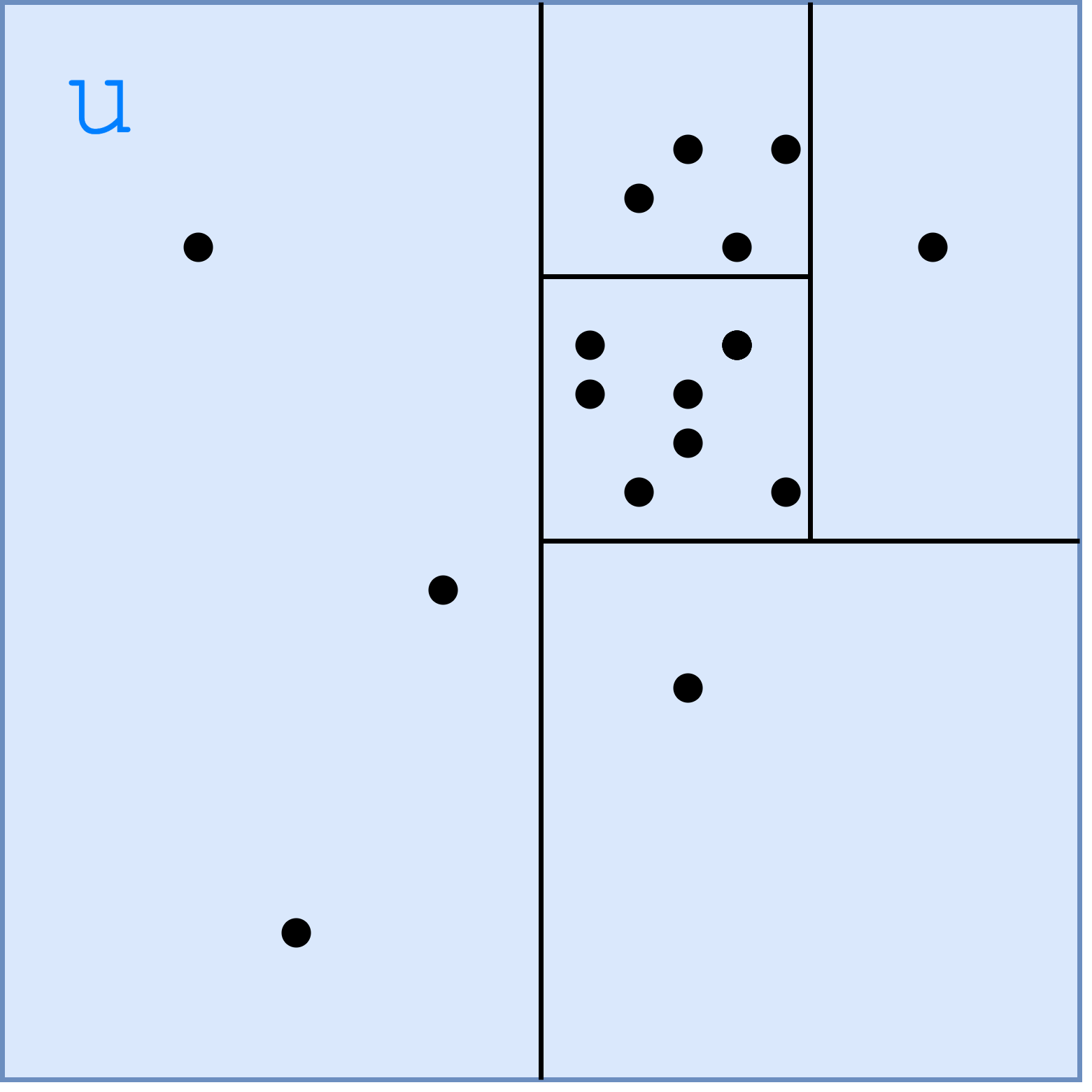}
        \hspace{0.04\linewidth}
        \includegraphics[width=0.29\linewidth]{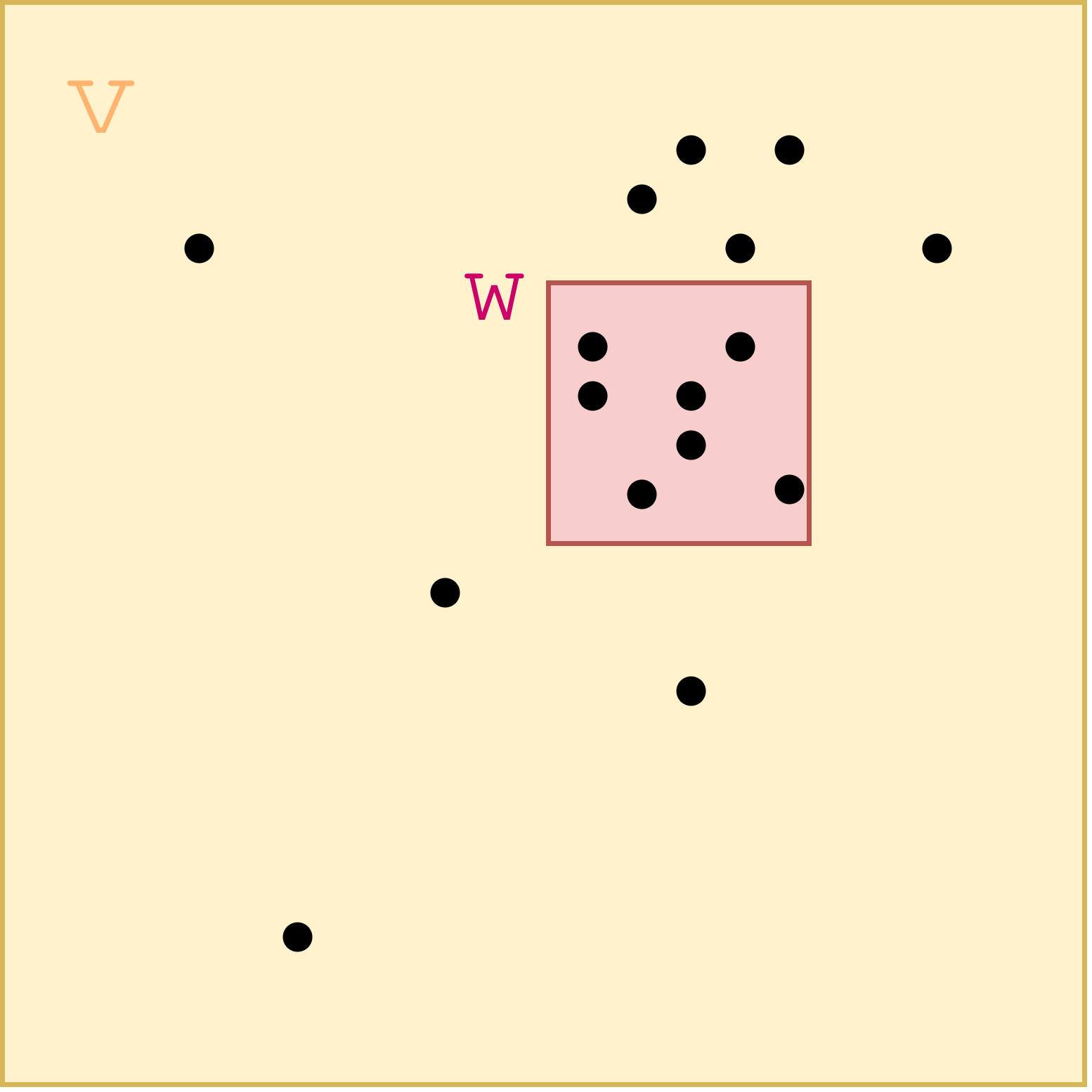}
        \hspace{0.04\linewidth}
        \includegraphics[width=0.29\linewidth]{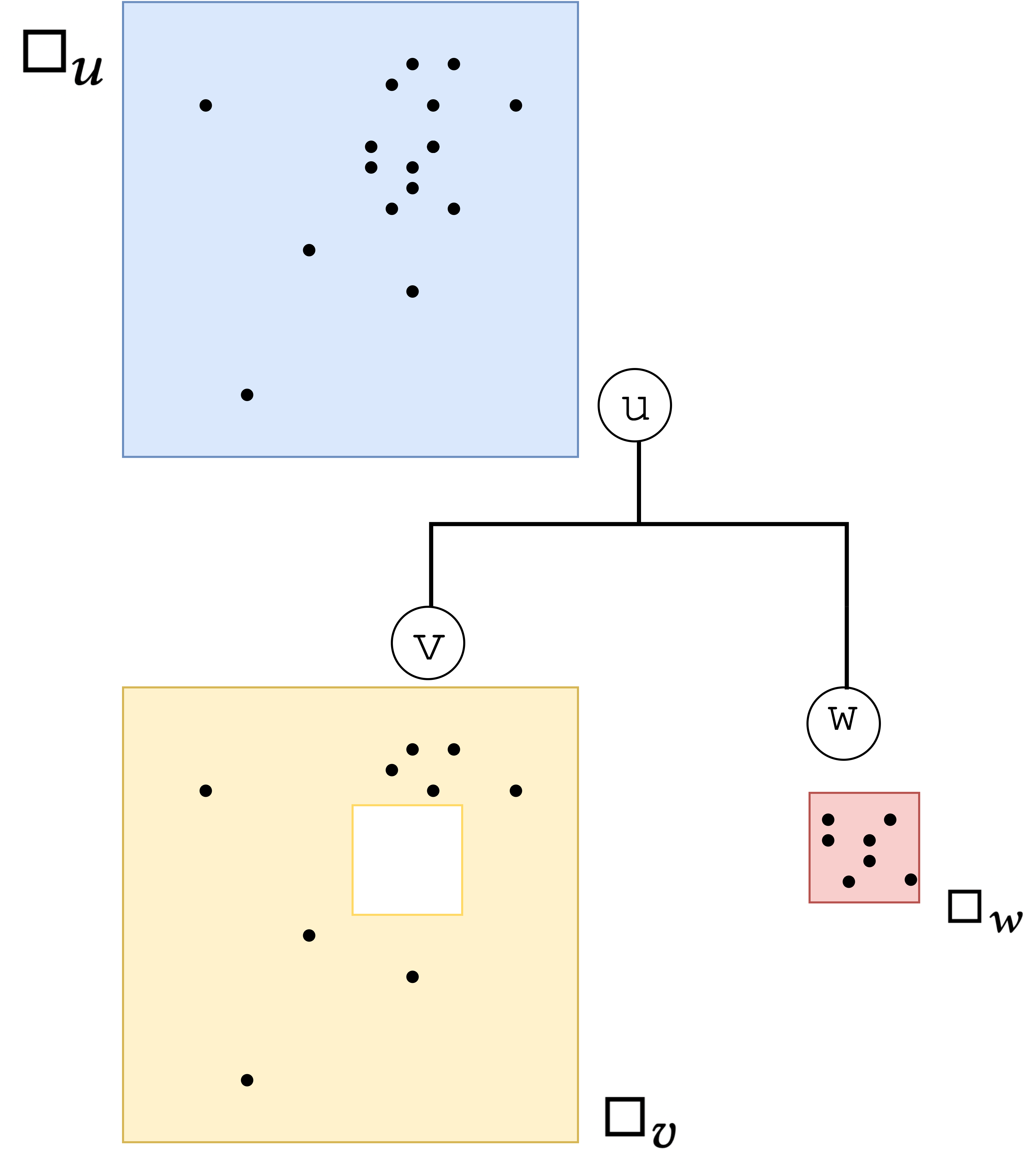}
        \vspace{1mm} % Space between images and caption
        \caption{Illustration of finding the children $v, w$ of a node $u$ by the centroid shrink process. The region $\square_v$ is a box with a hole while $\square_w$ is a box.}
        \label{fig:shrink}
    \end{minipage}
    \vspace{-1.5em}
\end{figure}

%% file: kcenter.tex
\vspace{-3.0mm}
\section{$k$-Center Clustering}
\label{sec:kCenter}
\vspace{-0.5mm}
In Section~\ref{sec:overview}, we already discussed the high-level idea of solving the relational $k$-center clustering efficiently. Here, we describe all the details and the formal proofs.
\vspace{-3.0mm}
\subsection{A constant approximation algorithm}
\vspace{-0.5mm}
We first describe the main part of the algorithm, where given a distance $r$, if $r$ is sufficiently large, it returns a set $\ret$ such that $\kcenter_{\ret}(\Q(\I))\leq (2+\eps)r$. Later, we run a binary search over the distances in $\Q(\I)$ to get a constant approximation for the relational $k$-center clustering.

%We run a binary search on the pairwise distances of the points in $\Q(\I)$. Obviously, we cannot compute all pairwise distances in $\Q(\I)$, instead, we only access the $j$-th smallest distance for an index $j$ determined by the binary search. We show later how we can access the $j$-th smallest distance efficiently. Let $r$ be the distance currently checked by the binary search.

\newcommand{\act}{a}
\paragraph{Main algorithm} Let $\tree$ consist of the root node $\root$ of an RBBD tree over $\Q(\I)$ as shown in Subsection~\ref{subsec:RBBDrel}.
%In addition to storing a representative point for every node $u$ of the RBBD tree, we also store a boolean variable $u.\act$, initially set to $1$, representing whether a node $u$ is \emph{active}. % (in this case we do not need to store the count in every node of the RBBD tree).
Initially, we set $\ret=\emptyset$. 
We run the following steps for at most $k$ iterations.
In the $i$-th iteration, we set $p_i=\tree.\representative()$ and we add $p_i$ in $\ret$. 
Then we call $\tree.\inactive(p_i,2r)$.
%We define the ball $\ball(p_i, 2r)$ and we run the query procedure $\tree(p_i,2r)$. Each time the query procedure visits a node $u$ of $\tree$ whose children have not been constructed, we execute the algorithm from Subsection~\ref{} to construct the children of $u$. Then, we check whether $u.\act=1$. If yes, then we continue the query search as described in Section~\ref{}. If $u.\act=0$ then we stop the search towards this branch of the tree. In the end, we get the set of active canonical nodes $\canonical(p_i,2r)$. For every $u\in \canonical(p_i,2r)$, we set $u.\act=0$. Then we traverse bottom-up all nodes that the query procedure $\tree(p_i,2r)$ visited and update their representative points so that every active node has a representative point that does not belong in the leaf nodes of a subtree rooted at an inactive node. The procedure is identical to what we described in Section~\ref{}.
%For example, for a node $u\in\canonical(p_i,2r)$ (so $u.\act=0$) let $v$ be its parent and $w$ is its sibling node where $w.\act=1$. When we traverse $v$ we set $v.rep=w.rep$. If both $w$ and $u$ are inactive, we set $v.\act=0$.
If after $k$ iterations $\tree.\representative()\neq \emptyset$ then 
we return an empty set $\ret=\emptyset$.
%we continue the binary search for larger values of $r$. 
On the other hand, if
$\tree.\representative()= \emptyset$
within the $k$ iterations, then we return $\ret$.
\revthree{The pseudocode of the algorithm is shown in Algorithm~\ref{alg:RelKCenterApprox}.}
%we store the solution $\ret$ and we continue the binary search for smaller values of $r$.
%In the end, we return the last computed set $\ret$ by the binary search.

\paragraph{Analysis}
If $r\geq \kcenter_{\opt(\Q(\I))}(\Q(\I))$, we show that the main algorithm returns a non-empty set $\ret$. 
From Lemma~\ref{lem:oracle1}, for every new point $p_i$ all (active) points in $\Q(\I)$ within distance $2r$ become inactive (along with some points within distance $2(1+\eps)r$) by $\tree.\inactive(p_i,2r)$.
%Indeed, as discussed in Section~\ref{sec:overview}, for every new point $p_i$ all points in $\Q(\I)$ within distance $2r$ become inactive (along with some points within distance $2(1+\eps)r$).
%, with probability at least $1-\frac{1}{N^{m+2}}$. 
Hence, in each iteration, all the points of at least one optimum cluster become inactive, as discussed in Section~\ref{sec:overview}. So after $k$ iterations all points in $\Q(\I)$ are inactive (i.e., $\tree.rep()=\emptyset$) and $\ret\neq \emptyset$.
%, with probability at least $1-\frac{1}{N^2}$.
Hence, $\kcenter_{\ret}(\Q(\I))\leq 2(1+\eps)r$. By initially setting $\eps\leftarrow\eps/2$, we have that if $r\geq \kcenter_{\opt(\Q(\I))}(\I)$ then $\kcenter_\ret(\Q(\I))\leq (2+\eps)r$.

%At this point, we conduct the analysis assuming that in every step, the binary search returns the correct distance to check. The algorithm described above has a $(2+\eps)$-approximation factor with high probability, by the discussion in Subsection~\ref{} and~\ref{}. If $r=\kcenter_{\opt(\Q(\I))}(\Q(\I))$, then the ball $\ball(p_i,2r)$ in the $i$-th iteration completely covers all points from at least one optimum cluster in $\opt(\Q(\I))$.  Similarly, the oracle $\tree.\inactive(p_i,2r)$ makes inactive all active points in $\Q(\I)$ that lie in ball $\ball(p_i,2r)$, with high probability. Since $\tree.\inactive(p_i,2r)$ makes also inactive some active points in $\Q(\I)$ that lie in the ball $\ball(p_i,2(1+\eps)r)$, we conclude that $\kcenter_\ret(\Q(\I))\leq 2(1+\eps)\kcenter_{\opt(\Q(\I))}(\Q(\I))$ with high probability. By initially setting $\eps\leftarrow\eps/2$, we conclude that $\kcenter_\ret(\Q(\I))\leq (2+\eps)\kcenter_{\opt(\Q(\I))}(\Q(\I))$.

%For the running time, the binary search checks $O(\log|\Q(\I)|^2) = O(\log N)$ distances. For every distance 

We run the $\tree.\inactive(\cdot,\cdot)$ and the $\tree.\representative()$ oracles $O(k)$ times. Each oracle runs in $O((\eps^{-d+1}+\log N)(N+\eps^{-2}\log^2 N))$ time, with probability at least $1-\frac{1}{N^{m+3}}$, as shown in Section~\ref{sec:reloracles}.
%For constant $\eps$, given that we can efficiently run a binary search over the pairwise distances in $\Q(\I)$, our algorithm runs in
%$$O\left(k\left(\eps^{-d+1}+\log N\right)\left(N+\eps^{-2}\log^2 N\right)\log N\right)=O\left(k(N\eps^{-d+1}+\eps^{-d-1}\log^2 N + N\log N +\eps^{-2}\log^3 N)\log N\right)$$
Overall, %if $\eps\in(0,1)$ is an arbitrarily small constant
the main algorithm runs in $O\left(k\left(\eps^{-d+1}+\log N\right)\left(N+\eps^{-2}\log^2 N\right)\right)=O(k\cdot N\cdot\log N+k\cdot \log^3 N)=\O\left(k\cdot N\right)$ time assuming that $\eps\in(0,1)$ is an arbitrary constant,
%\begin{equation}
%\label{eq:runtime}
%    O\left(k\left(\eps^{-d+1}+\log N\right)\left(N+\eps^{-2}\log^2 N\right)\right)=O(k\cdot N\cdot\log N+k\cdot \log^3 N)=\O\left(k\cdot N\right)
%\end{equation}
 with probability at least $1-\frac{1}{N^{2}}$.

\paragraph{Complete algorithm} 
Assume that we run a binary search on the pairwise distances of $\Q(\I)$. For each distance $r$ selected by the binary search we run the main algorithm. If it returns a non-empty set $\ret$ then we continue the binary search for smaller values of $r$. Otherwise, we continue with larger values of $r$. In the end, the last computed non-empty set $\ret$ from the binary search is returned as the final set of centers.
From the analysis of the main algorithm along with the discussion in Section~\ref{sec:overview}, such algorithm returns a set $\ret$ such that $\kcenter_{\ret}(\Q(\I))\leq (2+\eps)\kcenter_{\opt(\Q(\I))}(\Q(\I))$.
Unfortunately, we cannot sort all pairwise distances in $\Q(\I)$ and then run a binary search on them because we would need $\Omega(|\Q(\I)|^2)$ time. While it is not straightforward to run a binary search on $\ell_2$ pairwise distances, the key idea is to run a binary search over the $\ell_\infty$ pairwise distances of points in $\Q(\I)$.
We show the details of the binary search in Appendix~\ref{appndx:kCenter}.
Using the fact that for any pair of points $t, p\in \Re^d$, $||t-p||_\infty \leq ||t-p||_2\leq \sqrt{d}||t-p||_\infty$, we prove the next theorem in Appendix~\ref{appndx:kCenter}.

\vspace{-0.1em}
\begin{theorem}
\label{thm:constantkCenter}
        Given an acyclic join query $\Q$ over $d$ attributes, a database instance $\I$ of size $N$, a parameter $k$, and an arbitrarily small constant parameter $\eps\in(0,1)$, there exists an algorithm that computes a set $\ret\subseteq \Q(\I)$ of $k$ points and a real number $r_\ret$ such that $\kcenter_{\ret}(\Q(\I))\leq r_\ret\leq  (2\sqrt{d}+\eps)\kcenter_{\opt(\Q(\I))}(\Q(\I))$.
        %, with probability at least $1-\frac{1}{N}$.
        The running time of the algorithm is 
    $O(kN\log^2 N + k\log^4 N)$ with probability at least $1-\frac{1}{N}$.
\end{theorem}
\vspace{-0.1em}
Notice that our algorithm from Theorem~\ref{thm:constantkCenter} returns an estimation $r_\ret$ of $\kcenter_{\ret}(\Q(\I))$.
This is useful 
%in order to get a better approximation algorithm for the relational $k$-center clustering, as shown in Section~\ref{}. Furthermore, $r_\ret$ is useful 
in order to derive a better approximation for relational $k$-center clustering in Section~\ref{subsec:betterkcenter} and a constant approximation for the relational $k$-median/means clustering as shown in Section~\ref{sec:kmeans}.
\vspace{-0.7em}
\subsection{Better approximation algorithm}
\label{subsec:betterkcenter}
\vspace{-0.2em}
Using our result from Theorem~\ref{thm:constantkCenter}, along with the coreset construction from~\cite{agarwal2024computing}, we get the fastest known $(2+\eps)$-approximation algorithm for the relational $k$-center clustering problem. First, we execute the algorithm from Theorem~\ref{thm:constantkCenter} and we get $\ret$ and $r_\ret$. 
Both $\ret$ and $r_\ret$ are given as input to the coreset construction in~\cite{agarwal2024computing} and we get the main result of this section.
%Since $\ret$ is an $O(1)$-approximation solution to the relational $k$-center clustering and $r_\ret$ satisfies $\kcenter_\ret(\Q(\D))\leq r_\ret\leq O(1)\kcenter_{\opt(\Q(\I))}(\Q(\I))$, a coreset $\coreset\subseteq \Q(\I)$ with $|\coreset|=O(k\eps^{-d})$ for the relational $k$-center clustering can be constructed in $O(k\eps^{-d}N)$ time as shown in~\cite{}. Finally, a $\gamma$-approximation algorithm $\DkcenterAlg$ is executed on the coreset $\coreset$ to derive a $(1+\eps)\gamma$-approximation solution in $O(\timeCenter_\gamma(k))$ additional time.

\vspace{-0.1em}
%Combining everything, we conclude with the main result of this section.
\begin{theorem}
\label{thm:kCenter}
        Given an acyclic join query $\Q$ over $d$ attributes, a database instance $\I$ of size $N$, a parameter $k$, and an arbitrarily small constant parameter $\eps\in(0,1)$, there exists an algorithm that computes a set $\ret\subseteq \Q(\I)$ of $k$ points such that $\kcenter_{\ret}(\Q(\I))\leq  (1+\eps)\gamma\kcenter_{\opt(\Q(\I))}(\Q(\I))$.
        %, with probability at least $1-\frac{1}{N}$.
        The running time of the algorithm is 
    $O(kN\log^2 N + k\log^4 N + \timeCenter_\gamma(k))$ with probability at least $1-\frac{1}{N}$.
\end{theorem}
\vspace{-0.1em}
For example, if Feder and Greene algorithm~\cite{feder1988optimal} is executed on the coreset constructed by~\cite{agarwal2024computing} then we get a $2(1+\eps)$-approximation solution in $O(k\eps^{-d}\log k)$ additional time.
By setting $\eps\leftarrow \eps/2$, the approximation factor becomes $2+\eps$.
The overall asymptotic complexity for arbitrarily small constant $\eps\in(0,1)$ is $O(kN\log^2 N + k\log^4 N)$.
In Appendix~\ref{appndx:kCenter} we show an alternative $(2+\eps)$-approximation algorithm with the same running time without constructing the coreset from~\cite{agarwal2024computing}.

%% file: kMedianMeans.tex
\vspace{-0.6em}
\section{$k$-Median and $k$-Means Clustering}
\label{sec:kmeans}
\vspace{-0.3em}
We focus on the relational $k$-means clustering, while the algorithm for $k$-median follows almost verbatim.
In this section, we first describe a constant approximation algorithm returning $\O(k)$ centers.
%for the relational $k$-median and $k$-means clustering problems returning a set $\ret$ with slightly more than $k$ centers (Subsection~\ref{}).
Then, the returned set is given as input to the coreset construction in~\cite{esmailpour2024improved} to derive the final approximation algorithm and get exactly $k$ centers (Subsection~\ref{subsec:kmeans}).
Our new tools can also improve the coreset construction from~\cite{esmailpour2024improved}, leading to an even faster algorithm.
%Interestingly our RBBD tree construction in the relational setting can even accelerate the coreset construction from~\cite{}.
\vspace{-0.7em}
\subsection{Constant approximation algorithm with $\O(k)$ centers}
\vspace{-0.4em}
\paragraph{High level idea}
%As we described in Section~\ref{}, the state-of-the-art algorithms for the relational $k$-median/means clustering~\cite{} use the hierarchical aggregation method to propose an $O(1)$-approximation algorithm that returns a set of $\O(k^2)$ centers. 
We first show a constant approximation algorithm for the relational $k$-median/means problem on $\Q(\I)$ that runs in only $\O(kN)$ time, returning a set of $O(k\log^3 N)$ centers. 
We propose a modification of the algorithm in~\cite{har2004coresets} in the relational setting.
%get our intuition from the constant approximation algorithms for $k$-means/median in the standard computational setting~\cite{}. 
However, the algorithm of~\cite{har2004coresets} is heavily based on the properties of the standard computational setting, and it is not trivial how to extend it in the relational setting. We use our new tools for constructing an RBBD tree on--the--fly to implement an efficient constant approximation for the relational $k$-means/median problem.

We first execute our relational $k$-center algorithm on $\Q(\I)$ in order to get a polynomial upper and lower bound on the optimum $k$-means/median cost. At first, all the points in $\Q(\I)$ are considered active.
Then, the algorithm works in iterations. In the $i$-th iteration, we get a small set $X_i$ of $\O(k)$ samples from the active points. Next, tuples in $\Q(\I)$ that are close enough to $X_i$ are considered assigned to their closest point in $X_i$ and are implicitly removed (made inactive) from $\Q(\I)$. This procedure is repeated until every point in $\Q(\I)$ is removed. We note that some points that are assigned to the samples might pay a $k$-median/means cost much larger than the cost they pay in the optimum solution $\opt(\Q(\I))$. These points are called \emph{bad points}, while the rest points are called \emph{good points}. Similarly to~\cite{har2004coresets}, we show that the number of good points in each iteration is much larger than the number of bad points, allowing us to charge the cost of bad points to the good points. %Finally, it is inefficient to assign each tuple in $\Q(\I)$ to the closest center found by our algorithm in each iteration because of the number of tuples in $\Q(\I)$. 
In contrast to~\cite{har2004coresets}, where the authors assign each point (in the standard computational setting) to a center, we use our machinery for implementing an RBBD tree in the relational setting to implicitly assign groups of points in $\Q(\I)$ to approximately closest centers.

\vspace{-0.3em}
\paragraph{Algorithm}
First, we compute $|\Q(\I)|$ running the first phase of the Yannakakis algorithm~\cite{yannakakis1981algorithms}.
We construct the root node $\root$ of the RBBD tree $\tree$ on $\Q(\I)$ as described in Section~\ref{subsec:RBBDrel}.
%Every constructed node $u$ of $\tree$, we store i) $u.a$ which is $1$ if $u$ is active and $0$ if $u$ is inactive, ii) $u.\mathsf{count}$ number of active points from $\Q(\I)$ in the subtree rooted at $u$, and iii) $u.rep$ a representative point among the active points in $\Q(\I)$ that lie in $\square_u$. A point is active if its corresponding leaf does not have an inactive ancestor.
Let $\ret$ be the set of points we will return as a solution to $k$-means clustering, and let $r_\ret$ be the estimation of the cost $\kmeans_\ret(\Q(\I))$. Initially, we set $r_\ret=0$ and $\ret=\emptyset$.
Our algorithm runs in iterations.
Let $Q'_i\subseteq \Q(\I)$ be the set of active points in $\Q(\I)$ at the beginning of the $i$-th iteration. This set of points is not explicitly constructed. Instead, we only use $Q'_i$ to make our algorithm more intuitive. Initially, $Q'_1=\Q(\I)$. Furthermore, let $\tau_i=|Q'_i|$, so initially $\tau_1=|\Q(\I)|$. In contrast to $Q'_i$, $\tau_i$ will be explicitly stored in every step of the algorithm. A pseudocode of our method is shown in Appendix~\ref{appndx:kmeans}.
%For every step of the algorithm we first give the high level idea and then the execution in the relational setting.

$\bullet$ \textbf{Compute an estimation
of $\kmeans_{\opt(\Q(\I))}(\Q(\I))$.}
We run our algorithm for the relational $k$-center problem from Theorem~\ref{thm:kCenter}.
Let $V$ denote the set of points and $L$ the estimated cost returned by the $(2+\eps)$-approximation algorithm for the relational $k$-center clustering.
%our algorithm such as $\kcenter_{H}(\Q(\I))\leq L\leq (2+\eps)\kcenter_{\opt(\Q(\I))}(\Q(\I))$.
We can set $\eps$ to any constant value, say $0.1$. 
It is easy to verify that $\left(\frac{L}{2+\eps}\right)^2\leq \kmeans_{\opt(\Q(\I))}(\Q(\I))$ and $\kmeans_{\opt(\Q(\I))}(\Q(\I))\leq |\Q(\I)|\cdot L^2$, by definition. So, it holds that
$\frac{L^2}{9}\leq \kmeans_{\opt(\Q(\I))}(\Q(\I))\leq |\Q(\I)|\cdot L^2.$

After computing $L$, the algorithm runs in iterations. 
In iteration $i$, we first check whether $\tau_i\leq 480\cdot k\log^2 (|\Q(\I)|)$. If yes then we add in $\ret$ all remaining active points, i.e., $\ret=\ret\cup \tree.\report()$.

%This can be done by repeatedly visiting each active node $u$ of $\tree$ without an inactive descendant and reporting the points in $\square_u\cap \Q(\I)$ using the oracle in Lemma~\ref{}.

$\bullet$ \textbf{Get a set of $\Theta(k\log^{2}|\Q(\I)|)$ samples from $Q'_i$ uniformly at random.} In the $i$'th iteration, we get a set $X_i$ of $240\cdot k\log^2 |\Q(\I)|$ samples from $Q'_i$ uniformly at random. In particular, we set $X_i=\tree.\sample(240k\log^2 |\Q(\I)|)$ and we update $\ret=\ret\cup X_i$.

%We use $\tree$ to get each sample one by one. Starting from the root node $\root$ of $\tree$ we check the counts of each active child. Assume that both children of $\root$, say $v, w$ are active. We follow the path from node $v$ with probability $\frac{v.\mathsf{count}}{v.\mathsf{count}+w.\mathsf{count}}$ and path from node $w$ with probability $\frac{w.\mathsf{count}}{w.\mathsf{count}+v.\mathsf{count}}$. If $w$ (resp. $v$) was inactive we follow the path from node $v$ (resp. $w$). We repeat this process until we reach a leaf node or an active node whose children have not been constructed. If we have reached a leaf node then we trivially add the sampled point in $X$. In case we reached a node $u$ whose children are not constructed, we sample one point from $\square_u\cap \Q(\I)$ as described in Section~\ref{}.

%Next, the goal is to implicitly partition $Q'$ into the following way. Let $$Q'[a,b]=\{p\in Q'\mid \exists x\in X: x\in\widebar{\canonical}(X,b) \text{ and } \nexists x'\in X: x'\in\widebar{\canonical}(X,a)\},$$ where $\widebar{\canonical}(X,b)=\bigcup_{x\in X}\bigcup_{u\in \canonical(x,b)}Q'_u$.

$\bullet$ \textbf{Partition $Q'_i$ based on their distances from $X_i$.}
We implicitly partition the points in $Q'_i$ in classes based on their distance from the set $X_i$.
%Let $\tree^*$ be any fully constructed RBBD tree over $\Q(\I)$ such that if a node $u$ belongs in $\tree$ then $u$ also belongs in $\tree^*$. So if a node $u$ in $\tree$ is inactive then the same node $u$ in $\tree^*$ is also inactive. By the definition of $\tree$ such an RBBD tree $\tree^*$ exists with high probability. Notice that $\tree^*$ is not fully constructed by our algorithm (only its subset $\tree$ has been constructed), instead, it is only used for the next definition.
Let $a_0=0$, $b_0=\frac{L}{4|\Q(\I)|}$, $a_{\infty}=2L|\Q(\I)|$, $b_{\infty}=\infty$, and $a_j=2^{j-1}\frac{L}{4|\Q(\I)|}$,
$b_j=2^{j}\frac{L}{4|\Q(\I)|}$ for every $j=1,\ldots, M$, where $M=2\log(|\Q(\I)|)+3$.
%Next, we count the number of active points in each class of a valid family.
We construct a copy $\tree'$ of $\tree$ and repeat the next steps for every $j=0,1,\ldots, M,\infty$.
We define $\counts_{i,j}$ and initially set $\counts_{i,j}=0$.
We go through the points $x\in X_i$ one by one and we set $\counts_{i,j}=\counts_{i,j}+\tree'.\coun(x,b_j)$ and then we call $\tree'.\inactive(x,b_j)$.

Let $Q'_{i,j}\subseteq Q'_i$ be the set of points that were active just before the definition of $c_{i,j}$ and became inactive after applying $\tree'.\inactive(x,b_j)$ for all $x\in X_i$.
While we do not compute $Q'_{i,j}$ explicitly in the algorithm, in the analysis we show that for every iteration $i$:
%(1) for every pair $j_1,j_2\in [M]\cup\{0,\infty\}$, $Q'_{j_1}\cap Q'_{j_2}=\emptyset$, while $\bigcup_{j\in[M]\cup\{0,\infty\}}Q'_j=Q'$, (2) for every $j\in[M]\cup\{0,\infty\}$, the set $Q_j'\subseteq Q'$ contains all points of $Q'$ with distance greater than $(1+\eps)a_j$ and at most $b_j$ from $X$, might contain some points with distance at most $(1+\eps)b_j$ from $X$, and might contain some points with distance at least $a_j$ from $X$, and (3) for every $j\in[M]\cup\{0,\infty\}$, $\counts_j=|Q'_j|$.
\vspace{-0.3em}
\begin{enumerate}
    \item for every distinct pair $j_1,j_2\in [M]\cup\{0,\infty\}$, $Q'_{i,j_1}\cap Q'_{i,j_2}=\emptyset$,
while $\bigcup_{j\in[M]\cup\{0,\infty\}}Q'_{i,j}=Q'_i$,
    \item for every $j\in[M]\cup\{0,\infty\}$, the set $Q'_{i,j}\subseteq Q'_i$ contains all points of $Q'_i$ with distance greater than $(1+\eps)a_j$ and at most $b_j$ from $X_i$, might contain some points with distance at most $(1+\eps)b_j$ from $X_i$, and might contain some points with distance at least $a_j$ from $X_i$, and
    \item for every $j\in[M]\cup\{0,\infty\}$, $\counts_{i,j}=|Q'_{i,j}|$.
\end{enumerate}

$\bullet$ \textbf{Decide what subset of $Q'_i$ should be covered by $X_i$.}
Let $\beta_i$ be the largest index such that $\counts_{i,\beta_i}>\frac{\tau_i}{10\log |\Q(\I)|}$ (we will show that $\beta_i\leq M$ exists with high probability).
In the analysis we show that with high probability, the $k$-means cost we pay assigning all points in the first $\beta_i+1$ sets, $\hat{Q}_i=\bigcup_{\ell\in[\beta_i]\cup\{0\}}Q_{i,\ell}'$ to the centers $X_i$ is bounded (approximately) by the cost that the optimum solution pays assigning $\hat{Q}_i$, i.e., $\kmeans_{X_i}(\hat{Q}_i)\leq O(1)\cdot \kmeans_{\opt(\Q(\I))}(\hat{Q}_i)$.
We delete $\tree'$ and we repeat the procedure from the previous bullet using the tree $\tree'$ up to $j=\beta_i$.
In other words, $\tree$ gets the instance of $\tree'$ at the moment just before $c_{i,\beta_i+1}$ is defined.
In this version of $\tree$, 
each point in $\hat{Q}_i$ is inactive.
We show an example in Figure~\ref{fig:kmeans}.
%Equivalently, we claim that every point in $\hat{Q}$ can be assigned to a center in $X$ in the $k$-means solution we construct, so we are done with the points in $\hat{Q}$ and we should not consider this set again in the future. 
We observe that in the next iteration $Q'_{i+1}\leftarrow Q'_i\setminus \hat{Q}_i$. We update $\tau_{i+1}=\tau_i-\sum_{\ell\in[\beta_i]\cup\{0\}}\counts_{i,\ell}$ and $r_\ret=r_\ret+\sum_{\ell\in[\beta_i]\cup\{0\}}\counts_{i,\ell}((1+\eps)b_\ell)^2$. 
%Even though we do not compute $Q'$ or $\hat{Q}$ explicitly,
%We note that $\tau=|Q'|$.
%Finally, we update $r_\ret=r_\ret+\sum_{\ell\in[\beta]\cup\{0\}}\counts_\ell((1+\eps)b_\ell)^2$.

\input{kmeans-fig}

%In Appendix~\ref{appndx:kmeans}, we show the next theorem.

\paragraph{Analysis} The proofs of the next lemmas are shown in Appendix~\ref{appndx:kmeans}.

\begin{lemma}
\label{lem:help1}
For every iteration $i$,
%i) the set $X$ is a set of $O(k\log^2 (|\Q(\I)|))$ uniformly random samples from $Q'$, 
the family of sets $Q_{i,0}',\ldots, Q_{i,1}',\ldots, Q_{i,M}',Q_{i,\infty}'$ satisfy the three properties in the third bullet.%, with probability at least $1-\frac{1}{N^{2}}$.
%In each iteration the for the constructed set $U_j$ it holds that $U_j=U_j[a_j,b_j]$ for $j\in[M]\cup\{0,\infty\}$ as defined by Equation~\eqref{eq:classes}. iii) $r_j=|Q_j'|$ for $j\in[M]\cup\{0,\infty\}$
\end{lemma}

Next, we show that 1) the number of iterations of the algorithm is $O(\log N)$, so $|\ret|=O(k\log^3 N)$, and 2) for every iteration $i$, $\kmeans_{X_i}(\hat{Q}_i)=O\left(\kmeans_{\opt(\Q(\I))}(\hat{Q}_i)\right)$.
%We first prove the first argument.
For every iteration $i$, let $\badpoints=\{p\in Q_i'\mid \dist(p,X_i)>2\dist(p,\opt(\Q(\I)))\}$.
%be the set of points such that for every $p\in Q_{bad}$, there exists one iteration of the algorithm where $p\in\hat{Q}$ with $\dist(p,X)> 2\dist(p,\opt(\Q(\I)))$.
We show that the size of $\badpoints$ is small with high probability.
\begin{lemma}
\label{lem:bad}
    For every iteration $i$, $|\badpoints|\leq \frac{\tau_i}{20\log|\Q(\I)|}$, with probability at least $1-\frac{1}{N^{3}}$.
\end{lemma}

Since $|Q'_i|>480\cdot k\log^2|\Q(\I)|$, $M+2=2\log |\Q(\I)| +5$, and by Lemma~\ref{lem:help1}, there exists at least one set $Q_{i,j^*}'$ such that $\counts_{i,j^*}>\frac{\tau_i}{10\log |\Q(\I)|}$, by the pigeonhole principle. So, the index $\beta_i$ is well defined. In the proof of the next lemma, we also show that $\beta_i\neq \infty$.

%\vspace{-0.3em}
\begin{lemma}
\label{lem:helpred}
    For every iteration $i$, $|\hat{Q}_i|\geq \frac{\tau_i}{2}$, so the algorithm executes $O(\log |\Q(\I)|)=O(\log N)$ iterations, and $|\ret|=O(k\log^3 N)$, with probability at least $1-\frac{1}{N^2}$.
\end{lemma}

The next lemma bounds the $k$-means cost of $\ret$ and the running time of the algorithm.
%is low with high probability.

\begin{lemma}
\label{lem:cost}
   % $\kmeans_{\ret}(\Q(\I))\leq 69\kmeans_{\opt(\Q(\I))}(\Q(\I))$ and
   For any arbitrary small constant $\eps\!\in\!(0,1)$,
   it holds that $\kmeans_{\ret}(\Q(\I))\leq r_\ret\leq (84+\eps)\cdot \kmeans_{\opt(\Q(\I))}(\Q(\I))$, with probability at least $1-\frac{1}{N}$.
   The algorithm runs in $O(kN\log^5 N + k\log^7 N)$ time with probability at least $1-\frac{1}{N}$.
\end{lemma}

We argue that the same bounds hold for the relational $k$-median clustering. In fact, the constant approximation factor for the relational $k$-median problem is smaller since we do not raise the distances to the power of $2$. 
We conclude with the next Theorem.

%\vspace{-0.2em}
\begin{theorem}
\label{thm:kmeansconstant}
            Given an acyclic join query $\Q$ over $d$ attributes, a database instance $\I$ of size $N$, a parameter $k$, and an arbitrarily small constant parameter $\eps\in (0,1)$, there exists an algorithm that computes a set $\ret\subseteq \Q(\I)$ of $O(k\log^3 N)$ points and a number $r_\ret$ such that $\kmeans_{\ret}(\Q(\I))\leq r_\ret\leq  (84+\eps)\cdot\kmeans_{\opt(\Q(\I))}(\Q(\I))$, with probability at least $1-\frac{1}{N}$. The running time of the algorithm is 
    $O(kN\log^5 N + k\log^7 N)$ with probability at least $1-\frac{1}{N}$. An algorithm with the same theoretical guarantees and running time exists for relational $k$-median clustering.
\end{theorem}
\vspace{-0.7em}
\subsection{Better approximation algorithm}
\label{subsec:kmeans}
\vspace{-0.2em}
Using our result from Theorem~\ref{thm:kmeansconstant}, along with the coreset construction from~\cite{esmailpour2024improved} we get the fastest known $(4+\eps)\gamma$-approximation for the relational $k$-means clustering problem and the fastest known $(2+\eps)\gamma$-approximation for the relational $k$-median clustering problem. The RBBD construction in the relational setting can also accelerate the coreset construction in~\cite{esmailpour2024improved}.

First we execute the algorithm from Theorem~\ref{thm:kmeansconstant} and we get $\ret$ and $r_\ret$. 
%Since $\ret$ is an $O(1)$-approximation solution to the relational $k$-means clustering and $r_\ret$ satisfies $\kmeans_\ret(\Q(\I))\leq r_\ret\leq O(1)\kmeans_{\opt(\Q(\I))}(\Q(\I))$, 
Given $\ret, r_\ret$, a coreset $\coreset\subseteq \Q(\I)$ with $|\coreset|=O(k\eps^{-d}\log^4 N)$ for the relational $k$-means clustering can be constructed in $O(kN\log^5 N + k^2\log^{11}N\cdot\min\{\log^{d+1}(k), k\log^5 N\})$ time~\cite{esmailpour2024improved}.
%, assuming arbitrarily small constant $\eps$. 
Finally, a $\gamma$-approximation algorithm $\DkmeansAlg$ is executed on the coreset $\coreset$ to derive a $(4+\eps)\gamma$-approximation in $O(\timeMeans_\gamma(k\log^4 N))$ additional time.
%For example, Feder and Greene algorithm~\cite{} is executed on $\coreset$ in $O(|\coreset|\log k)=O(k\eps^{-d}\log k)$ additional time to get a $2(1+\eps)$-approximation solution for the $k$-center clustering problem on $\Q(\I)$. Notice that the overall asymptotic complexity of the algorithm does not change.
%By setting $\eps\leftarrow \eps/4$, the approximation factor becomes $(4+\eps)\gamma$. 
In Appendix~\ref{appndx:fastercoreset} we show how the relational RBBD tree construction can be used to improve the running time of the coreset construction in~\cite{esmailpour2024improved} to $O(kN\log^5 N +k\log^7 N)$ time.
%Skipping the low level details, we conclude with the main result of this section.
%Overall, the running time of the algorithm is $O(kN\log^5 N + k^2\log^{11}N\cdot\min\{\log^{d+1}(k), k\log^5 N\}+\timeMeans_{\gamma}(k\log^4 N))$.

%\vspace{-0.5em}

\begin{theorem}
\label{thm:kmeansopt}
            Given an acyclic join query $\Q$, a database instance $\I$ of size $N$, a parameter $k$, and an arbitrarily small constant parameter $\eps\in(0,1)$, there exists an algorithm that computes a set $\ret\subseteq \Q(\I)$ of $k$ points such that
            $\kmeans_{\ret}(\Q(\I))\leq  (4+\eps)\gamma\kmeans_{\opt(\Q(\I))}(\Q(\I))$, with probability at least $1-\frac{1}{N}$. The running time of the algorithm is 
            $O(kN\log^5 N + k\log^{7}N+\timeMeans_{\gamma}(k\log^4 N))$
            %$O(kN\log^5 N + k^2\log^{11}N\cdot\min\{\log^{d+1}(k), k\log^5 N\}+\timeMeans_{\gamma}(k\log^4 N))$
            with probability at least $1-\frac{1}{N}$. An algorithm with approximation factor $(2+\eps)\gamma$ and the same running time exists for the relational $k$-median clustering problem.
\end{theorem}

%% file: kmeans-fig.tex
\begin{figure}[t]
    \centering
\begin{tikzpicture}[scale=1.2][t]

% Light fill colors
\definecolor{myred}{RGB}{255,100,100}
\definecolor{myblue}{RGB}{100,140,255}
\definecolor{mygreen}{RGB}{100,210,140}

% SHADING (drawn first, so points appear above)
\fill[myred!10]   (1.8, 1.2)   circle (0.52);
\fill[myred!10]    (1.8, 1.2)   circle (0.26);
\fill[myblue!10]  (3.3, 1.2)   circle (0.52);
\fill[myblue!10]   (3.3, 1.2)   circle (0.26);
\fill[mygreen!10] (6.2, 1.25)  circle (0.52);
\fill[mygreen!10]  (6.2, 1.25)  circle (0.26);

% --- BLACK POINTS ---
% RED CIRCLE POINTS
%\fill[black] (1.80, 1.28) circle (0.028);    % in first shaded region (innermost disk)
\fill[black] (1.62, 1.60) circle (0.028);    % in second shaded region (first annulus)
\fill[black] (1.65, 0.78) circle (0.028);    % in second shaded region
\fill[black] (2.07, 1.45) circle (0.028);    % in second shaded region
\fill[black] (1.95, 0.85) circle (0.028);    % in second shaded region
\fill[black] (1.50, 0.95) circle (0.028);    % in second shaded region
\fill[black] (1.38, 1.05) circle (0.028);    % in second shaded region
\fill[black] (1.35, 1.15) circle (0.028);    % in second shaded region

% BLUE CIRCLE POINTS
\fill[black] (3.25, 1.28) circle (0.028);    % in first shaded region (innermost disk)
\fill[black] (3.05, 1.47) circle (0.028);    % in second shaded region (first annulus)
\fill[black] (3.12, 0.85) circle (0.028);    % in second shaded region
\fill[black] (3.55, 1.47) circle (0.028);    % in second shaded region
\fill[black] (3.58, 0.86) circle (0.028);    % in second shaded region

% GREEN CIRCLE POINTS
\fill[black] (6.20, 1.43) circle (0.028);    % in first shaded region (innermost disk)
\fill[black] (6.05, 1.57) circle (0.028);    % in second shaded region (first annulus)
\fill[black] (6.13, 0.90) circle (0.028);    % in second shaded region
\fill[black] (6.43, 1.58) circle (0.028);    % in second shaded region
\fill[black] (6.48, 0.88) circle (0.028);    % in second shaded region

% Overlap (red outer & blue middle)
\fill[black] (2.45, 1.18) circle (0.028);
\fill[black] (2.60, 1.10) circle (0.028);

% --- OUTSIDE/OTHER POINTS ---
\fill[black] (1.50, 1.80) circle (0.028);
\fill[black] (2.35, 1.73) circle (0.028);
%\fill[blue] (2.85, 1.82) circle (0.028);
%\fill[orange] (3.70, 1.76) circle (0.028);
\fill[black] (4.20, 1.29) circle (0.028);
\fill[black] (2.70, 0.57) circle (0.028);
\fill[black] (3.75, 0.66) circle (0.028);
\fill[black] (4.96, 0.80) circle (0.028);
\fill[orange] (6.22, 1.81) circle (0.028);
%\fill[pink] (6.05, 0.65) circle (0.028);
\fill[black] (4.32, 1.69) circle (0.028);
\fill[black] (5.50, 1.81) circle (0.028);

% Additional points to the far right of green circles
\fill[black] (7.2, 1.2) circle (0.028);
\fill[black] (7.5, 1.6) circle (0.028);
\fill[black] (7.6, 1.0) circle (0.028);
\fill[black] (7.9, 1.4) circle (0.028);
\fill[black] (7.4, 0.8) circle (0.028);
\fill[black] (8.0, 1.7) circle (0.028);
\fill[black] (8.1, 1.2) circle (0.028);

% More points between blue and green circles, outside both
\fill[black] (3.6, 1.0) circle (0.028);
\fill[black] (3.69, 1.2) circle (0.028);
\fill[black] (4.6, 1.5) circle (0.028);
\fill[black] (4.8, 1.0) circle (0.028);
\fill[black] (5.2, 1.3) circle (0.028);
\fill[black] (3.5, 0.9) circle (0.028);
\fill[black] (5.0, 1.7) circle (0.028);

% --- COLORED POINTS ---
\fill[red]   (1.8, 1.2)   circle (0.045);
\fill[blue]  (3.3, 1.2)   circle (0.045);
\fill[green!70!black] (6.2, 1.25) circle (0.045);

% --- LABELS (smaller and just below centers) ---
\node[font=\footnotesize, below, blue] at (3.3, 1.23) {$x_1$};
\node[font=\footnotesize, below, red]  at (1.8, 1.23)  {$x_2$};
\node[font=\footnotesize, below, green!60!black] at (6.2, 1.28) {$x_3$};

\node[font=\footnotesize, above, orange] at (6.22, 1.81) {$p$};

% --- CIRCLES (solid and dashed) ---

% Red
\foreach \r in {0.26, 0.52, 1.04}{
  \draw[red, thick]   (1.8, 1.2)   circle (\r);
  \draw[red, thick, dashed]   (1.8, 1.2)   circle ({1.15*\r});
}
% Blue
\foreach \r in {0.26, 0.52, 1.04}{
  \draw[blue, thick]  (3.3, 1.2)   circle (\r);
  \draw[blue, thick, dashed]  (3.3, 1.2)   circle ({1.15*\r});
}
% Green
\foreach \r in {0.26, 0.52, 1.04}{
  \draw[green!70!black, thick] (6.2, 1.25) circle (\r);
  \draw[green!70!black, thick, dashed] (6.2, 1.25) circle ({1.15*\r});
}

\end{tikzpicture}
\caption{Example of the $k$-means algorithm, where three colored points $x_1, x_2, x_3$ represent the points in the random set $X_i$. Solid circles represent $\ball(x_\ell, b_j)$, and dashed circles $\ball(x_\ell, (1+\eps) b_j)$ for $\ell \in [3], j \in \{0,1,2\}$. If $\beta_i=1$, the points in the shaded area belong in $\hat{Q}_i$, i.e., they are considered covered and they become inactive. The orange point $p$ may or may not become inactive since it lies in $\ball(x_3,(1+\eps)b_1)\setminus \ball(x_3,b_1)$.}
    \label{fig:kmeans}
    %\vspace{-1.1em}
\end{figure}
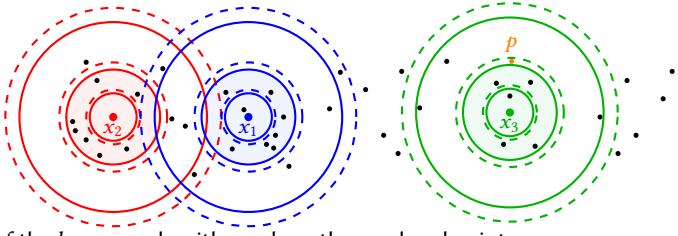

%% file: extensions.tex
\vspace{-1em}
\section{Extensions}
\label{sec:ext}
\vspace{-0.5em}
\paragraph{Generality}
Our method is quite general: Any algorithm in the standard computational setting that has access to the input data through $\numqueries$ queries to a BBD tree (using the stored information of count and representative point in each node) can be converted to a randomized algorithm in the relational setting with the same approximation guarantee and running time $O(\numqueries\cdot N)$ (ignoring logarithmic factors).
As many geometric approximation algorithms fall into this category, our framework yields efficient relational counterparts for a variety of tasks. For example, in Appendix~\ref{appndx:gonzalez}, we show that our new method can be used to implement an approximate version of the Gonzalez algorithm~\cite{gonzalez1985clustering}, leading to efficient algorithms for relational diversity and fairness problems.

\vspace{-0.2em}
\paragraph{Cyclic queries}
We use the notion of fractional hypertree width~\cite{gottlob2014treewidth} and apply a standard procedure to extend our algorithms to every join query $\Q$.
%We use the notion of \emph{fractional hypertree width}~\cite{gottlob2014treewidth} denoted by $\fhw$ which roughly measures how close $\Q$ is to being acyclic.
For a cyclic query $\Q$, we convert it to an equivalent acyclic query such that each relation is the result of a (possibly cyclic) join query with fractional edge cover at most $\fhw(\Q)$. We evaluate the (possibly cyclic) queries to get the new relations and then apply our algorithms to the acyclic query.
All the guarantees are the same as in Theorems~\ref{thm:kCenter},~\ref{thm:kmeansopt},~\ref{thm:kmeansCATopt} except for the running time, where we simply replace any $N$ factor with $N^{\fhw}$.
Since it is a typical method in database theory, we give the details in Appendix~\ref{appndx:generalQueries}.

%% file: conclusion.tex
\vspace{-0.7em}
\section{Future work}
\vspace{-0.2em}
%In this paper, we introduced a novel approach that combines database theory with geometric insights to design faster algorithms for well-known optimization problems, such as clustering, in the relational setting. Our method simulates the on-the-fly construction of an RBBD tree to enable efficient geometric algorithms over relational data. 
This work opens several directions for future research. A central open question is whether our $\O(kN)$-time algorithms for $k$-center/median/means clustering are optimal.
Although our methods outperform the current state-of-the-art, it is interesting to study whether other geometric techniques can be adapted to achieve $\O(N)$ time algorithms in the relational setting. 
We also aim to investigate whether known lower bounds for $k$-center clustering in the standard computational setting~\cite{feder1988optimal, bateni2023optimal} can be extended to the relational setting, potentially establishing near-optimal results.
%Another intriguing direction is whether there exists an implementation of the (approximate) Gonzalez algorithm in the relational setting with a runtime better than $\O(k^2N)$.

%% file: appendix.tex
\newpage
\appendix
\section{Missing details from Section~\ref{sec:RBBD}}
\label{appndx:RBBD}

\subsection{Execution of centroid shrink when $\square_u$ is a box with a hole}
\label{appndx:RBBDhole}
First, we describe the procedure in the standard BBD tree.
Initially $\rho=\square_u^+$. They compute the smallest midpoint box that contains all points in $\rho \cap P_u$ along with the inner box $\square_u^-$. Let $\rec'\subseteq \rec$ be this box. They apply $O(d)$ fair splits on $\rec'$ until they find a non-trivial fair split, i.e., a fair split that splits the points in $\rec'\cap P_u$ into two non-empty subsets. Let $\hat{\rec}\subset \rec$ be one of the two boxes created by the non-trivial fair split such that $|\hat{\rec}\cap P_u|\geq \frac{|\rec'\cap P_u|}{2}$. %We say that $\hat{\rho}$ is the majority midpoint box.
The key observation is that $|\hat{\rec}\cap P_u|\leq |\rec\cap P_u|-1$.
If $\hat{\rec}$ contains the inner box $\square_u^-$, then they check the following condition. If $|\hat{\rec}\cap P_u|> \frac{2}{3}|P_u|$, they set $\rec=\hat{\rec}$ and they continue recursively.
If $|\hat{\rec}\cap P_u|\leq \frac{2}{3}|P_u|$, then they construct two children nodes $w, v$ of $u$ such that $\square_w=\rec'\setminus\hat{\rec}$ is a box (without a hole), and  $\square_v$ is a box with a hole with $\square_v^+=\hat{\rec}$ and $\square_v^-=\square_u^-$.

Otherwise, let $\hat{\rec}$ be the first box found by the recursive procedure that does not contain the inner box $\square_u^-$. Then they create two children $w,v$ of $u$ such that $\square_w$ is a box with a hole with $\square_w^+=\rec$ and $\square_w^-=\rho'$, and $\square_v$ is a box with a hole with $\square_v^+=\rho'$ and $\square_v^-=\square_u^-$. Furthermore, they create two children $v_1, v_2$ of $v$. $\square_{v_2}$ is a box with a hole with $\square_{v_2}^+=\rec'\setminus \hat{\rec}$ and $\square_{v_2}^-=\square_u^-$, while $\square_{v_1}=\hat{\rec}$ is a box (without a hole). The recursive procedure continues setting $\rec=\hat{\rec}$. Notice that $\rec=\hat{\rec}$ does not contain any hole. They continue the recursive approach by finding the smallest midpoint box $\rec'$ that contains $P_u\cap \rec$ and applying a fair split on $\rec'$ until they find $\hat{\rec}$ such that $\frac{1}{3}|P_u|\leq |\hat{\rec}\cap P_u|\leq \frac{2}{3}|P_u|$. In the end, they create two children nodes $v_{1a}, v_{1b}$ of $v_1$ such that $\square_{v_{1b}}$ is a box with a hole, with $\square_{v_{1b}}^+=\rec'$ and $\square_{v_{1b}}^-=\hat{\rec}$, while $\square_{1a}=\hat{\rec}$ is a box.
Overall, the centroid shrink might create two new nodes $v, w$ or six new nodes $w, v, v_1, v_{1a}, v_{1b}, v_2$. As shown in~\cite{arya1998optimal}, the smallest midpoint box $\rec'\subseteq\rec$ that contains all points in $P_u\cap \rec$ can be found in $O(1)$ time if we know the corners of $\rec$ and the corners of $\square_{u}^-$.
Since in every recursive iteration $|\hat{\rec}\cap P_u|\leq |\rec'\cap P_u|-1$, the centroid procedure is executed in $O(|P_u|)$ time.

For the RBBD tree, we execute the same operations on $\net_u$, instead of $P_u$. The first time we find a box $\hat{\rec}$ (after the non-trivial fair split) that does not contain $\square_{u}^-$, we construct the subtree with the six new nodes as described above.

\subsection{Proof of Theorem~\ref{thm:RBBD}}
\label{appndx:constructionRBBD}
We prove Theorem~\ref{thm:RBBD} by proving the next lemmas.

\begin{lemma}
\label{lem:rcs}
The randomized centroid shrink procedure on a node $u$ computes a midpoint box $\hat{\rec}$ such that $\frac{1}{3}-\eps\leq \frac{|\hat{\rec}\cap P_u|}{|P_u|}\leq \frac{2}{3}+\eps$, with probability at least $1-\frac{1}{\Omega(n^\cons)}$.
\end{lemma}
\begin{proof}
  The recursive approach in the randomized centroid shrink is similar to the standard centroid shrink, but instead of executing it on $P_u$ it is executed in $\net_u$.
  Hence, the randomized centroid shrink procedure will finish after $O(\eps^{-2}\log n^\cons)=O(\eps^{-2}\log n^\cons)$ recursive calls. The algorithm stops the first time it finds a majority midpoint box $\hat{\rec}$ such that $|\hat{\rec}\cap \net_u|\leq \frac{2}{3}|\net_u|$. By definition, the set $\net_u$ is an $\eps$-sample with probability at least $1-\frac{1}{\Omega(n^\cons)}$. Hence, $\frac{|\hat{\rec}\cap P_u|}{|P_u|}\leq \frac{|\hat{\rec}\cap \net_u|}{|\net_u|}+\eps\leq \frac{2}{3}+\eps$. It remains to show the other direction of the inequality. Let $\rec'$ be the midpoint box that we executed the last fair split (to find $\hat{\rec}$). By definition, we had $\frac{|\net_u\cap \rec'|}{|\net_u|}>\frac{2}{3}$. Notice that by definition of $\hat{\rec}$, it holds that $|\net_u\cap \hat{\rec}|\geq \frac{|\net_u\cap \rec'|}{2}$. Hence, $\frac{|\net_u\cap \hat{\rec}|}{|\net_u|}\geq \frac{1}{2}\cdot \frac{|\net_u\cap \rec'|}{|\net_u|}\geq \frac{1}{3}$. Since $\net_u$ is an $\eps$-sample with probability at least $1-\frac{1}{\Omega(n^\cons)}$, we conclude $\frac{|\hat{\rec}\cap P_u|}{|P_u|}\geq \frac{|\hat{\rec}\cap \net_u|}{|\net_u|}-\eps\geq \frac{1}{3}-\eps$. The lemma follows.
\end{proof}

\begin{lemma}
\label{lem:RBBDhelp2}
    The RBBD tree has the same properties and the same asymptotic query complexity as the BBD tree, with high probability. The height of the RBBD tree is $O(\log n)$, the number of nodes is $O(n)$ and the query complexity is $O(\log n + \eps^{-d+1})$, with probability at least $1-\frac{1}{n^{\cons-1}}$.
\end{lemma}
\begin{proof}
    The construction of the RBBD tree is similar to that of the BBD tree, except for the centroid shrink. Using the randomized centroid shrink, for a node $u$ we find a box $\hat{\rec}$ such that $\frac{1}{3}-\eps\leq \frac{|\hat{\rec}\cap P_u|}{|P_u|}\leq \frac{2}{3}+\eps$, with probability at least $1-\frac{1}{\Omega(n^\cons)}$. 
    
    For now, let us assume that the inequality of Lemma~\ref{lem:rcs} holds for every execution of the randomized centroid shrink. Next, we show that under this assumption, asymptotically, there is no difference between the BBD tree and the RBBD tree for sufficiently small constant $\eps$. In~\cite{arya1998optimal} the authors show that 
    every $4$ levels of descent in the BBD tree, the number of points associated with the nodes decreases by at least a factor $\frac{2}{3}$, so the height of the tree is $O(\log n)$ and the BBD tree has $O(n)$ nodes. Similarly, in our case, observe that each randomized centroid shrink introduces three new levels into the tree, and we alternate this with fair splits, so it follows that with each four levels the number of points decreases by at least a factor of $\frac{2}{3}-\eps$. For a small $\eps$, for example $\eps<0.1$, the height of the tree is $O(\log n)$ by the same argument. Equivalently we can argue that every $8$ levels of descent in the RBBD tree, the number of points associated with the nodes decreases by at least a factor $\frac{2}{3}$, the height of the tree is $O(\log n)$ and the number of nodes is $O(n)$.
Next, the authors in~\cite{arya1998optimal} argue that every $4d$ levels of descent in the BBD tree, the sizes of the associated cells decrease by at least a factor of $\frac{2}{3}$. Similarly, in our case, note that to decrease the size of a cell, we must decrease its size
along each of $d$ dimensions. Since fair splits are performed at least at every fourth level of the tree, and each such split decreases the longest side by at least a factor of $2/3$, it follows that after at most $d$ splits (that is, at most $4d$ levels of the tree) the size decreases by this factor. All proofs of the BBD tree~\cite{arya1998optimal} also hold for the RBBD tree assuming that all randomized centroid shrink operators are executed correctly, i.e., assuming that Lemma~\ref{lem:rcs} holds with probability $1$.
Given that the inequality of Lemma~\ref{lem:rcs} holds with probability at least $1-\frac{1}{\Omega(n^{\cons})}$ and using the union bound over $O(n)$ nodes, we have that the probability that all randomized centroid shrink operators are executed correctly is at least $1-\frac{1}{n^{\cons-1}}$. Hence, we conclude that all proofs and properties of the BBD tree also hold (asymptotically) for the RBBD tree with probability at least $1-\frac{1}{n^{\cons-1}}$.

Finally, the query procedure between BBD tree and RBBD tree is identical. Given a query ball $\ball(p,r)$ we start from the root of the tree. Let $u$ be the node of the RBBD tree the query procedure processes.
We check whether $\ball(p,r)\cap \square_u=\emptyset$. If yes, we stop the execution towards that branch of the tree.
If $\square_u\subseteq \ball(p,(1+\eps)r)$ we add $u$ in the set of the canonical nodes $\canonical(p,r)$.
In any other case, we continue the search in both children of $u$. 
In the end the RBBD tree always computes a set of canonical nodes $\canonical$ such that $\ball(p,r)\cap P\subseteq P\cap \bigcup_{u\in \canonical}\square_u \subseteq \ball(p,(1+\eps)r)\cap P$.
Since the height of the RBBD tree is $O(\log n)$ with high probability, following our discussion above, the query analysis from~\cite{arya1998optimal, arya2000approximate} can also be used for the RBBD tree to show that $\canonical(p,r)$ can be computed in $O(\eps^{-d+1}+\log n)$ time and $|\canonical(p,r)|=O(\eps^{-d+1}+\log n)$, with probability at least $1-\frac{1}{n^{\cons-1}}$.

\end{proof}

\begin{lemma}
    Assume that we can get a random sample set $\net_u$ from a box or a box with a hole, in $f(n,\eps)$ time, where $f$ is a positive real function that depends on $n$ and $\eps$, then the RBBD tree can be constructed in $O\left(n\log (n)+n\cdot \left(f(n,\eps)+\eps^{-2}\log^2 n\right)\right)$ time, with probability at least $1-\frac{1}{n^{\cons-1}}$.
\end{lemma}
\begin{proof}
    Since the fair split procedure is identical to the standard BBD tree, all fair split procedures are executed in $O(n\log n)$ time using the machinery from the standard BBD tree. Next, we focus on the randomized centroid shrink procedure. For a node $u$ we get the set $\net_u$ in $f(n,\eps)$ time. We find the smallest midpoint box $\rec'$ that contains $\net_u$ along with the inner box $\square_u^-$ in $O(\eps^{-2}\log n)$ time since we explicitly know all points in $\net_u$ and the corners of $\square_u^+$, using the highest and lowest values in each coordinate as described in~\cite{arya1998optimal}.
After the computation of $\rec'$, we execute $O(1)$ fair splits in order to find the first non-trivial fair split. The hyperplane in every fair split is computed straightforwardly in $O(1)$ time. Using a simple binary search tree over each coordinate of the points in $\net_u$, we can check whether a fair split is trivial in $O(\log (\eps^{-2}\log n))$ time. After finding the first non-trivial fair split, we compute $\hat{\rec}$ in $O(\log (\eps^{-2}\log n))=O(\log n)$ time using the search binary tree that corresponds to the last non-trivial fair split by checking whether $|\hat{\rec}\cap \net_u|>|(\rec'\setminus \hat{\rec})\cap \net_u|$. After each non-trivial split, the box $\hat{\rec}$ contains at most $|\rec'\cap \net_u|-1$ samples from $\net_u$. Hence, this recursive approach runs for at most $O(|\net_u|)=O(\eps^{-2}\log n)$ times. Thus, each randomized centroid shrink runs in $O(f(n,\eps)+\eps^{-2}\log^2 n)$ time. 
From Lemma~\ref{lem:RBBDhelp2}, the RBBD tree has height $O(\log n)$ and $O(n)$ nodes with probability at least $1-\frac{1}{n^{\cons-1}}$, so the running time of all randomized centroid shrink procedures is $O(n(f(n,\eps)+\eps^{-2}\log^2 n))$, with probability at least $1-\frac{1}{n^{\cons-1}}$.
    \end{proof}

\subsection{Sampling from a box with a hole}
\label{appndx:randcentshr} \input{figdecomp}
We consider the more involved case where $u$ is a node such that $\square_u$ is a box with a hole. The goal is to sample uniformly at random $\Lambda=\Theta(\eps^{-2}\log(|\Q(\I)|^2))=\Theta(\eps^{-2}\log N)$ points from $\square_u\cap \Q(\I)$. Using Yannakakis algorithm \cite{yannakakis1981algorithms}, $|\Q(\I)|$ can be computed at the beginning of our approach in $O(N)$ time. From Lemma~\ref{lem:Rects}, it is known how to sample $z$ tuples in a box efficiently using the procedure $\mathsf{SampleRect}$. However, $\square_u$ is a box with a hole, so we should sample from the difference of two boxes, $\square_u^+\setminus \square_u^-$. We partition $\square_u$ into $2d$ disjoint boxes using the corners of $\square_u^-$. See Figure~\ref{fig:decomp} for an example in $2$ dimensions. Let $\mathcal{R}=\{\rho_1,\ldots, \rho_{2d}\}$ be the set of $2d$ disjoin boxes that partition $\square_u$.
First, we decide how many samples (out of the total $\Lambda$ samples) to draw from each box in $\mathcal{R}$, and then we call the function $\mathsf{SampleRect}$ to obtain this number of samples from each box.
We compute $C_i=\mathsf{CountRect}(\Q,\I,\rho_i)$ for every $i\in[2d]$. Based on the values $C_1,\ldots, C_{2d}$, for every $j\in [\Lambda]$, we sample the box in $\mathcal{R}$ that a sample point should be taken from, i.e., a box $\rho_i\in\mathcal{R}$ is selected with probability $\frac{C_i}{\sum_{\ell\in[2d]} C_\ell}$. Let $s_i$ be the number of samples we should get from $\rho_i$. Then, for every $i\in[2d]$, we run $S_i=\mathsf{SampleRect}(\Q,\I,\rho_i,s_i)$. We set $\net_u=\bigcup_{i\in[2d]}S_i$. %\aryan{I think the definition of $s_i$ and $S_i$ are confusing. We should explain that we first compute how many samples we need from each cell. Now it seems we are sampling one by one.}
We claim that $\net_u$ is an $\eps$-sample for the points in $\Q(\I)\cap \square_u$. It is sufficient to show that $\net_u$ is a set of $\Lambda$ uniform random samples (with replacement) from $\square_u\cap \Q(\I)$. Equivalently, we show that the probability of selecting a specific tuple $t\in\square_u\cap \Q(\I)$ taking one sample is $\frac{1}{|\Q(\I)\cap \square_u|}$. Let $t$ be any tuple that belongs (without loss of generality) in $\rho_i$. The probability that $t$ is selected taking one sample from $\square_u\cap \Q(\I)$ is $\frac{C_i}{\sum_{\ell\in[2d]C_\ell}}\cdot \frac{1}{C_i}=\frac{1}{\sum_{\ell\in[2d]C_\ell}}=\frac{1}{|\Q(\I)\cap \square_u|}$. Hence $\net_u$ is a set of $\Lambda$ uniform samples (with replacement) from $\square_u\cap \Q(\I)$ and $\net_u$ is an $\eps$-sample with high probability. After computing $\net_u$ we run the recursive approach of finding the minimum midpoint box that contains $\net_u$ and the inner box $\square_u^-$ and applying $O(d)$ fair splits until we find the first non-trivial fair split. This approach only uses $\net_u$, $\square_u^+$, and $\square_u^-$ to implement, so the implementation is the same as in the standard computational setting.
%Hence, the correctness follows from the fact that $\net_u$ is an $\eps$-sample and Subsection~\ref{}.
Finally, we bound the running time. 
The $\mathsf{CountRect}$ procedure runs in $O(N)$ time so all numbers $C_1,\ldots, C_{2d}$ are computed in $O(2dN)=O(N)$ time. Then, each $\mathsf{SampleRect}(\Q,\I,\rho_i,s_i)$ approach runs in $O(N+C_i\log N)$. Hence, the set $\net_u$ is computed in $O(N+\eps^{-2}\log^2 N)$ time. After computing $\net_u$, the next steps of the randomized centroid shrink are executed in $O(\eps^{-2}\log^2 N)$ time as shown in Lemma~\ref{thm:RBBDrel}. Overall, the randomized centroid shrink in the relational setting runs in $O(N+\eps^{-2}\log^2 N)$ time. %with probability at least $1-\frac{1}{N^{m+3}}$.

\subsection{Count and representative point in a box with a hole}
\label{appndx:RBBDcount}
We assume that $\square_u$ is a box with a hole. We construct the partition $\mathsf{R}=\{\rho_1,\ldots, \rho_{2d}\}$ as we did in the randomized centroid shrink. For each $\rho_i\in \mathsf{R}$, we compute $C_i=\mathsf{CountRect}(\Q, \I,\rho_i)$ and we set $u.c=\sum_{\ell\in[2d]}C_\ell$. Similarly, we compute a representative $u.rep$ by executing $\mathsf{ReprRect}(\Q,\I, \rho_i)$ until we find a non-empty box (if any) in set $\mathsf{R}$. Both $u.c$ and $u.rep$ for a node $u$ are computed in $O(N)$ time using Lemma~\ref{lem:Rects} and the fact that $2d=O(1)$.

\subsection{Proofs of Lemmas in Section~\ref{sec:reloracles}}
\label{appndx:reloracles}
\begin{proof}[Proof of Lemma~\ref{lem:oracle1}]
   %Recall that the RBBD tree has the same properties as the BBD tree with probability at least $1-\frac{1}{N^{2}}$.  Hence, 
   If all nodes of $\tree$ were active then we would get a set of canonical nodes $\canonical'(x,r)$ such that $B(x,r)\cap \Q(\I)\subseteq \bigcup_{u\in \canonical'(x,r)}(\square_u\cap \Q(\I))\subseteq B(x,(1+\eps)r)\cap \Q(\I)$. Since there are inactive nodes, the query procedure on $\tree$ skips branches of the tree that are already inactive. Hence $\canonical(x,r)\subseteq \canonical'(x,r)$.
   Let $p$ be an inactive point, i.e., $p\in \Q(\I)\setminus Q'$. 
   Point $p$ will remain inactive because our procedure only makes new nodes inactive (it does not make any node active).
   Then, assume an active point $p\in Q'$ such that $\dist(p,x)>(1+\eps)r$.
   By definition, there is no node $u\in\canonical(x,r)$ such that $p\in \square_u$, so $p$ will remain active. Finally, let $p\in Q'$ be an active point such that $\dist(p,x)<r$.
   Then there would be a node $u'\in \canonical'(x,r)$ such that $p\in \square_{u'}$.
   Since $p$ is active there is no node $v\in \tree$ such that $v.a=0$ and $p\in \square_v$. 
   %Hence $u'\in \canonical(x,r)$.
   %Hence, there will be a node $u\in\canonical(x,r)$ such that $p\in\square_u$. 
   Node $u'$ becomes inactive, so $p$ will be inactive.
   We conclude that $B(x,r)\cap Q'\subseteq \bigcup_{u\in\canonical(x,r)}(\square_u\cap Q')\subseteq B(x,(1+\eps)r)\cap Q'$.

   Then we show that the values $u.rep$ and $u.c$ are correct after the update operation. This is shown by induction. Let $u$ be a node and $v,w$ be its children. If both $v.a=w.a$ then our algorithm correctly sets $u.a=0$ because $\square_u=\square_v\cup \square_w$. If $v.a=1$, $w.a=0$ and $v.rep, v.c$ are correct then our algorithm correctly sets $u.rep=v.rep$ and $u.c=v.c$ since $\square_u\cap Q'=\square_v\cap Q'$. If $v.a=0$ and $w.a=1$, we argue equivalently. Finally, if $v.a=w.a=1$, then if the counts and representative points in $v,w$ are correct then indeed $v.rep\in \square_u\cap Q'$, and by definition $|\square_u\cap Q'|=|\square_v\cap Q'|+|\square_w\cap Q'|=v.c+w.c$.

  % Finally, for the running time, we have $|\canonical(x,r)|=O(\eps^{-d+1}+\log N)$. Since the RBBD tree has the same properties as a BBD tree (with high probability), we have that the query procedure visits $O(\eps^{-d+1}+\log n)$ nodes to construct $\canonical(x,r)$. Let $D$ be the set of nodes visited by the query procedure. In the worst case, all nodes in $D$ had to be constructed in $\tree$. From Theorem~\ref{thm:RBBDrel}, we have that the construction time of a node takes $O(N+\eps^{-2}\log^2 N)$ time. Furthermore, notice that the update operation only visits each node in $D$ once, so overall the oracle runs in $O((\eps^{-d+1}+\log N)(N+\eps^{-2}\log^2 N))$.

Every new node is constructed in $O(N+\eps^{-2}\log^2 N)$  time
%with probability at least $1-\frac{1}{N^{m+5}}$, 
as shown in Lemma~\ref{thm:RBBDrel}.
The oracle runs a query on the RBBD tree to compute the canonical nodes $\canonical(x,r)$.
Since the number of samples we get in order to implement a randomized centroid shrink is $O(\eps^{-2}\log N^{m+4})$, the query procedure in the RBBD tree has the same properties as the query procedure in the BBD tree with probability at least $1-\frac{1}{N^{m+3}}$ (as shown in Lemma~\ref{lem:RBBDhelp2}).
Hence $|\canonical(x,r)|=O(\eps^{-d+1}+\log N)$ and $\canonical(x,r)$ is computed in $O(\eps^{-d+1}+\log N)$
with probability at least $1-\frac{1}{N^{m+3}}$.
%and the construction time is $O(\eps^{-2}\log N)$ with probability at least $1-\frac{1}{N^{m+4}}$ from Theorem~\ref{thm:RBBDrel}. 
%By the union bound, 
Thus, the running time is $O((\eps^{-d+1}+\log N)((N+\eps^{-2}\log^2 N)))$ with probability at least $1-\frac{1}{N^{m+3}}$. The same holds for the proofs of the next lemmas in this section.
\end{proof}

\begin{proof}[Proof of Lemma~\ref{lem:oracle2}]
    From the proof of Lemma~\ref{lem:oracle1}, we have that the query procedure computes a set of nodes $\canonical(x,r)$ such that $B(x,r)\cap Q'\subseteq \bigcup_{u\in \canonical(x,r)}(\square_u\cap Q')\subseteq B(x,(1+\eps)r)\cap Q'$. Since the counters on the active (constructed) nodes are correct, the lemma follows.
\end{proof}

\begin{proof}[Proof of Lemma~\ref{lem:oracle3}]
    If $Q'=\emptyset$ then indeed it must be the case that $\root.a=0$. Otherwise, there would be a path of active nodes from the root of $\tree$ to a node $u$ (without constructed children) where $\square_u\cap \Q(\I)\neq \emptyset$. If $Q'\neq \emptyset$, and given that the representative points in each node are correct (Lemma~\ref{lem:oracle1}), we have that the returned point $\root.rep$ is a point from $Q'$. In order to access $\root.a$ and $\root.rep$ we spend $O(1)$ time.
\end{proof}

\begin{proof}[Proof of Lemma~\ref{lem:oracle4}]
Let $U$ be a subset of the nodes in $\tree$ constructed as follows: a node $u$ belongs in $U$ if and only if i) $u\in \tree$, ii) there is no inactive node from the root to $u$ (including the roor and $u$), and iii) the children of $u$ have not been constructed in $\tree$ or $u$ is a leaf node of a fully constructed RBBD tree.
Notice that the $\mathsf{ReportRect}(\Q,\I,\square_u)$ procedure is executed exactly for every $u\in U$. By definition of the active points $Q'$ we have that $Q'=\bigcup_{u\in U}(\square_u\cap \Q(\I))$. Indeed, for each point $p\in Q'$ there is no inactive node $w\in \tree$ such that $p\in \square_w$.
Hence the procedure $\mathsf{ReportRect}(\Q,\I,\square_u)$ for every $u\in U$ returns exactly $Q'$. 

For the running time, we note that $|U|\leq |\tree|$ and each procedure $\mathsf{ReportRect}(\Q,\I,\square_u)$ runs in $O(N+|\square_u\cap \Q(\I)|)$ time. Hence the total running time is $O(|\tree|\cdot N+|Q'|)$.
\end{proof}

\begin{proof}[Proof of Lemma~\ref{lem:oracle5}]
Our goal is to show that in each iteration, a point $p\in Q'$ is selected with probability $\frac{1}{|Q'|}$.
Let $p\in Q'$ be a point stored in a leaf node $u$ of a complete RBBD tree (such a node may not have been constructed in $\tree$ yet).
Recall that a point $p$ belongs in $Q'$ if and only if there is no inactive node in $\tree$ from the root to $u$. Assume that $v$ is the last node in the path from $\root$ to $u$ that was constructed in $\tree$.
    Let $\root\rightarrow v_1\rightarrow\ldots\rightarrow v_\kappa\rightarrow v$. Let $w_1,\ldots, w_\kappa$ be the siblings of nodes $v_1,\ldots, v_\kappa$, respectively.
    From Section~\ref{subsec:RBBDrel}, we know that a point in a cell $\square_{v}$ is sampled with probability $\frac{1}{|\square_v\cap \Q(\I)|}$.
    The probability that $p$ is selected is $\prod_{j\in [\kappa]}\frac{v_j.c}{v_j.c+w_j.c}\cdot \frac{1}{|\square_v\cap \Q(\I)|}$. Notice that $v$ is active and no children of $v$ have been constructed in the tree, so $\square_v\cap \Q(\I)=\square_v\cap Q'$. Furthermore, by the correctness of the counts, $v_1.c+w_1.c=|Q'|$ and $v.c=|\square_v\cap Q'|$. Moreover, $v_j.c=v_{j+1}.c+w_{j+1}.c$, for every $j\in [\kappa]$. We conclude that the probability of selecting $p\in Q'$ is $\frac{1}{|Q'|}$.

    For every sample, we spend $O(\log N)$ time to get a node $v$ of $\tree$ which is either a leaf node or its children have not been constructed, because the height of $\tree$ is $O(\log N)$ with high probability. After we find $v$, the sampling procedure in $\square_v\cap Q'=\square_v\cap \Q(\I)$ is executed in $O(N)$ time as shown in Lemma~\ref{lem:Rects}. Since we get $z$ samples, the total running time is $O(zN)$.
    
\end{proof}

\section{Missing details from Section~\ref{sec:kCenter}}
\label{appndx:kCenter}

\begin{envrevthree}
\begin{algorithm}[t]
\caption{\revthree{{\sc Main algorithm}$(\Q, \I, k, \varepsilon, r)$}}
\label{alg:RelKCenterApprox}
     \revthree{Build the root of RBBD tree $\tree$ over $\Q(\I)$\;
     $\ret=\emptyset$\;
     \While{$\tree.\representative()\neq \emptyset$}{
        \If{$|\ret|=k$}{
            \Return $\emptyset$\;
        }
        $p=\tree.\representative()$\;
        $\ret=\ret\cup\{p\}$\;
        $\tree.\inactive(p,2r)$
     }
     \Return $\ret$}\;
\end{algorithm}
\end{envrevthree}

\begin{proof}[Proof of Theorem~\ref{thm:constantkCenter}]
%\label{appndx:binarysearchkcenter}
%Unfortunately, we cannot sort all pairwise distances of points in $\Q(\I)$ and then run a binary search on them because we would need $\Omega(|\Q(\I)|^2)$ time leading to a runtime which can be orders of magnitude larger than the size of the database. While it is not straightforward to run a binary search on $\ell_2$ pairwise distances, the key idea is to run a binary search over the $\ell_\infty$ pairwise distances of points in $\Q(\I)$.

%We first run the Yannakakis algorithm to remove dangling tuples from $\I$\footnote{Dangling tuples are those not participating in any result of the underlying join query.}. Then 
We notice that for any pair $t, p\in \Q(\I)$, $||t-p||_\infty=\max_{j\in[d]}|\pi_{A_j}(t)-\pi_{A_j}(p)|$. Hence, the $\ell_\infty$ distance between two points can be found as the absolute difference between two of the input tuples' coordinates. Using this observation, we construct a set $V$ of $O(N)$ points in $\Re^1$ such that all pairwise $\ell_\infty$ distances in $\Q(\I)$ are included in the pairwise distances of points in $V$. Initially, $V=\emptyset$. For every $A_j\in\allattr$, and every $R_i\in\allrel$ such that $A_j\in\allattr_i$, we add in $V$ the projections of tuples in $\R_i$ with respect to $A_j$, i.e., $V=V\cup\{\pi_{A_j}(R_i)\}$. After considering all attributes and relations, notice that for every $\ell_\infty$ pairwise distance between two tuples in $\I$, there exist two points in $V$ with the same distance.
After we constructed $V$, we run a binary search on the pairwise distances of the points in $V$. In order to do it, we use the result from~\cite{salowe1989infinity} to return the $z$-th smallest pairwise distance of a set $V$ of points in $\Re^1$ in $O(|V|\cdot \log |V|)=O(N\log N)$ time. Hence, we run a binary search requesting the $z$-th smallest distance for the selected $O(\log N)$ values of $z$ in the range $[0,1,\ldots, |V\times V|]$. %, where $|V\times V|=O(N^2)$. 
Let $r'$ be the $z$-th smallest distance returned by~\cite{salowe1989infinity}. We set $r=\sqrt{d}\cdot r'$ and we run the main algorithm described in Section~\ref{sec:kCenter}.
In the end, we return the last non-empty set $\ret$ returned by the binary search.

\begin{lemma}
    Using the binary search, the main algorithm for the relational $k$-center problem returns a set $\ret$ of $k$ points such that $\kcenter_\ret(\Q(\I))\leq 2\sqrt{d}(1+\eps)\kcenter_{\opt(\Q(\I))}(\Q(\I))$.%, with probability at least $1-\frac{1}{N}$.
\end{lemma}
\begin{proof}
    For any pair of points $t, p\in \Re^d$, $||t-p||_\infty \leq ||t-p||_2\leq \sqrt{d}||t-p||_\infty$.
    Let $r^*=\kcenter_{\opt(\Q(\I))}(\Q(\I))$ and let $t^*, p^*$ be two tuples in $\Q(\I)$ such that that $||t^*-p^*||_2=r^*$. Let $\hat{r}=||t^*-p^*||_\infty$. By the definition of $V$ there will be two points in $V$ with distance exactly $\hat{r}$. By definition, we have $\hat{r}\leq r^*\leq \sqrt{d}\hat{r}$.
    For every $r'\geq \hat{r}$, we have $r=\sqrt{d}r'\geq \sqrt{d}\hat{r}\geq r^*$.  Hence, from the analysis of Section~\ref{sec:kCenter}, the main algorithm returns a non-empty set $\ret'$ of (at most) $k$ points such that $\kcenter_{\ret'}(\Q(\I))\leq 2(1+\eps)r$.
    %, with probability at least $1-\frac{1}{N^2}$. %Now consider $r'=\hat{r}$. In this case we have $r=\sqrt{d}\hat{r}\leq  \sqrt{d}r^*$. Again, our relational algorithm will return a set $\ret'$ of $k$ points that covers all points in $\Q(\I)$ within distance $2(1+\eps)r\leq 2\sqrt{d}(1+\eps)r^*$ so the binary search will try to search smaller $r'$.
    Let $r'\leq \hat{r}$ be the distance in the last valid iteration of the binary search where a set $\ret\neq \emptyset$ was returned. We have $r=\sqrt{d}r'$. In this case, after at most $k$ iterations,   $\tree.\representative()=\emptyset$
   (equivalently, all points in $\Q(\I)$ are inactive). Since in each iteration
    we make all points within distance $2r$ (and some points within distance at most $2(1+\eps)r$) inactive around every $p_i\in \ret$, it means that all points in $\Q(\I)$ lie within distance $2(1+\eps)r$ from $\ret$.
    Hence, we have $\kcenter_{\ret}(\Q(\I))\leq 2(1+\eps)r= 2\sqrt{d}(1+\eps)r'\leq 2\sqrt{d}(1+\eps)\hat{r}\leq 2\sqrt{d}(1+\eps)r^*$. The lemma follows.
    %\aryan{Why did we separate the cases $r' = \hat{r}$ and $r' > \hat{r}$. Why not have only two cases where $r' \geq \hat{r}$ and $r' < \hat{r}$?}
%    Notice that in every iteration of our algorithm we define a ball $\ball(p_i,2r)$ and it holds that $\ball(p_i,2r^*)\subseteq \ball(p_i,2r)\subseteq \ball(p_i,2\sqrt{d}r^*)$. Hence, if $\ret$ is the last valid solution computed by our algorithm, we have that $\kcenter_\ret(\Q(\I))\leq 2\sqrt{d}(1+\eps)\kcenter_{\opt(\Q(\I))}(\Q(\I))$.
\end{proof}
Finally, as an estimation of $\kcenter_{\ret}(\Q(\I))$, we set $r_\ret=2(1+\eps)r$, where $r$ is the number that was given as input to the main algorithm that returned the final set $\ret$.
Notice that $\kcenter_\ret(\Q(\I))\leq r_\ret\leq 2\sqrt{d}(1+\eps)\kcenter_{\opt(\Q(\I))}(\Q(\I))$. %As we show later, $r_\ret$ is useful in order to get a better approximation algorithm for the relational $k$-center clustering. Furthermore, $r_\ret$ is useful in order to derive a constant approximation for the relational $k$-median/means clustering as shown in Section~\ref{}.

The running time of the algorithm using the binary search over the $\ell_\infty$ distances is only larger by a $\log N$ factor from the running time of the main algorithm. Indeed, the binary search explores $O(\log N)$ distances, and it spends $O(N\log N)$ time to compute each explored distance using the algorithm from~\cite{salowe1989infinity}. Hence the overall running time is $$O\left(N\log^2 N + k\left(\eps^{-d+1}+\log N\right)\left(N+\eps^{-2}\log^2 N\right)\log N\right)=O(k\cdot N\cdot\log^2 N+k\cdot \log^4 N)=\O\left(k\cdot N\right),$$
with probability at least $1-\frac{1}{N}$.
\end{proof}

\subsection{Alternative $(2+\eps)$-approximation}
\label{apndx:simple2approx}
A $(2+\eps)$-approximation algorithm can be derived directly from our main algorithm in Section~\ref{sec:kCenter}, without computing the coreset from~\cite{agarwal2024computing}. 
Let $r_\ret$ be the value returned by Theorem~\ref{thm:constantkCenter}.
It is clear that $\frac{r_\ret}{2\sqrt{d}+\eps}\leq \kcenter_{\opt(\Q(\I))}(\Q(\I))\leq r_\ret$. We discretize the range $[\frac{r_\ret}{2\sqrt{d}+\eps}, r_\ret]$ by a multiplicative step $(1+\eps)$, creating the set $\Psi=\{\frac{r_\ret}{2\sqrt{d}+\eps} (1+\eps)^0, \frac{r_\ret}{2\sqrt{d}+\eps}(1+\eps)^1,\ldots, (1+\eps)^\psi\frac{r_\ret}{2\sqrt{d}+\eps}\}$, where $\psi=O(\frac{1}{\eps}\log(2\sqrt{d}+\eps))$. We run a binary search on $\Psi$. For every $r\in\Psi$ selected by the binary search, we execute the main algorithm from Section~\ref{sec:kCenter}. We continue to larger or smaller values of $r$ in $\Psi$ as shown in the description. In the end, it is straightforward to see that in the last step of the binary search where a set $\ret'$ is returned for a distance $r\in \Psi$, 
$r\leq (1+\eps)\kcenter_{\opt(\Q(\I))}(\Q(\I))$, so
$\kcenter_{\ret'}(\Q(\I))\leq 2(1+\eps)^2\kcenter_{\opt(\Q(\I))}(\Q(\I))$
%if we set $r_{\ret'}=2(1+\eps)r$ it holds that $\kcenter_{\ret'}(\Q(\I))\leq r_{\ret'}\leq 2(1+\eps)^2\kcenter_{\opt(\Q(\I))}(\Q(\I))$.
For $\eps\leftarrow \eps/5$ we have that $(1+\eps/5)^2\leq 1+\eps$ hence $\kcenter_{\ret'}(\Q(\I))\leq r_{\ret'}\leq (2+\eps)\kcenter_{\opt(\Q(\I))}(\Q(\I))$.
Notice that it is not necessary to construct the set $\Psi$ explicitly (i.e., the $j$-th smallest distance in $\Psi$ can be computed in $O(1)$ time). Instead, we only run $O(\log\frac{\log(2\sqrt{d}+\eps)}{\eps})$ additional steps of our $k$-center algorithm. For an arbitrarily small constant $\eps\in(0,1)$, the total running time remains $O(kN\log^2 + k\log^4 N)$, with probability at least $1-\frac{1}{N}$.

\begin{algorithm}[t]
\caption{{\sc Rel-k-Means-Approx}$(\Q, \I, k, \varepsilon)$}
\label{alg:RelKMeansApprox}
    Compute $|\Q(\I)|$ using Yannakakis algorithm\;
    Build the root of RBBD tree $\tree$ over $\Q(\I)$\;
    
    $(V, L) \leftarrow$ {\sc Rel-k-Center}$(\Q, \I, k, \varepsilon)$; $M \leftarrow 2\log(|\Q(\I)|) + 3$\;
    
    $a_0 \leftarrow 0$; $b_0 \leftarrow \frac{L}{4|\Q(\I)|}$; $a_\infty \leftarrow 2L|\Q(\I)|$; $b_\infty \leftarrow \infty$\;

    $a_j \leftarrow 2^{j-1}\frac{L}{4|\Q(\I)|}$, $b_j \leftarrow 2^j\frac{L}{4|\Q(\I)|} ~~~\forall j\in[M]$\;

    $\ret \leftarrow \emptyset$; $r_{\ret} \leftarrow 0$; $\tau_1 \leftarrow |\Q(\I)|$; $i \leftarrow 1$\;
    
    \While{$\tau_i > 480 \cdot k\log^2 |\Q(\I)|$}{
        $X_i \leftarrow \tree.\text{sample}(240k\log^2 |\Q(\I)|)$\;
        $\ret \leftarrow \ret \cup X_i$\;
        
        $\tree' \leftarrow \text{copy}(\tree)$\;
        
        \For{$j \in \{0, 1, \ldots, M, \infty\}$}{
            $\counts_{i,j} \leftarrow 0$\;
            \ForEach{$x \in X_i$}{
                $\counts_{i,j} \leftarrow \counts_{i,j} + \tree'.\text{count}(x, b_j)$\;
                $\tree'.\text{inactive}(x, b_j)$\;
            }
        }
        
        $\beta_i \leftarrow \max\{j : \counts_{i,j} > \frac{\tau_i}{10\log |\Q(\I)|}\}$\;
        Reset $\tree$ to the state of $\tree'$ before defining $\counts_{i,\beta_i+1}$\;
        
        $\tau_{i+1} \leftarrow \tau_i - \sum_{\ell \in [\beta_i] \cup \{0\}} \counts_{i,\ell}$\;
        
        $r_{\ret} \leftarrow r_{\ret} + \sum_{\ell \in [\beta_i] \cup \{0\}} \counts_{i,\ell} \cdot ((1+\varepsilon)b_\ell)^2$\;
        
        $i \leftarrow i + 1$\;
    }
    
    $\ret \leftarrow \ret \cup \tree.\text{report}()$\;
    \Return $\ret$\;
\end{algorithm}

\section{Missing proofs from Section~\ref{sec:kmeans}}
\label{appndx:kmeans}

\begin{proof}[Proof of Lemma~\ref{lem:help1}]
%[Proof of Lemma~\ref{lem:help1}]
From Lemma~\ref{lem:oracle5}, we note that the set $X_i$ is a set of $\O(k)$ samples selected from $Q'_i$ uniformly at random.
For every $j\in\{0,1,2,\ldots, M, \infty\}$, we execute $\tree'.\inactive(x,b_j)$ over every $x\in X_i$. Let $p\in Q'_i$ be a point that was active at the moment that $c_{i,j}$ was defined and became inactive after applying $\tree'.\inactive(x,b_j)$ for every $x\in X_i$. By the correctness of the $\tree'.\inactive(x,b_j)$ oracle (Lemma~\ref{lem:oracle1}) $p\notin \bigcup_{\ell<j}Q_{i,\ell}'$ and $p\in Q_{i,j}'$. Furthermore, $p\notin \bigcup_{\ell>j}Q_{i,\ell}'$ because a point becomes inactive once. Hence for every pair $j_1,j_2\in [M]\cup\{0,\infty\}$, $Q_{i,j_1}'\cap Q'_{i,j_2}
=\emptyset$. Furthermore, when $j=\infty$, all remaining active points, by definition, will become inactive. Hence $\bigcup_{j\in [M]\cup\{0,\infty\}}Q_{i,j}'=Q'_i$. So property (1) holds.

Next, if we fix an index $j\in [M]\cup\{0,\infty\}$, for every $x\in X_i$, we execute $\counts_{i,j}=\counts_{i,j}+\tree'.\coun(x,b_j)$ and then $\tree'.\inactive(x,b_j)$. By the proof of Lemma~\ref{lem:oracle1} and Lemma~\ref{lem:oracle2}, the $\tree'.\coun(x,b_j)$ oracle returns exactly the number of points in $Q'_i$ that become inactive when we execute $\tree'.\inactive(x,b_j)$. Hence, it also holds that $\counts_{i,j}=|Q_{i,j}'|$, and property (3) holds.

Finally, we show property (2).
Let $j=0$, and let $x_1\in X_i$ be the first sample we process. We execute $\tree'.\inactive(x_1,b_0)$. By Lemma~\ref{lem:oracle1}, this oracle makes all points in $Q'_i$ within distance $b_0$ from $x_1$ inactive (it might also make some points within distance $(1+\eps)b_0$ inactive). Similarly, for any other $x_s\in X_i$, the oracle $\tree'.\inactive(x_s,b_0)$ makes inactive all currently active points within distance $b_0$ from $x_s$ (it might also make some points within distance $(1+\eps)b_0$ inactive). Hence, by the end of the process for $j=0$, all points in $Q'_{i,0}$ within distance $b_0$ from a point in $X_i$ become inactive with respect to $\tree'$. Some points within distance $(1+\eps)b_0$ from $X_i$ also might have become inactive. Hence $Q_{i,0}'$ satisfies property (2). By strong induction assume that every set $Q_{i,j}'$ for each $j<j^*$ satisfies property (2), where $j^*\in [M]\cup\{\infty\}$. For every $x\in X_i$, the oracle $\tree'.\inactive(x,b_{j^*})$ makes inactive all currently active points in $Q'_i$ that are within distance $b_{j^*}$ from $x$ (and might make some points inactive within distance $(1+\eps)b_{j^*}$). Let $Q_{i,j^*,x}'$ be this set of points and hence we have $Q_{i,j^*}'=\bigcup_{x\in X}Q_{i,j^*,x}'$.
From property (1), we have that $Q_{i,j^*}'\cap (\bigcup_{j<j^*}Q_{i,j}')=\emptyset$. By the induction hypothesis the set $\bigcup_{j<j^*}Q_{i,j}'$ contains all points in $Q'_i$ within distance $b_{j^*-1}$ from $X_i$ and might contain some points within distance $(1+\eps)b_{j^*-1}$ from $X_i$. Hence, $Q'_{i,j^*}$ contains all points with distance greater than $(1+\eps)b_{j^*-1}=(1+\eps)a_{j^*}$ and at most $b_j$ from $X_i$, might contain some points within distance at most $(1+\eps)b_{j^*}$ from $X_i$, and might contain some points with distance greater than $b_{j^*-1}=a_{j^*}$ from $X_i$. Hence, $Q'_{i,j^*}$ satisfies property (2).

%By the proofs of lemmas in Section~\ref{sec:reloracles}, each oracle returns the correct results with probability at least $1-\frac{1}{N^{m+2}}$. We execute $\O(k)$ oracles, so the correctness of the lemma follows with probability at least $1-\frac{1}{N^2}$.
%As we show in Lemma~\ref{}, the algorithm executes $O(\log N)$ iterations so al operations hold with probability at least $1-\frac{1}{N^{\Omega(1)}}$.
%Let $p$ be a point stored in a leaf node $u$ the RBBD tree (such node might have not been constructed yet. A point $p$ belongs in $Q'$ if and only if there is no inactive node from the root to $u$. Assume that $v$ is the last node in the path from $\root$ to $u$ that was constructed in $\tree$ in the current iteration. Let $\root\rightarrow v_1\rightarrow\ldots\rightarrow v_\kappa\rightarrow v$. Let $w_1,\ldots, w_\kappa$ be the siblings of nodes $v_1,\ldots, v_\kappa$, respectively. From Section~\ref{}, we know that a point in a cell $\square_{v}$ is sampled with probability $\frac{1}{|\square_v\cap \Q(\I)|}$. The probability that $p$ is selected in $X$ is $\prod_{j\in [\kappa]}\frac{v_j.\mathsf{count}}{v_j.\mathsf{count}+w_j.\mathsf{count}}\cdot \frac{1}{|\square_v\cap \Q(\I)|}$. Notice that $v$ is active and no children of $v$ have been constructed in the tree, so $\square_v\cap \Q(\I)=\square_v\cap Q'$. Furthermore, by the correctness of the counts, $v_1.\mathsf{count}+w_1.\mathsf{count}=|Q'|$. We conclude that the probability of selecting $p\in Q'$ is $\frac{1}{|Q'|}$. The first part of the lemma follows.

\end{proof}

\begin{proof}[Proof of Lemma~\ref{lem:bad}]
We fix the iteration $i$ of the algorithm.
    For every $p_s\in\opt(\Q(\I))$, we place the smallest ball $B_s$ with center $p_s$ that contains $\frac{\tau_i}{20\cdot k\cdot \log|\Q(\I)|}$ points from $Q'_i$.
    %If $c_1$ is large enough, i
%    The set $X_i$ is correctly selected uniformly at random with probability at least $1-\frac{1}{N^{m+2}}$.
    Since $X_i\subset Q_i'$ and $|X_i|=\Theta(k\log^2|\Q(\I)|)$, we have that there is at least one point of $X_i$ inside every ball $B_s$, with probability at least $1-\frac{1}{N^3}$~\cite{har2004coresets, har2011geometric}.
     Namely, $X_i\cap B_s\neq \emptyset$ for each $s\in [k]$.
    Let $x_s\in X_i$ denote the point that lies in the ball $B_s$, for each $s\in[k]$.
Let $p\in Q'_i\setminus ((\bigcup_{\ell\in[k]}B_\ell)\cap Q'_i)$ be a point outside of the balls $\bigcup_{\ell\in [k]}B_\ell$. Let $p_s$ be the closest center of $\opt(\Q(\I))$ to $p$, i.e., $\dist(p,\opt(\Q(\I)))=\dist(p,p_s)$. Since $x_s\in B_s\cap Q'_i$, notice that $\dist(p,x_s)\leq 2\dist(p,p_s)$, so $p\in Q'_i\setminus \badpoints$, i.e., $p$ is a good point.
Thus, with probability at least $1-\frac{1}{N^3}$, all bad points in $Q'_i$ lie inside the balls $B_1,\ldots, B_k$. But every one of those balls, contain at most $\frac{\tau_i}{20k\log |\Q(\I)|}$ points of $Q'_i$. It follows that  $|\badpoints|\leq k\cdot \frac{\tau_i}{20k\log |\Q(\I)|}=\frac{\tau_i}{20\log |\Q(\I)|}$, with probability at least $1-\frac{1}{N^3}$.
\end{proof}

\begin{proof}[Proof of Lemma~\ref{lem:helpred}]
   Clearly, $|\hat{Q}_i|\geq \tau_i-\counts_{i,\infty}-M\cdot \frac{\tau_i}{10\log |\Q(\I)|}$, since $\beta_i$ is the largest index such that $Q'_{i,\beta_i}>\frac{\tau_i}{10\log |\Q(\I)|}$.
   %, with probability at least $1-\frac{1}{N^{2}}$. 
    Notice that 
    %with probability  at least $1-\frac{1}{N^{2}}$, 
    for every point $p\in Q'_{i,\infty}$, it holds $\dist(p,X_i)\geq 2L|\Q(\I)|$. However $\kmeans_{\opt(\Q(\I))}(\Q(\I))\leq |\Q(\I)| L^2$ so $\dist(p,\opt(\Q(\I)))\leq \sqrt{|\Q(\I)|}L$ and hence $p\in \badpoints$. Thus, $Q'_{i,\infty}\subseteq \badpoints$.
Furthermore, from Lemma~\ref{lem:bad} we have $|\badpoints|\leq \frac{\tau_i}{20\log|\Q(\I)|}$ with probability at least $1-\frac{1}{N^3}$. Hence, $|\hat{Q}_i|\geq \tau_i - \frac{\tau_i}{20\log |\Q(\I)|} - (2\log |\Q(\I)| +3)\frac{\tau_i}{10\log|\Q(\I)|}\geq \frac{\tau_i}{2}$. So, in every iteration, at least half of the remaining points become inactive. We conclude that the algorithm runs for $O(\log |\Q(\I)|)=O(\log N)$ iterations with probability at least $1-\frac{1}{N^2}$. In each iteration we add in $\ret$ the set $X_i$ and $|X_i|=O(k\log^2 N)$, so $|\ret|=O(k\log^3 N)$, with probability at least $1-\frac{1}{N^2}$.
\end{proof}

\begin{proof}[Proof of Lemma~\ref{lem:cost}]

%We first prove the approximation factor. From previous lemmas, all arguments hold in each iteration of the algorithm with probability at least $1-\frac{1}{N^2}$. There are $O(\log N)$ iterations, so the approximation factor we compute hold with probability at least $1-\frac{1}{N}$.
We first prove the approximation factor.
    We fix the iteration $i$ of the algorithm. 
    If $\beta_i>0$, then we argue as follows. Let $\P=Q_{i,\beta_i}'\setminus \badpoints$. For any point $p\in \hat{Q}_i\cap \badpoints$ and $q\in \P$, we have $\dist(p,X_i)\leq 2(1+\eps)\dist(q,X_i)$. This inequality follows from the definition of the sets $Q'_{i,j}$ and the $\tree'.\inactive(\cdot,\cdot)$ oracle. Furthermore, $|\P|>\frac{\tau_i}{10\log |\Q(\I)|}-\frac{\tau_i}{20\log |\Q(\I)|}=\frac{\tau_i}{20\log |\Q(\I)|}\geq |\hat{Q}_i\cap \badpoints|$, with probability at least $1-\frac{1}{N^2}$.
    By definition, every point in $\P$ is a good point with respect to $X_i$ so $\kmeans_{X_i}(\P)\leq 4\kmeans_{\opt(\Q(\I))}(\P)$. Furthermore, since $\P\subseteq \hat{Q}_i$ we have that $\kmeans_{\opt(\Q(\I))}(\P)\leq \kmeans_{\opt(\Q(\I))}(\hat{Q}_i)$.
    Thus,
    $$\kmeans_{X_i}(\hat{Q}_i\cap \badpoints)\leq 4(1+\eps)^2\kmeans_{X_i}(\P)\leq 16(1+\eps)^2\kmeans_{\opt(\Q(\I))}(\P)\leq 16(1+\eps)^2\kmeans_{\opt(\Q(\I))}(\hat{Q}_i).$$
    So,
    \begin{align*}\kmeans_{X_i}(\hat{Q}_i)&=\kmeans_{X_i}(\hat{Q}_i\cap \badpoints)+\kmeans_{X_i}(\hat{Q}_i\setminus \badpoints)\leq 16(1+\eps)^2\kmeans_{\opt(\Q(\I))}(\hat{Q}_i)+4\kmeans_{\opt(\Q(\I))}(\hat{Q}_i)\\&\leq 20(1+\eps)^2\cdot \kmeans_{\opt(\Q(\I))}(\hat{Q}_i).
    \end{align*}

    If $\beta_i=0$, then for every point $q\in \hat{Q}_i$ we have $\dist(q,X_i)\leq (1+\eps)\frac{L}{4|\Q(\I)|}$. Thus, for $\eps\leq 0.1$,
    $$\kmeans_{X_i}(\hat{Q}_i)=\sum_{q\in \hat{Q}_i}\dist^2(q,X_i)\leq (1+\eps)^2\frac{L^2}{16|\Q(\I)|}< \frac{L^2}{9|\Q(\I)|}\leq \frac{1}{|\Q(\I)|}\kmeans_{\opt(\Q(\I))}(\Q(\I)).$$ 
    %The last inequality holds because $V\subseteq \hat{Q}$.
    In any case, we have $\kmeans_{X_i}(\hat{Q}_i)\leq \max\{20(1+\eps)^2\cdot\kmeans_{\opt(\Q(\I))}(\hat{Q}_i),\frac{1}{|\Q(\I)|}\kmeans_{\opt(\Q(\I))}(\Q(\I))\}$.
    
    By definition, recall that $\hat{Q}_{i_1}\cap \hat{Q}_{i_2}=\emptyset$ for every $i_1\neq i_2$. Overall, $\kmeans_{\ret}(\Q(\I))\leq 21(1+\eps)^2\kmeans_{\opt(\Q(\I))}(\Q(\I))$, with probability at least $1-\frac{1}{N}$.
    %Since every operation holds with probability at least $1-\frac{1}{N^{\Omega(1)}}$, the bound holds with probability at least $1-\frac{1}{N^{\Omega(1)}}$.

    Finally, for every point $p\in\hat{Q}_i$ such that $p\in Q_{i,j'}'$ for an index $j'\leq \beta_i$
    we charge $((1+\eps)b_{j'})^2$ in $r_\ret$. By definition, if $j'>0$,
    $\dist(p,X_i)\leq (1+\eps)b_{j'}\leq 2(1+\eps)\dist(p,X_i)$.
    Hence the $k$-means cost of $p$ with respect to $X_i$ is charged at most $4(1+\eps)^2\dist^2(p,X_i)$ and at least $\dist^2(p,X_i)$ in $r_\ret$.
If $j'=0$, then by the bound of $\kmeans_{X_i}(\hat{Q}_i)$ when $\beta_i=0$, we charge a value which is at most $\frac{1}{|\Q(\I)|}\kmeans_{\opt(\Q(\I))}(\Q(\I))$ and at least $\kmeans_{\ret}(\hat{Q})$ in $r_\ret$.
Generally, for all points in $\hat{Q}_i$, we charge a value at most $\sum_{p\in\hat{Q}_i}4(1+\eps)^2\dist^2(p,X_i)\leq 84(1+\eps)^4\kmeans_{\opt(\Q(\I))}(\hat{Q}_i)$ and at least $\sum_{p\in\hat{Q}_i}\dist(p,X_i)\geq\kmeans_{\ret}(\hat{Q}_i)$ in $r_\ret$. Overall, $\kmeans_{\ret}(\Q(\I))\leq r_\ret\leq 84(1+\eps)^4 \kmeans_{\opt(\Q(\I))}(\Q(\I))$.
%, with probability at least $1-\frac{1}{N^{\Omega(1)}}$.
By setting $\eps\leftarrow \eps/400$, we get $\kmeans_{\ret}(\Q(\I))\leq r_\ret\leq (84+\eps) \kmeans_{\opt(\Q(\I))}(\Q(\I))$, wuit probability at least $1-\frac{1}{N}$.

Next, we proceed with the running time.
First, we run the $k$-center algorithm from Theorem~\ref{thm:kCenter} in $O\left(kN\log^2 N + k\log^4 N\right)$ time with probability at least $1-\frac{1}{10N}$.
Then the algorithm runs for $O(\log N)$ iterations with probability at least $\frac{1}{N^2}$. For every iteration $i$, we get the set $X_i$ in $O(kN\log^2 N)$ time, with probability at least $1-\frac{1}{N^{m+3}}$, using the result of Lemma~\ref{lem:oracle5}.
    There are $O(\log N)$ different values of $\ell\in[M]\cup\{0,\infty\}$, and for each $\ell$ we compute $\counts_{i,\ell}$ in $O(kN\log^3 N +k\log^5 N)$ time with probability at least $1-\frac{O(\log^2N)}{N^3}$, 
    executing the $\tree'.\inactive()$ oracles $O(k\log^2 N)$ times (recall from Lemma~\ref{lem:oracle2} that $\tree'.\inactive()$ runs in $O(N\log N +\log^3 N)$ time).
    %using Lemma~\ref{lem:oracle2} over $O(k\log^2 N)$ executions of the $\tree'.\inactive()$ oracles.
   % All values $\counts_{i,\ell}$ are computed in $O()$
   Fixing an iteration $i$, all values $\counts_{i,\ell}$ are computed in $O(kN\log^4 N +k\log^6 N)$ time with probability at least $1-\frac{O(\log^3 N)}{N^3}\geq 1-\frac{1}{N^2}$.
   Overall, during the $O(\log N)$ iterations (with probability at least $1-\frac{1}{N^2}$), the running time to compute all counts is $O(kN\log^5 N +k\log^7 N)$ with probability at least $1-\frac{2}{N^2}$.
   %using Lemma~\ref{lem:oracle2}.
%    Furthermore, the time we need to construct a copy of $\tree$ is $O(kN\log^2 N +k\log^4 N)$.
  %  Thus, in each iteration the $k$-center algorithm from Theorem~\ref{thm:kCenter} dominates the running time.
  Similarly, in the last iteration, denoted by $i'$, notice that $|\tree|=O(k\log^5 N)$ with probability at least $1-\frac{2}{N^{2}}$. From Lemma~\ref{lem:oracle4}, and the fact that $|\tau_{i'}|=O(k\log^2 N)$, the oracle $\tree.\report()$ runs in $O(kN\log^5 N)$ time, with probability at least $1-\frac{3}{N^{2}}$
    Overall, our algorithm runs in $O(kN\log^5 N+k\log^7 N)$ time, with probability at least $1-\frac{1}{N}$.
\end{proof}

\newcommand{\diam}{\mathsf{diam}}
\subsection{Faster coreset construction}
\label{appndx:fastercoreset}
Before we describe the faster coreset construction from~\cite{esmailpour2024improved}, we briefly show that the oracles defined in Section~\ref{sec:reloracles} can also be defined for boxes, instead of balls.

In the standard computational setting BBD tree can also be used for reporting, counting and sampling queries when the query shape is not a ball but a box $\rho$ in $\Re^d$. The query procedure $\mathcal{T}(\rho)$ returns a set of $O(\log n + \eps^{-d+1})$ canonical nodes $\mathcal{U}(\rho)$ such that $\rho\subseteq \bigcup_{u\in \mathcal{U}(\rho)}\square_u\subseteq \rho\oplus \eps_\rho$, where $\eps_\rho=\eps\cdot\mathsf{diam}(\rho)$, $\mathsf{diam}(\rho)$ is the diameter of $\rho$, and $\rho\oplus \eps_\rho$ is a the Minkowski sum of $\rho$ with $\eps_\rho$, i.e., $\rho\oplus \eps_\rho$ is a superset of $\rho$ such that every point of $\rho\oplus \eps_\rho$ is within distance $\eps\cdot\mathsf{diam}(\rho)$ from its closest point in $\rho$.

In the relational setting, having a partially constructed RBBD tree, given a box $\square$, the oracle $\tree.\inactive(\square)$ makes inactive all points in $\bigcup_{u\in \canonical(\square)}(\square_u\cap \Q(\I))$, similarly to the inactive oracle in Section~\ref{sec:reloracles}. 
Similarly, $\tree.\coun(\square)$ counts all active points in $\bigcup_{u\in \canonical(\square)}|\square_u\cap \Q(\I))|$,
Furthermore, we define the oracle $\tree.\representative(\square)$ that returns an active point from $\square\cap \Q(\I)$. We run the $\tree(\square)$ query constructing new nodes if necessary (as shown in Section~\ref{sec:reloracles}), skipping nodes that are already inactive. We compute $\canonical(\square)$ and for a node $u\in \canonical(\square)$ we return $u.rep$.
Using the definitions of the RBBD tree (and BBD tree) if $Q'$ is the set of active points, $\tree.\inactive(\square)$ makes inactive the set of points $\bigcup_{u\in \canonical(\square)}\square_u\cap Q'$, such that $\square\cap Q'\subseteq \bigcup_{u\in\canonical(\square)}(\square_u\cap Q')\subseteq (\square\oplus \eps_{\square}\cap Q')$.
For count, we notice that the oracle returns a number $\sum_{u\in\canonical(\square)}|\square_u\cap Q'|$ such that $|\square\cap Q'|\leq \sum_{u\in\canonical(\square)}|\square_u\cap Q'|\leq |\square\oplus\eps_{\square}\cap Q'|$.
Similarly, $\tree.\representative(\square)$ returns an active point $p\in Q'$ such that $p\in \square\oplus\eps_{\square}\cap Q'$. 
The running time of each new  oracle is $O((\eps^{-d+1}+\log N)(N+\eps^{-2}\log^2 N))$, with probability at least $1-\frac{1}{N^{m+3}}$.

The coreset construction from~\cite{esmailpour2024improved} requires the following operation.
Let $\square_1$ be a square-box (all sides are equal) in $\Re^d$. For $\square_1$ we should get a representative point from $\square_1\cap \Q(\I)$ and compute the count $\square_1.c=|\square_1\cap \Q(\I)|$. Repeatedly, for every new square-box $\square_i\in \Re^d$, get a representative point from $(\square_i\setminus(\bigcup_{j<i}\square_j))\cap \Q(\I)$ and compute the count $\square_i.c=|(\square_i\setminus(\bigcup_{j<i}\square_j))\cap \Q(\I)|$. In other words, every time we see a new square-box $\square_i$ the goal is to count the number of points (and get a representative) in $\Q(\I)$ that belong in $\square_i$ but do not belong in any of the previously defined square-boxes $\square_j$ for $j<i$.
While our goal is not to fully describe their coreset construction, intuitively, every square-box $\square_i$ they consider creates a new point in the coreset $\coreset$ (the representative point) with weight $\square_i.c$.
If $\ret$ is the set of points returned in the first phase, they get $O(|\ret|\eps^{-d}\log N)$ square-boxes and thus their coreset size is $O(|\ret|\eps^{-d}\log N)$.
In~\cite{esmailpour2024improved}, the authors compute an estimation of $\square_i.c$ by applying uniform random sampling in $\square_i$. However, the number of samples they get for every square-box in order to have an accurate estimation is roughly $\O(|\ret|)=\Omega(k)$, leading to an additional $k^2$ term in the running time, skipping the additional $\log N$ factors. Using our new RBBD construction in the relational setting, the operation of computing $\square_i.c$ becomes simpler. 
Initially, let $\tree$ consist of the root of an RBBD tree over $\Q(\I)$, as described in Section~\ref{subsec:RBBDrel}. For every new square-box $\square_i$ we consider, we execute $\tree.\coun(\square_i)$ to get the count, $\tree.\representative(\square_i)$ to get a representative point, and $\tree.\inactive(\square_i)$ to make make all points in $\bigcup_{u\in \canonical(\square_i)}\square_u\cap \Q(\I)$ inactive.
Such an algorithm constructs the coreset in $O((\eps^{-d+1}+\log N)(N+\eps^{-2}\log^2 N)|\ret|\eps^{-d}\log N)$ time, with high probability. In our case $|\ret|=O(k\log^3 N)$ so for arbitrarily small constant $\eps$, the running time is
$O(kN\log^5 N + k\log^{7}N)$.
Finally, we notice that in~\cite{esmailpour2024improved} the authors only compute an estimation of $\square_i.c$. Using the RBBD tree we also do not compute the count exactly. Instead, using the guarantees of the RBBD tree, we compute the number of points in each square-box assuming that the boundaries of each box are off by an additional term $\eps\cdot \diam$. Recall from our discussion above that each square-box $\square_i$ is approximated by $\square_i\oplus \eps_{\square_i}$. However, this only increases the overall cost by a $(1+\eps)^2$ factor (since the boxes are square-boxes the diameter defines the maximum distance between two points in each square-box, which is used in~\cite{esmailpour2024improved} to bound the overall error). Setting $\eps\leftarrow\eps/4$, we get exactly the same approximation as in~\cite{esmailpour2024improved}.

\input{appndxCat}
\input{apnndxGonzalez}

\input{appndxCyclic}

%% file: figdecomp.tex
\begin{figure}[t]
    \centering
    \includegraphics[width=0.2\textwidth]{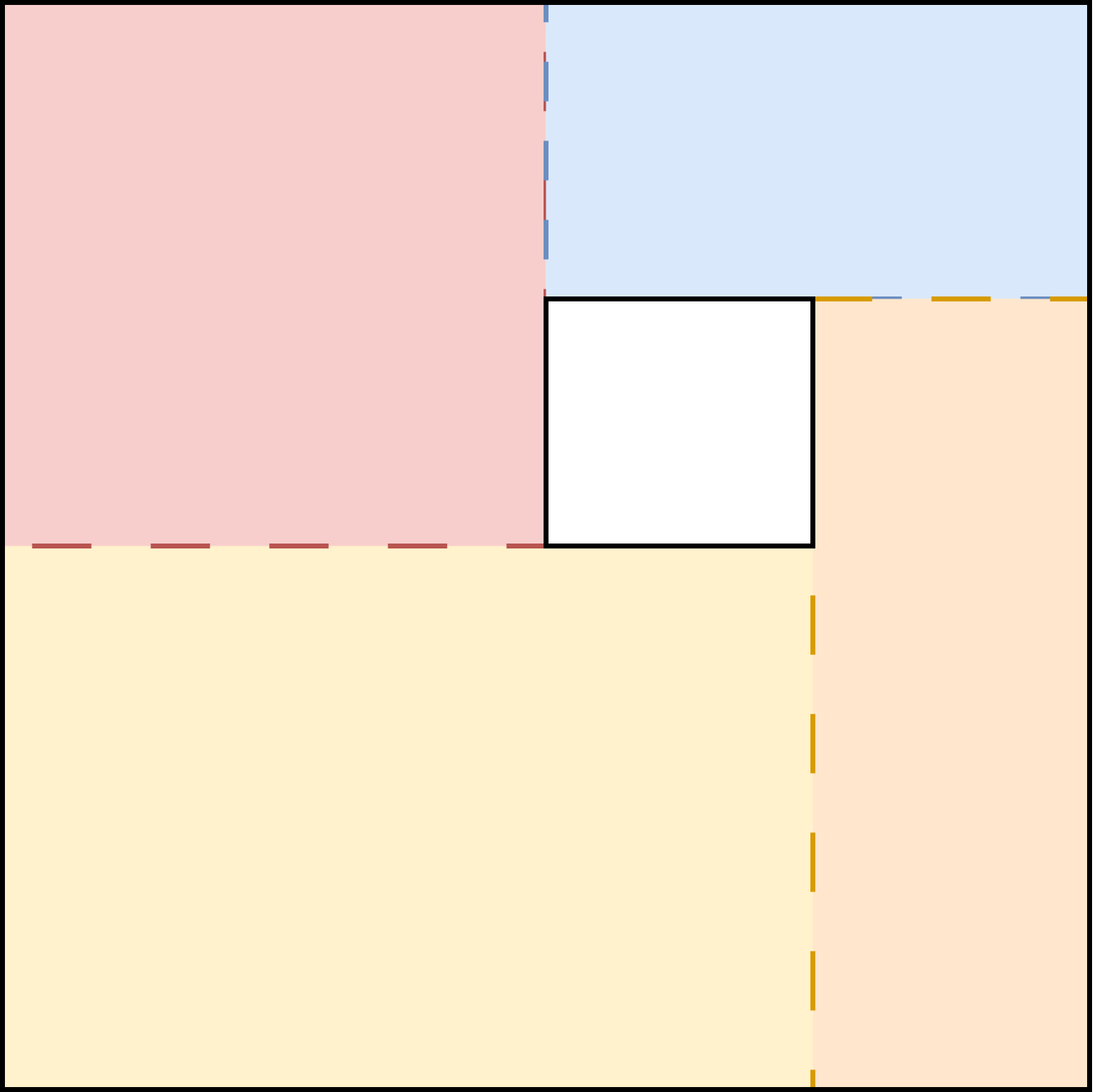}
    \caption{Partition of a 2-d box with a hole into 4 boxes.}
    \label{fig:decomp}
\end{figure}

%% file: appndxCat.tex
\input{categorical}

%\section{Missing details from Section~\ref{sec:cat}}
%\label{appndx:catdesc}
%Before we describe our algorithm, we briefly review the algorithm of~\cite{curtin2020rk}.

\begin{envrevone}
\subsection{Approximate Closest Centers}
\label{appndx:cat}
Recall that $\bar{\Q}(\I) = \bar{\pi}_{\allattr_{(n)}}(\Q(\I))$, and note that $|\bar{\Q}(\I)| = |\Q(\I)|$. 

We design an algorithm that efficiently computes an estimation $\bar{\omega}(p)$ of the weight $\omega(p)$ for all $p\in C$.
We use the oracles introduced in Section~\ref{sec:reloracles}, with a few minor modifications. Recall that for each node $u$ of a partially constructed RBBD-tree $\tree$, we store the corresponding box (or box with a hole) $\square_u$, the binary variable $u.a$, the representative point $u.\mathrm{rep}$, and the counter $u.c$. In our modified version, we additionally store and maintain $k^{d_c}$ auxiliary counters at every constructed node of $\tree$.
More specifically, for each node $u \in \tree$, we maintain counters of the form $u(i_1,i_2,\ldots,i_{d_c})$, where $1 \leq i_j \leq k$ for every $j \in [d_c]$. Each counter is defined as 
$$u(i_1,i_2,\ldots,i_{d_c})\!=\!|\{t \in \Q(\I) \!\mid\! \pi_{\allattr_{(n)}}(t) \in \square_u, \forall A^{(j)} \!\in\! \allattr_{(c)}, \dist(\pi_{A^{(j)}}(t), C_{A^{(j)}}) = \dist(\pi_{A^{(j)}}(t), C_{A^{(j)}}[i_j])\}|,$$ 
where $C_{A^{(j)}}[i_j]$ denotes the $i_j$-th point in $C_{A^{(j)}}$ (under an arbitrary but fixed ordering). 
Intuitively, each counter $u(i_1,i_2,\ldots,i_{d_c})$ stores the number of tuples in $\Q(\I)$ whose projection on the numerical attributes lies within $\square_u$, and whose projection on each categorical attribute $A^{(j)}$ is closest to the $i_j$-th center in $C_{A^{(j)}}$.

We first describe how to compute $u(i_1,i_2,\ldots,i_{d_c})$ for a newly constructed node $u \in \tree$. 
Assume that $\square_u$ is a box without a hole; if it contains a hole, we simply repeat the procedure for each of the $2d_n$ non-intersecting boxes, as explained in Appendix~\ref{appndx:randcentshr}. 
Following the approach of Lemma~\ref{lem:Rects} (originally described in~\cite{agarwal2024computing}), for every relation $R \in \allrel$ we identify all tuples $t \in R$ such that: for every $A \in \allattr(R) \cap \allattr_{(n)}$, $\pi_A(t) \in \pi_A(\square_u)$, and for every $A^{(j)} \in \allattr(R) \cap \allattr_{(c)}$, $\pi_{A^{(j)}}(t)$ has the $i_j$-th center in $C_{A^{(j)}}$ as its closest center. All other tuples are discarded. Let $R' \subseteq R$ denote the remaining tuples in relation $R$. We then execute the counting version of Yannakakis' algorithm on the join $\Join_{R \in \allrel} R'$ and set $u(i_1,i_2,\ldots,i_{d_c}) = |\Join_{R \in \allrel} R'|$ in $O(N)$ time. The correctness follows directly from the definition of $R'$. Since Yannakakis' algorithm is executed $k^{d_c}$ times, constructing a new node in $\tree$ takes $O(k^{d_c} N)$ time.
We update the value of $u(i_1,i_2,\ldots,i_{d_c})$ similarly to $u.c$ after marking new inactive nodes. For instance, consider a node $u$ with children $\hat{v}$ and $\bar{v}$.
By definition, it holds that $u(i_1,i_2,\ldots,i_{d_c}) = \hat{v}(i_1,i_2,\ldots,i_{d_c}) + \bar{v}(i_1,i_2,\ldots,i_{d_c})$. If $\hat{v}$ becomes inactive, then $u(i_1,i_2,\ldots,i_{d_c}) = \bar{v}(i_1,i_2,\ldots,i_{d_c})$.

For a ball $\ball(x,r)$, let $\tree.\widebar{\inactive}(x,r)$ be defined similarly to $\tree.\inactive(x,r)$, except that it also updates all values $u(i_1,i_2,\ldots,i_{d_c})$ as described above. 
Furthermore, let $\tree.\widebar{\coun}(x,r)$ denote an oracle analogous to $\tree.\coun(x,r)$, with the only difference being that it returns $k^{d_c}$ counters. 
Specifically, for every $(i_1,\ldots,i_{d_c}) \in [k]^{d_c}$, it returns a counter $\tau(i_1,\ldots,i_{d_c}) = \sum_{u \in \mathcal{U}(x,r)} u(i_1,\ldots,i_{d_c})$.

\paragraph{Algorithm}
After defining the new oracles $\tree.\widebar{\inactive}(x,r)$ and $\tree.\widebar{\coun}(x,r)$, we proceed with the main algorithm for setting the weights. 
Initially, we set $\bar{\omega}(c) = 0$ for every $c \in C$.

As in the first step of the algorithm from Section~\ref{sec:kmeans}, we compute a value $L$ such that $\frac{L^2}{9} \leq \kmeans_{\opt(\bar{\Q}(\I))}(\bar{\Q}(\I)) \leq \kmeans_{C'}(\bar{\Q}(\I)) \leq |\Q(\I)|^3 \cdot L^2$.

For every $i \in \{0,1,\ldots,O(\frac{1}{\eps}\log N)\}$, we set $r = (1+\eps)^i \frac{L}{9|\Q(\I)|}$ and repeat the following steps for each $c' \in C'$. 
We call the function $\tree.\widebar{\coun}(c',r)$, and for every $(i_1,\ldots,i_{d_c}) \in [k]^{d_c}$, we update $\bar{\omega}(p) = \bar{\omega}(p) + \tau(i_1,\ldots,i_{d_c})$, for every $p \in C$ such that $\pi_{\allattr_{(n)}}(p) = c'$ and, for every $A^{(j)} \in \allattr_{(c)}$, $\pi_{A^{(j)}}(p)$ is the $i_j$-th center in $C_{A^{(j)}}$. 
We then call $\tree.\widebar{\inactive}(c',r)$.

\paragraph{Analysis}
Let $\assigfunc:\Q(\I)\rightarrow C$ be the assignment function that maps each tuple in $\Q(\I)$ to a center in $C$ according to our algorithm. In particular, when a tuple $t \in \Q(\I)$ contributes to $\bar{\omega}(p)$ for a unique $p \in C$, we set $\assigfunc(t) = p$. By definition, $\bar{\omega}(p) = |\{t \in \Q(\I) \mid \assigfunc(t) = p\}|$. 

%It is straightforward to see that every tuple $t \in \Q(\I)$ is assigned to some center in $C$. 
Since $\kmeans_{C'}(\bar{\Q}(\I)) \leq |\Q(\I)|^3 L^2$, we have $\dist(\pi_{\allattr_{(n)}}(t), C') \leq |\Q(\I)|^{3/2}L$, and the maximum radius considered for $r$ is greater than $|\Q(\I)|^{3/2}L$. Hence, every tuple $t \in \Q(\I)$ is assigned to at least one center in $C$.

\begin{lemma}
$\sum_{t \in \Q(\I)} \dist^2(t, \assigfunc(t)) \leq (1+\eps)\kmeans_{\opt(\Q(\I))}(\Q(\I))$.
\end{lemma}

\begin{proof}
Let $t \in \Q(\I)$ be a tuple assigned to a center in $C$ for some radius $r > \frac{L}{9|\Q(\I)|}$, and let $\Q_1(\I) \subseteq \Q(\I)$ be the set of such tuples. By the correctness of the RBBD-tree, $\pi_{\allattr_{(n)}}(t) \notin \bigcup_{c' \in C'} \mathcal{B}(c', r/(1+\eps))$, implying $\dist(\pi_{\allattr_{(n)}}(t), \pi_{\allattr_{(n)}}(\assigfunc(t))) \leq (1+\eps)^2 \dist(\pi_{\allattr_{(n)}}(t), C')$. Moreover, by construction, for every $A \in \allattr_{(c)}$, $\pi_A(\assigfunc(t)) \in C_A$ is the closest center of $\pi_A(t)$ in $C_A$. Thus, $\dist(t, \assigfunc(t)) \leq (1+\eps)^2 \dist(t, C)$.

Now let $t \in \Q(\I)$ be a tuple assigned to a center in $C$ for $r = \frac{L}{9|\Q(\I)|}$, and let $\Q_2(\I) \subseteq \Q(\I)$ denote these tuples. Then, we have $\dist(t, \assigfunc(t)) \leq (1+\eps)\frac{L}{9|\Q(\I)|}$.

Combining the two cases, we have 
$\sum_{t \in \Q(\I)} \dist^2(t, \assigfunc(t)) = \sum_{t \in \Q_1(\I)} \dist^2(t, \assigfunc(t)) + \sum_{t \in \Q_2(\I)} \dist^2(t, \assigfunc(t)) \leq (1+\eps)^4 \sum_{t \in \Q(\I)} \dist^2(t, C) + (1+\eps)^2 \frac{L^2}{81|\Q(\I)|} \leq (1+2\eps)^4 \kmeans_{\opt(\Q(\I))}(\Q(\I))$. 
Setting $\eps \leftarrow \eps/12$ yields the desired bound $\sum_{t \in \Q(\I)} \dist^2(t, \assigfunc(t)) \leq (1+\eps)\kmeans_{\opt(\Q(\I))}(\Q(\I))$.
\end{proof}

The overall running time of the algorithm is $\O(k^{d_c+1}N)$, since it executes $\O(k)$ oracle calls, each running in $\O(k^{d_c}N)$ time to compute and update the counters $u(i_1,\ldots,i_{d_c})$.

%\section{Proof of Theorem~\ref{thm:kmeansCATopt}}
%\label{appndx:proofCat}

\end{envrevone}

%% file: categorical.tex
%\vspace{-1em}
\begin{envrevone}
\section{Relational clustering on categorical and numerical attributes}
\label{sec:cat}
\newcommand{\adom}{\mathsf{adom}}

In this section, we show how our tree data structure can be integrated with the framework of~\cite{curtin2020rk} to obtain the fastest known algorithms for relational clustering in the presence of both numerical and categorical attributes.
%Up to this point, we have assumed that the domain of every attribute is numerical—either real or integer—formally, for every $A \in \allattr$, $\dom(A) = \Re \cup \mathbb{Z}$.
%We now generalize this setting. 
Let $\allattr_{(c)} \subset \allattr$ denote the set of categorical attributes, and let $\allattr_{(n)} = \allattr \setminus \allattr_{(c)}$ denote the set of numerical attributes, with $d_c = |\allattr_{(c)}|$ and $d_n = |\allattr_{(n)}|$.
For simplicity, we assume that each categorical attribute in $\allattr_{(c)}$ can take at most $z = O(N)$ distinct values.
In the next subsections, we first show how our clustering algorithms extend to join-project queries under bag semantics (ignoring categorical attributes from the clustering objective) and then we combine our results with the framework from~\cite{curtin2020rk} to derive a faster algorithm (than~\cite{curtin2020rk})  considering both the categorical and numerical attributes in the objective.

\vspace{-0.8em}
\subsection{Projecting categorical attributes away: join--project queries under bag semantics}
\label{subsec:bagsem}
\vspace{-0.3em}

Consider a scenario where the database schema includes categorical attributes $\allattr_{(c)}$, but the objective is to cluster the join results only with respect to their numerical attributes. 
Formally, for $X \subseteq \allattr$ and $Y \!\!\subseteq \!\!\Q(\I)$, let $\bar{\pi}_{X}(Y)$ be the multiset
$\bar{\pi}_X(Y)\!\! =\!\!\{ \pi_X(t) \!\mid\! t \!\in \!Y \}$.
The only difference from the standard projection operator is that $\bar{\pi}$ preserves duplicates (this is equivalent to the project operator in \emph{bag semantics}). 
Given a join query $\Q$ over a database instance $\I$, the goal is to cluster
$\bar{\Q}(\I) = \bar{\pi}_{\allattr_{(n)}}(\Q(\I))$,
that is, the projection of $\Q(\I)$ onto the numerical attributes $\allattr_{(n)}$ under bag semantics. 
%This is equivalent to clustering the result of a join--project query under bag semantics.

All our algorithms from Theorems~\ref{thm:kCenter} and~\ref{thm:kmeansconstant} naturally extend to this setting with the same approximation and time guarantees. 
The RBBD-tree is now defined only on the $d_n$ numerical attributes $\allattr_{(n)}$, and for every box (or box with a hole) $\rho$, all operations from Lemma~\ref{lem:Rects} can be performed on $\bar{\Q}(\I) \cap \rho$ in $O(N)$ time. 
For instance, sampling a tuple uniformly at random from $ \bar{\Q}(\I)\cap \rho$ is equivalent to sampling uniformly at random a tuple from
$\{ t \in \Q(\I) \mid \pi_{\allattr_{(n)}}(t) \in \rho \}$.
More generally, all RBBD-tree operations used by our algorithms, such as counting, reporting, or sampling tuples within a box, on the result of a projected join query $\bar{\Q}$ under bag semantics can be straightforwardly implemented on the corresponding join query $\Q$ under set semantics. 
Therefore, all our algorithms can be executed on join--project queries under bag semantics while preserving the same approximation and time guarantees.

\vspace{-0.2em}
\paragraph{Remark}
While our algorithms do not generally extend to join-project queries under set semantics, they remain applicable in cases where tuple multiplicities do not affect the objective value. Notably, this includes the relational $k$-center clustering problem (Section~\ref{sec:kCenter}) and the Relational Remote Edge problem (Appendix~\ref{appndx:DivFairProbs}), for which join-project queries under set and bag semantics produce equivalent instances.

\vspace{-0.5em}
\subsection{Clustering with both numerical and categorical attributes} \label{subsec:cat-domain}
\vspace{-0.2em}

We focus on the $k$-means clustering problem as in~\cite{curtin2020rk}; however, all our results can be extended to relational $k$-center and $k$-median clustering as well. 
In machine learning, a common approach for handling categorical variables is to apply \emph{one-hot encoding}, which converts each categorical variable into a binary representation. 
For example, consider a categorical attribute $A \in \allattr$ that can take $z$ distinct values. 
Assume an arbitrary but fixed ordering of these values. 
Then, the $j$-th value of $A$ can be represented by a binary vector with a $1$ in the $j$-th position and $0$ elsewhere.

\vspace{-0.2em}
\paragraph{Expensive algorithm}
In the relational setting we can replace each categorical attribute $A \in \allattr_{(c)}$ with $z$ binary attributes $A_{1}^{(\mathrm{bin})}, \ldots, A_{z}^{(\mathrm{bin})}$. 
Each tuple $t$ in a relation $R$ containing attribute $A$ is replaced by a new tuple $t'$ such that if $\pi_A(t)$ corresponds to the $j$-th value of $A$, then $\pi_{A_j^{(\mathrm{bin})}}(t') = 1$ and $\pi_{A_h^{(\mathrm{bin})}}(t') = 0$ for every $h \in [z] \setminus \{j\}$. 
After this transformation, the join query $\Q$ involves $m$ relations and a total of $d_n + d_c \cdot z$ numerical attributes.
Thus, the distance function $\dist(\cdot, \cdot)$ remains a Euclidean metric over numerical attributes. 
By Theorem~\ref{thm:kmeansconstant}, we obtain a $5\gamma$-approximation algorithm (for $\eps = 1$) for the relational $k$-means problem involving both numerical and categorical attributes, with a data complexity of
$\O(2^{z} k N + \timeMeans_\gamma(2^{z} k))$.
When $z = O(1)$, this matches the running time established in Theorem~\ref{thm:kmeansconstant}. 
However, in practice, $z$ may be large. For instance, a categorical attribute may represent countries, colors, or product categories making the above approach impractical due to its exponential dependence on $z$. 
Next, we show an alternative algorithm for relational $k$-means clustering that avoids explicit one-hot encoding and whose running time does not depend exponentially on $z$.

\paragraph{Known algorithm~\cite{curtin2020rk}}
First, all dangling tuples are removed from the database. 
For every attribute $A \in \allattr$, let $S_A$ denote the set of all values of $A$ that appear in the database instance $\I$; that is, $S_A = \adom_\I(A)$.
%, where $\adom_\I(A)$ is the active domain of $A$ in $\I$. 
For each value $x \in S_A$, they compute the weight 
$w(x) = \sum_{t \in \Q(\I)} 1[\pi_A(t) = x]$,
representing the number of join results whose projection on attribute $A$ equals $x$. 
Using Yannakakis' algorithm, all weights $w(x)$ over all values $x\in\adom(A)$ and attributes $A\in \allattr$ can be computed in $\O(N)$ time.
Then, for each attribute $A \in \allattr$, two cases are distinguished. 
If $A \in \allattr_{(c)}$, they apply a simple and efficient procedure to compute the optimal $k$-means clustering of the weighted categorical set $S_A$ in $\O(N)$ time. 
If $A \in \allattr_{(n)}$, they invoke~\cite{gronlund2017fast, wu1991optimal, fleischer2006online} to compute the optimal $k$-means clustering of the weighted numerical set $S_A$ in $\O(kN)$ time. 
In either case, let $C_A$ denote the resulting set of $k$ centers for attribute $A$.
After computing $C_A$ for all attributes $A \in \allattr$, they construct the grid 
$C = \bigtimes_{A \in \allattr} C_A$,
which contains all possible combinations of centers across attributes. 
Note that $|C| = O(k^d)$, where $d = |\allattr|$ is the total number of attributes. 
For each center $p \in C$, they compute a weight $\omega(p)$, defined as the number of tuples in $\Q(\I)$ that have $p$ as their closest center. 
Using the techniques of Lemma~\ref{lem:Rects}, each weight $\omega(p)$ can be computed in $\O(N)$ time.
It is shown that the weighted set $C$ 
(with the weight function $\omega$) 
defines a coreset for the relational $k$-means clustering. 
Running a $\gamma$-approximation algorithm for $k$-means 
%(assuming both numerical and categorical attributes)
on the weighted set $C$ yields a set of $k$ centers $\ret$ such that
$\kmeans_\ret(\Q(\I)) \leq (2\sqrt{\gamma} + 1)^2 \, \kmeans_{\opt(\Q(\I))}(\Q(\I))$.
The running time of the algorithm is 
$\O(k^d N + \timeMeans_\gamma(k^d))$.

\paragraph{Our algorithm}
We design a faster algorithm by leveraging our new results to handle the numerical attributes more efficiently. 
We begin by processing the categorical attributes following the approach of~\cite{curtin2020rk}: 
For each $A \in \allattr_{(c)}$, we construct $S_A$ and compute the weight $w(x)$ for every $x \in S_A$ in the same manner as~\cite{curtin2020rk}. 
We then apply the straightforward algorithm for optimal $k$-means clustering on the weighted set $\adom(A)$ and let $C_A$ denote the optimal set of centers.
Next, we process the numerical attributes $\allattr_{(n)}$. 
Unlike~\cite{curtin2020rk}, we process all numerical attributes simultaneously. 
Specifically, we execute the relational $k$-means algorithm from Theorem~\ref{thm:kmeansconstant} while skipping the categorical attributes, i.e., we run the algorithm on 
$\bar{q}(\I) = \bar{\pi}_{\allattr_{(n)}}(\Q(\I))$,
as described in Appendix~\ref{subsec:bagsem}. 
Let $C'$ be the resulting set of $k$ centers. 
It is important to note that we do \emph{not} remove duplicates in $\bar{q}(\I)$, since our algorithm inherently accounts for the correct weighting of the tuples. 
In particular, running $k$-means clustering directly on $\bar{q}(\I)$
%(without explicitly computing it) 
is equivalent to running $k$-means clustering on the weighted set $\pi_{\allattr_{(n)}}(\Q(\I))$ with weights
$w(p) = \sum_{t \in \Q(\I)} 1[\pi_{\allattr_{(n)}}(t) = p]$.

We then construct the set of candidate centers as 
$C = C' \times \left( \bigtimes_{A \in \allattr_{(c)}} C_A \right)$.
Unfortunately, unlike~\cite{curtin2020rk}, we cannot directly compute, for each $p \in C$, the exact number of tuples in $\Q(\I)$ that have $p$ as their closest center in $C$. 
Instead, we use our ideas from the RBBD-tree construction in the relational setting to approximately count the number of tuples that are closest to each center $p \in C$. 
%The idea is closely related to the technique used in our algorithm from Section~\ref{sec:kmeans}.
%More precisely, 
For every $p \in C$, we compute an approximate weight
$\bar{\omega}(p) = |\{ t \in \Q(\I) \mid \assigfunc(t) = p \}|$,
where $\assigfunc : \Q(\I) \rightarrow C$ is a function assigning each tuple $t \in \Q(\I)$ to a center in $C$, satisfying
$\sum_{t \in \Q(\I)} \dist^2(t, \assigfunc(t)) 
    \leq 2 \, \kmeans_{\opt(\Q(\I))}(\Q(\I))$.
This is equivalent to the process of first clustering the tuples in $\bar{\Q}(\I)$ according to Theorem~\ref{thm:kmeansconstant} with an approximation factor of $2 \cdot 5\gamma$ (instead of $5\gamma$), and subsequently computing the true weights $\omega(p)$ (recall that $\omega(p)$ is the number of tuples in $\Q(\I)$ that have $p$ as their closest center). 
We show the details of computing the weight $\bar{\omega}(\cdot)$ in Appendix~\ref{appndx:cat}. 
Finally, we run a $k$-means clustering algorithm 
%(supporting both numerical and categorical attributes)
on the weighted set $C$, obtaining a set $\ret$ of $k$ final centers.

\begin{theorem}
\label{thm:kmeansCATopt}
            Given an acyclic join query $\Q$ with $d_n$ numerical and $d_c$ categorical attributes, a database instance $\I$ of size $N$, and a parameter $k$, there exists an algorithm that computes a set $\ret\subseteq \Q(\I)$ of $k$ points such that
            $\kmeans_{\ret}(\Q(\I))\leq  \left((\sqrt{10}+1)\sqrt{\gamma}+\sqrt{10}\gamma\right)^2 \!\!\kmeans_{\opt(\Q(\I))}(\Q(\I))$ in $\O(k^{d_c+1}N+\timeMeans_\gamma(k^{d_c+1}))$ time, with probability at least $1-\frac{1}{N}$.
 
        %    , with probability at least $1-\frac{1}{N}$. The running time of the algorithm is $\O(k^{d_c+1}N+\timeMeans_\gamma(k^{d_c+1}))$
           % , with probability at least $1-\frac{1}{N}$.
\end{theorem}
\begin{proof}
%[Proof of Theorem~\ref{thm:kmeansCATopt}]
    From Appendix~\ref{appndx:cat}, we have that 
$\sum_{t \in \Q(\I)} \dist^2(t, \assigfunc(t)) \leq 2 \, \kmeans_{\opt(\Q(\I))}(\Q(\I))$.
This is equivalent to the process of first clustering $\bar{\Q}(\I)$ according to Theorem~\ref{thm:kmeansconstant} with an approximation error of $2 \cdot 5\gamma$ (instead of $5\gamma$) and then computing the exact weights $\omega(p)$. 
Following the analysis in~\cite{curtin2020rk}, the overall approximation factor becomes 
$((\sqrt{10} + 1)\sqrt{\gamma} + \sqrt{10}\gamma)^2$.

For each categorical attribute $A \in \allattr_{(c)}$, the set of centers $C_A$ is computed in $\O(N)$ time, as discussed before. 
The algorithm from Theorem~\ref{thm:kmeansconstant}, when applied to $\bar{\pi}_{\allattr_{(n)}}(\Q(\I))$, runs in $\O(kN + \timeMeans_\gamma(k))$ time with high probability. 
Since $|C| = O(k^{d_c + 1})$, all weights $\bar{\omega}(p)$ for $p \in C$ is computed in $\O(k^{d_c + 1})$ time (see Appendix~\ref{appndx:cat}). 
Finally, the $k$-means clustering on the weighted set $C$ is executed in $\O(\timeMeans_\gamma(k^{d_c + 1}))$ time.
Overall, the total running time of our algorithm is
$\O(k^{d_c + 1}N + \timeMeans_\gamma(k^{d_c + 1}))$, with high probability.
\end{proof}
\end{envrevone}

%% file: apnndxGonzalez.tex
\section{Extensions to other optimization problems}
\label{appndx:gonzalez}

Using our machinery from the previous sections, we show how we can implement efficient algorithms for various optimization problems in diversity and fairness. First, we show how to implement an oracle that computes (approximately) the farthest point in $\Q(\I)$ from a given set of points $S$. Then, we show an $\O(kN)$ time approximate implementation of the well-known Gonzalez algorithm~\cite{gonzalez1985clustering} in the relational setting. Gonzalez algorithm is used to solve various problems in diversity and fairness.
%Finally, we show how the RBBD tree can be used to solve the set cover problem in the relational setting.

We first start with some definitions.

\paragraph{Gonzalez algorithm}
Gonzalez algorithm~\cite{gonzalez1985clustering} is a well-known and powerful algorithm that is used to derive constant approximation algorithms for multiple combinatorial optimization problems such as clustering and diversity problems~\cite{ravi1994heuristic, tamir1991obnoxious}. The algorithm is quite simple in the standard computational setting: Let $P$ be a set of $n$ points in $\Re^d$ and let $C = \emptyset$ be an initially empty set to be returned. Gonzalez algorithm runs in $k$ iterations. In each iteration, we find the point $\argmax_{p\in P\setminus C}\dist(p,C)$ and add it to the set $C$. A naive implementation of the Gonzalez algorithm runs in $O(nk)$ time. In the Euclidean space, the running time can be improved to $O(n\log k)$.
%an approximate implementation of the Gonzalez algorithm (where in each iteration a point $p\in P$ is found such that $\dist(p,C)\geq (1-\eps)\max_{p\in P\setminus C}\dist(p,C)$) runs in $O(n\log n)$ time.
%As mentioned in~\cite{}, the Gonzalez algorithm cannot be  implemented on relational data since it is equivalent to an instance of $\mathsf{FAQ-AI}(k-1)$~\cite{}.

\begin{definition}
    [Relational implementation of Gonzalez algorithm]
    Given a database instance $\I$, a join query $\Q$, and a positive parameter $k$, the goal is to implement the Gonzalez algorithm for $k$ iterations. Namely, given a set $\ret$ with $|\ret|<k$, compute a point $p\in \Q(\I)$ such that $\dist(p,\ret)=\max_{q\in \Q(\I)\setminus \ret}\dist(q,\ret)$.
\end{definition}
%\aryan{I think this is a little bit confusing. Maybe it is better to write it as given a set, we only want to find the farthest point from it. Later, we can explain that having this oracle (which gives the farthest point outside of a set), it is straightforward to implement the Gonzalez algorithm in the relational settings.}
We also define the \emph{$\eps$-approximate implementation of the Gonzalez algorithm}. In every iteration of a relational $\eps$-approximate implementation of the Gonzalez algorithm, a point $p\in \Q(\I)\setminus \ret$ is added in $\ret$ such that $\dist(p,\ret)\geq (1-\eps)\max_{q\in \Q(\I)\setminus \ret}\dist(q,\ret)$.
%\aryan{I think again if we define it as an oracle that finds an approximation of the farthest point from a set, it would be better.}

\input{farthest}
\input{gonzalez}

%% file: farthest.tex
\subsection{Compute the farthest point}
%Using the machinery from the previous sections, we show how we can run an approximate version of the Gonzalez algorithm. We focus on one iteration of the Gonzalez algorithm in the relational setting. 
Let $S\subseteq \Q(\I)$ be a subset of $\Q(\I)$ with $|S|<k$. The goal is to compute a point $p\in\Q(\I)\setminus S$ such that $\dist(p,S)=\max_{t\in \Q(\I)\setminus S}\dist(t,S)$. We propose an algorithm that runs in $\O(kN)$ time to solve this problem approximately, i.e., we compute a point $p\in\Q(\I)\setminus S$ such that $\dist(p,S)\geq (1-\eps)\max_{t\in \Q(\I)\setminus S}\dist(t,S)$, for an arbitrarily small constant $\eps$.
%This procedure is executed $O(k)$ times in during the Gonzalez algorithm.
%Thus, we describe an $\O(k^2N)$ implementation of an $\eps$-approximate version of Gonzalez algorithm in the relational setting.

\paragraph{High level idea}
The high-level idea is to guess $\max_{t\in \Q(\I)}\dist(t,S)$ running a binary search similar to the binary search in Section~\ref{sec:kCenter}. For every distance $r$ we guess, we check whether the balls $\bigcup_{s\in S}B(s,r)$ cover all points in $\Q(\I)$, using the RBBD data structure. If yes, then we continue the binary search for smaller values of $r$. If not, i.e., there are still active nodes in the RBBD tree, we continue the binary search for larger values of $r$. %In the end of this procedure we will end up with a distance $r'$ such that $\bigcup_{s\in S}B(s,(1+\eps)r')$ covers all points in $\Q(\I)$ while $\bigcup_{s\in S}B(s,r')$ does not cover all points in $\Q(\I)$.
From the last iteration where $\bigcup_{s\in S}B(s,r)$ does not cover all points in $\Q(\I)$, 
we add in $S$ any point from $\Q(\I)$ that does not have an inactive ancestor in the RBBD tree. The distance from this point to $S$ will be a good approximation of the maximum distance of any point in $\Q(\I)$ from $S$.

\paragraph{Algorithm}
Let $S\subseteq \Q(\I)$.
We construct the root node $\root$ of an RBBD tree $\tree$ over $\Q(\I)$ as described in Section~\ref{subsec:RBBDrel}.
%For every constructed node $u$ of $\tree$ we store $u.a$ which is $1$ if $u$ is active, and $0$ otherwise. We also store $u.rep$ which is a representative active point $u.rep\in \square_u\cap \Q(\I)$. A point is active if it lies in a leaf node of the RBBD tree over $\Q(I)$ and does not have any inactive ancestor node.

We run the binary search as in Section~\ref{sec:kCenter}.
The only difference is that if $\hat{r}$ is the $z$-th smallest distance returned by~\cite{salowe1989infinity}, we run another binary search in the range $[\hat{r}, \sqrt{d}\hat{r}]$. Notice that for any pair of points $t,q\in \Re^d$, $||t-q||_\infty\leq ||t-q||_2\leq \sqrt{d}||t-q||_\infty$. Hence, if we discretize $[\hat{r}, \sqrt{d}\hat{r}]$ by a relative step $(1+\eps)$, i.e., $Y_{\hat{r}}=\{\hat{r}(1+\eps)^0, \hat{r}(1+\eps)^1, \hat{r}(1+\eps)^2,\ldots, \hat{r}(1+\eps)^{\log(\sqrt{d})/\eps}\}$ and we run a binary search on $Y_{\hat{r}}$ in the end the binary search will end up at a distance $r^*$ where the algorithm satisfies a property at distance $r^*$ but does not satisfy the same property at distance $(1+\eps)r^*$. This is related to the binary search we run in Appendix~\ref{apndx:simple2approx}.

Let $r$ be one of the distances returned by the binary search.
For every point $s\in S$, we run $\tree.\inactive(s,r)$.
After repeating the procedure above for every $s\in S$, we call $\tree.\representative()$. If it returns $\emptyset$ we continue the binary search with smaller values of $r$, otherwise we continue the binary search with larger values of $r$.
Before we continue in the next iteration of the binary tree we modify $\tree$ making all its nodes active again and setting the initial representative points on their nodes.
In the end of this binary search, let $r'$ be the last (largest) distance we found where
$p=\tree.\representative()\neq \emptyset$. We return $p$.

\paragraph{Correctness}
\begin{lemma}
$\dist(p,S)\geq \frac{1}{(1+\eps)^2}\max_{t\in \Q(\I)\setminus S}\dist(t,S)$.%, with probability at least $1-\frac{1}{N}$.
\end{lemma}
\begin{proof}
The algorithm executes $\O(k)$ oracles from Section~\ref{sec:reloracles}. %Since each oracle is correct with probability at least $1-\frac{1}{N^{m+2}}$, the algorithm is correct with probability at least $1-\frac{1}{N}$.
%The query procedure and the update operations (on representative points and active nodes) of the RBBD tree is correct from our discussion in the previous sections.
If $r'$ was the largest distance found by the binary search such that $\tree.\representative\neq\emptyset$, we have that $\dist(p,S)>r'$.

Recall, that the binary search (as defined in Appendix~\ref{appndx:kCenter}) is executed on the $\ell_\infty$ pairwise distances of $\Q(\I)$ and then discretize the ranges $[r,\sqrt{d}r]$ by a factor $(1+\eps)$. By definition, in the execution of binary search, for $r''=(1+\eps)r'$ we had that $\tree.\representative()=\emptyset$, otherwise $r'$ would not be the largest distance we found with $\tree.\representative()\neq\emptyset$. By the properties of RBBD tree from Section~\ref{sec:RBBD}, we have that
every point $q\in\Q(\I)$ is inactive, i.e., there exists $s\in S$ such that $q\in \bigcup_{u\in\canonical(s,r'')}(\square_u\cap \Q(\I))$. Since the RBBD tree approximates the balls within a factor $(1+\eps)$, we have that for every $q\in \Q(\I)$, $\dist(q,S)\leq (1+\eps)r''$. Hence, $\max_{t\in \Q(\I)\setminus S}\dist(t,S)\leq (1+\eps)r''$.
We conclude that $\dist(p,S)>r'=\frac{r''}{1+\eps}\geq \frac{1}{(1+\eps)^2}\max_{t\in \Q(\I)\setminus S}\dist(t,S)$.
\end{proof}

Notice that $\frac{1}{(1+\eps/2)^2}\geq 1-\eps$, so if we set $\eps\leftarrow\eps/2$, we get $\dist(p,S)\geq (1-\eps)\max_{t\in \Q(\I)\setminus S}\dist(t,S)$.%, with high probability.

\begin{lemma}
The running time of the algorithm is $O\left(k(\eps^{-d+1}+\log N)(N+\eps^{-2}\log^2 N)\log N\right)$, with probability at least $1-\frac{1}{N}$.
\end{lemma}
\begin{proof}
    The binary search runs in $O(\log N)$ iterations. In order to compute a distance in each step of the binary search we spend $O(N\log N)$ time as shown in Appendix~\ref{appndx:kCenter}. Hence all values for the binary search is computed in $O(N\log^2 N)$ time. For each step of the binary search, the $\tree.\inactive(s,r)$ oracle runs in $O((\eps^{-d+1}+\log N)(N+\eps^{-2}\log^2 N))$ with probability at least $1-\frac{1}{N^{m+2}}$ and the $\tree.\representative()$ runs in $O(1)$ time.
    Hence, the overall running time of our algorithms is $O(N\log^2 N+k(\eps^{-d+1}+\log N)(N+\eps^{-2}\log^2 N)\log N)=O(k(\eps^{-d+1}+\log N)(N+\eps^{-2}\log^2 N)\log N)$, with probability at least $1-\frac{1}{N}$.
\end{proof}

\begin{theorem}
\label{thm:iterGonz}
   Given an acyclic query $\Q$, a database instance $\I$ of size $N$, an arbitrarily small constant parameter $\eps$, and a set $S\subseteq \Q(\I)$ with $|S|<k$, a point $p\in\Q(\I)\setminus S$ can be computed such that $\dist(p,S)\geq (1-\eps)\max_{t\in \Q(\I)\setminus S}\dist(t,S)$ in $O(kN\log^2 N + k\log^4 N)$ time, with probability at least $1-\frac{1}{N}$.
\end{theorem}

Starting from $S=\emptyset$ and repeating the algorithm from Theorem~\ref{thm:iterGonz} $k$ times, we get an $\eps$-approximate implementation of the Gonzalez algorithm in the relational setting in $O(k^2N\log^2 N +k^2\log^4 N)$ time.
Next, we show an even faster implementation of the Gonzalez algorithm using our relational $k$-center clustering algorithm from Theorem~\ref{thm:kCenter}.

%\begin{corollary}
%Given an acyclic query $\Q$, a database instance $\I$ of size $N$, and an arbitrarily small constant parameter $\eps$, there exists an $\eps$-approximate version of the Gonzalez algorithm on $\Q(\I)$ in $O(k^2N\log^2 N +k^2\log^4 N)$ time, with probability at least $1-\frac{1}{N}$.
%\end{corollary}

%\paragraph{Discussion}

%% file: gonzalez.tex
\subsection{Implementation of Gonzalez Algorithm in the Relational Setting}

%Using our machinery from the previous sections, we show how we can implement efficient algorithms for various optimization problems in diversity, fairness, and coverage.

We propose an $\O(kN)$ execution of an $\eps$-approximate implementation of the Gonzalez algorithm. Such an implementation will allow us to derive faster algorithms for multiple diversity and fairness problems in the relational setting.

In the standard computational setting, the next argument is known~\cite{kurkure2024faster}. Let $P\subset \Re^d$ and $S\subseteq P$ be a set of $k'=O(k\eps^{-d})$ centers returned by a $(2+\eps)$-approximation algorithm for the $k'$-center clustering problem. An implementation of the Gonzalez algorithm on $S$ for $k$ steps returns an $\eps$-approximate implementation of the Gonzalez algorithm on $P$.

In the relational setting, given a parameter $k$, we run our algorithm from Theorem~\ref{thm:kCenter} using the parameter $k'=O(k\eps^{-d})$ (instead of $k$). Let $S\subseteq \Q(\I)$ be the set of centers returned by our algorithm with $|S|=O(k\eps^{-d})$ such that $\kcenter_{S}(\Q(\I))\leq (1+\eps)\kcenter_{\opt(\Q(\I))}(\Q(\I))$, with probability at least $1-\frac{1}{N}$. Finally, we run the geometric implementation of Gonzalez algorithm~\cite{feder1988optimal} on $S$ for $k$ iterations and let $\ret$ be the returned set. From~\cite{feder1988optimal} and the correctness of Theorem~\ref{thm:kCenter}, we get that $\ret$ is the result of an $\eps$-implementation of the Gonzalez algorithm on $\Q(\I)$.
%, with probability at least $1-\frac{1}{N}$.
The running time to get $S$ is $O(k\eps^{-d}n\log^2 N + k\eps^{-d}\log^4N)$ with probability at least $1-\frac{1}{N}$ (Theorem~\ref{thm:kCenter} for $k'=O(k\eps^{-d})$ instead of $k$). Furthermore, the geometric implementation of the Gonzalez algorithm~\cite{feder1988optimal} on $S$ runs in $O(k\eps^{-d}\log k)$ time. Overall, the running time of the $\eps$-implementation of the Gonzalez algorithm is $O(k\eps^{-d}n\log^2 N + k\eps^{-d}\log^4N)$ with probability at least $1-\frac{1}{N}$.
%Using the same argument, the fastest previously known $\eps$-approximate implementation of the Gonzalez algorithm in the relational setting used the relational $k$-center algorithm from~\cite{}, leading to a $\O(k^2N+k^4)$ time algorithm.
%We conclude to the next theorem.
\begin{theorem}
\label{thm:gonzalez}
Given an acyclic query $\Q$, a database instance $\I$ of size $N$, and an arbitrarily small constant parameter $\eps\in(0,1)$, there exists an $\eps$-approximate version of the Gonzalez algorithm on $\Q(\I)$ for $k$ steps in $O(kN\log^2 N +k\log^4 N)$ time. The %correctness and the 
running time holds with probability at least $1-\frac{1}{N}$.
\end{theorem}

\subsection{Diversity and fairness problems}
\label{appndx:DivFairProbs}
\paragraph{Diversity}
It is known that the the Gonzalez algorithm can be used to construct small coresets for multiple problems in diversity and fairness~\cite{indyk2014composable}.
More specifically, given a database instance $\I$ of size $|\I|=N$, a join query $\Q$, and a positive integer parameter $k$, the goal is to compute a set $\ret\subseteq \Q(\I)$ of size $|\ret|=k$ such that,
\begin{itemize}
    \item $\min_{p,q\in\ret}\dist(p,q)$ is maximized (Relational Remote Edge -- RRE).
    \item the weight (with respect to the Euclidean distance) of the minimum spanning tree on $\ret$ is maximized (Relational Remote Tree -- RRT).
    \item the minimum cost (with respect to the Euclidean distance) of the travelling salesman problem on $\ret$ is maximized (Relational Remote Cycle -- RRC).
    \item $\sum_{p\in \ret}\min_{q\in\ret\setminus\{p\}}\dist(p,q)$ is maximized (Relational Remote Pseudoforest -- RRP).
    \item The minimum cost (with respect to the Euclidean distance) of a perfect matching of $\ret$ is maximized (Relational Remote Matching -- RRM).
\end{itemize}

For each of the defined problem, we run our algorithm from Theorem~\ref{thm:gonzalez}. Then using the results from~\cite{indyk2014composable}, we conclude to the next theorem.
\begin{theorem}
\label{thm:diversity}
            Given an acyclic join query $\Q$, a database instance $\I$ of size $N$, and a parameter $k$, there exists an $O(1)$-approximation algorithm for the RRE, RRT, and RRC problem over $\Q(\I)$ in $O(kN\log^2 N +k\log^4 N)$ time. Both the aproximation factor and the running time hold with probability at least $1-\frac{1}{N}$.
            Furthermore, there exists an $O(\log k)$-approximation algorithm for the RRP, and RRM problem over $\Q(\I)$ in $O(kN\log^2 N +k\log^4 N)$ time. %Both the approximation factor and 
            The running times hold with probability at least $1-\frac{1}{N}$.
\end{theorem}

We highlight that the RRE problem has been studied in~\cite{agarwal2024computing}. Agarwal et al.~\cite{agarwal2024computing} designed an $O(1)$-approximation algorithm for the RRE problem over $\Q(\I)$ (for any acyclic join query $\Q$) in $\O(\min\{k^2N,kN+k^{d/2}\})$ time, where $d$ is the total number of attributes, $d=|\allattr|$. Our algorithm improves the running time by a multiplicative factor of $k$ or an additive term $k^{d/2}$.
%To the best of our knowledge, no known algorithms exist we are the first that studied the RRT, RRC, RRP, and RRM problems.

\paragraph{Fairness}
Interestingly, the $\eps$-approximate implementation of the Gonzalez algorithm is useful in fairness problems.
We define two problems, the Relational Fair Remote Edge (RFRE) and the Relational Fair Remote Clique (RFRC) problem. Without loss of generality assume that there exists one relation, say $R_1$, where in addition to the numerical attributes $\allattr_1$, has one additional categorical attribute $A'$ with $\dom(A')=\{c_1, c_2,\ldots, c_g\}$, where $g$ is the number of different values. Let $Q_c=\{t\in \Q(\I)\mid \pi_{A'}(t)=c\}$, i.e., the tuples in $\Q(\I)$ that belong to group $c\in \dom(A')$. 
%We define the RFRE problem as follows: 
Given $g$ positive integer parameters $k_1,\ldots, k_g$ the goal is to compute a set $\ret\subseteq \Q(\I)$ 
such that $|\ret\cap Q_c|\geq k_c$ for every $c\in\dom(A')$ and
\begin{itemize}
    \item $\min_{p,q\in \ret}\dist(p,q)$ is maximized (RFRE).
    \item $\frac{1}{2}\sum_{p,q\in \ret}\dist(p,q)$ is maximized (RFRC).
\end{itemize}

The problem of fair remote edge (or fair max-min as it is also called in the literature) has been studied extensively in the standard computational setting~\cite{kurkure2024faster, addanki2022improved, moumoulidou2021diverse}. It is known that running the Gonzalez algorithm for $O(k\eps^{-d})$ iterations, where $k=k_1+\ldots k_g$, on the element of each group creates a coreset of size $O(gk\eps^{-d})$ for the fair remote edge problem. Then an expensive $\frac{1}{2}$-approximation algorithm over the coreset is executed~\cite{kurkure2024faster, addanki2022improved, moumoulidou2021diverse}.
If our goal is to get an $O(1)$-approximation algorithm then the size of the coreset is $O(gk)$.
Similarly, we can run the $\eps$-approximate implementation from Theorem~\ref{thm:gonzalez} over $Q_c$, for each different $c\in\dom(A')$ to construct a coreset. Using Theorem~\ref{thm:gonzalez} along with the algorithms in~\cite{kurkure2024faster, addanki2022improved, moumoulidou2021diverse}, we get the next result.

\begin{theorem}
\label{thm:RFRE}
            Given an acyclic join query $\Q$, a database instance $\I$ of size $N$, and parameters $k_1,\ldots, k_g$, such that $k=k_1+\ldots+k_g$, there exists a $O(1)$-approximation algorithm for the RFRE problem over $\Q(\I)$ in $O(kN\log^2 N +k\log^4 N + T_{FRE}(gk^2))$ time, where $T_{FRE}$ is the running time of an $O(1)$-approximation algorithm in the standard computational setting for the fair remote edge problem over $gk^2$ points.
            %Both the approximation factor and 
            The running time holds with probability at least $1-\frac{1}{N}$.
\end{theorem}

Next, the problem of fair remote clique (or fair max-sum) has also been studied in the standard computational setting~\cite{mahabadi2023core}. Mahabadi and Trajanovski~\cite{mahabadi2023core} construct a coreset of size $O(\sum_{i\in[g]}k_i^2)=O(gk^2)$ as follows. They run Gonzalez algorithm for $k_i$ iterations on the points of the $i$-th group. Then for every point returned by the Gonzalez algorithm they add in their coreset the $k_i$ nearest neighbors among the points in the $i$-th group.

In the relational setting, for every $c\in\dom(A')$ we first run the $\eps$-approximate implementation of the Gonzalez algorithm from Theorem~\ref{thm:gonzalez} over $Q_c$. Then, for every point $p$ returned by Theorem~\ref{thm:gonzalez} we find the $k_i$ nearest neighbors in $Q_i$ using the Euclidean-based oracle presented in~\cite{agarwal2024computing} in $O(N+k_i\log N)$. Overall, using Theorem~\ref{thm:gonzalez},~\cite{mahabadi2023core}, and~\cite{agarwal2024computing}, we get the next result.

\begin{theorem}
\label{thm:RFRC}
            Given an acyclic join query $\Q$, a database instance $\I$ of size $N$, and parameters $k_1,\ldots, k_g$, such that $k=k_1+\ldots+k_g$, there exists a $O(1)$-approximation algorithm for the RFRC problem over $\Q(\I)$ in $O(kN\log^2 N +k\log^4 N + T_{FRC}(gk^2))$ time, where $T_{FRC}$ is the running time of an $O(1)$-approximation algorithm in the standard computational setting for the fair remote clique problem over $gk^2$ points.
            %Both the approximation factor and 
            The running time holds with probability at least $1-\frac{1}{N}$.
\end{theorem}

%% file: appndxCyclic.tex
\renewcommand{\T}{\mathcal{H}}
\section{Extension to cyclic queries}
\label{appndx:generalQueries}
A fractional edge cover of join query $\Q$ is a point $x = \{x_R \mid R\in \allrel\} \in \mathbb{R}^{m}$ such that for any attribute $A \in \allattr$, $\sum_{R \in \allrel_A} x_R \ge 1$, where $\allrel_A$ are all the relations in $\allrel$ that contain the attribute $A$.
As shown in~\cite{atserias2013size}, the maximum output size of a join query $\Q$ is $O(N^{\lVert x \rVert_1})$. Since the above bound holds for any fractional edge cover, we define $\rho^\star = \rho^\star(\Q)$ to be the fractional cover with the smallest $\ell_1$-norm. Equivalently, $\rho^\star(\Q)$ is the value of the objective function of the optimal solution of linear programming ($\mathsf{LP}$):
    $\min \sum_{R \in \allrel} x_R, \;\text{s.t.}\; \forall R \in \allrel: x_R \ge 0 \;\text{and}\; \forall A \in \allattr: \sum_{R \in \allattr_R} x_R \ge 1$.

Next, we give the definition of the Generalized Hypertree Decomposition (GHD).
A GHD of $\Q$ is a pair $(\T, \lambda)$, where
$\T$ is a tree as an ordered set of nodes and $\lambda: \T \to 2^{\allattr}$ is a labeling function which associates to each vertex $u \in \T$ a subset of attributes in $\allattr$, called $\lambda_u$, such that the following conditions are satisfied:
i) (coverage) For each $R \in \allrel$, there is a node $u \in \T$ such that $\allattr_R\subseteq\lambda_u$, where $\allattr_R$ is the set of attributes contained in $R$;
 ii) (connectivity) For each $A \in \allattr$, the set of nodes $\{u \in \T: A \in \lambda_u\}$ forms a connected subtree of $\T$.

Given a join query $\Q$, one of its GHD $(\T, \lambda)$ and a node $u \in \T$, the width of $u$ is defined as the optimal fractional edge covering number of its derived hypergraph $(\lambda_u, \E_u)$, where $\E_u = 
\{\allattr_R\cap \lambda_u: R \in \allrel\}$. Given a join query and a GHD $(\T, \lambda)$, the width of $(\T, \lambda)$ is defined as the maximum width over all nodes in $\T$. Then, the fractional hypertree width of a join query follows:
The fractional hypertree width of a join query $\Q$, denoted as $\fhw(\Q)$, is
$\fhw(\Q) = \min_{(\T, \lambda)} \max_{u \in \T} \rho^\star(\lambda_u, \E_u)$,
i.e., the minimum width over all GHDs.

\paragraph{Relational algorithms for $k$-center/median/means clustering} Overall, $O(N^\fhw)$ is an upper bound on the number of join results materialized for each node in $\T$. It is also the time complexity to compute the join results for each node in $\T$~\cite{atserias2013size}. Hence, we convert our original cyclic query into an acyclic join query (with join tree $\T$) where each relation has $O(N^\fhw)$ tuples. We execute our algorithms from Theorem~\ref{thm:kCenter} and Theorem~\ref{thm:kmeansopt} (also Theorem~\ref{thm:diversity}, Theorem~\ref{thm:RFRE} and Theorem~\ref{thm:RFRC}) to the new acyclic join query. The approximation factors remain the same. The running times are adjusted by replacing the $N$ factor with an $N^{\fhw}$ factor.